\documentclass[12pt]{article}
\usepackage{verbatim}
\usepackage{amsfonts}
\usepackage{amsmath}
\usepackage{times}
\usepackage{appendix}
\usepackage{graphicx}
\usepackage{color}
\usepackage{enumerate}
\usepackage{fancyhdr,latexsym,amsmath,amsfonts,amssymb,amsbsy,amsthm,url}
\usepackage[margin=0.5in,footskip=0.25in]{geometry}
\usepackage{graphicx,epsfig}
\usepackage{breqn}
\newtheorem{theorem}{Theorem}[]

\newtheorem{remark}{Remark}[]

\usepackage[english]{babel}
\usepackage[dvipsnames]{xcolor}
\usepackage{tikz}
\usepackage{epstopdf}
\usepackage{subfig}
\usepackage{adjustbox}

\makeatletter

\renewcommand{\section}{
	\@startsection
	{section}
	{1}
	{0pt}
	{1.1\baselineskip}
	{0.2\baselineskip}
	{\sc \centering}
}

\renewcommand{\subsection}{
	\@startsection
	{subsection}
	{1}
	{0pt}
	{1.1\baselineskip}
	{0.2\baselineskip}
	{\sc \centering}
}

\renewcommand{\subsubsection}{
	\@startsection
	{subsubsection}
	{1}
	{0pt}
	{1.1\baselineskip}
	{0.2\baselineskip}
	{\sc \centering}
}

\makeatother

\renewcommand{\thesection}{\arabic{section}}

\usepackage[flushleft]{threeparttable}
\usepackage{rotating,booktabs,multirow}
\usepackage{colortbl}
\usepackage{makecell,cellspace,caption}

\usepackage{algorithm}
\usepackage{algorithmic}

\begin{document}
	
\title{\large\sc Strategic Control of Drug-Resistant HIV: Multi-Strain Modeling with Diagnosis, Adherence, and Treatment Switching}
\normalsize
\author{
\sc{Ashish Poonia} \thanks{Department of Mathematics, Indian Institute of Technology Guwahati, Guwahati-781039, India, e-mail: apoonia@iitg.ac.in}
\and 
\sc{Siddhartha P. Chakrabarty} \thanks{Department of Mathematics, Indian Institute of Technology Guwahati, Guwahati-781039, India, e-mail: pratim@iitg.ac.in, Phone: +91-361-2582606}}

\date{}
\maketitle

\begin{abstract}
A central challenge in Human Immunodeficiency Virus (HIV) public health policy lies in determining whether to universally expand treatment access, despite the risk of sub-optimal adherence and consequent drug resistance, or to adopt a more strategic allocation of resources that balances treatment coverage with adherence support. This dilemma is further complicated by the need for timely switching to second-line therapy, which is critical for managing treatment failure but imposes additional burdens on limited healthcare resources. In this study, we develop and analyze a compartmental model of HIV transmission that incorporates both drug-sensitive and drug-resistant strains, diagnosis status, and treatment progression, including switching to second-line therapy upon detection of resistance. Basic reproduction numbers for both strains are derived, and equilibrium analysis reveals the existence of a disease-free state and two endemic states, where the drug-sensitive strain may be eliminated while the drug-resistant strain persists. Local and global sensitivity analyses are performed, using partial rank correlation coefficient (PRCC) and Sobol methods, to identify key parameters influencing different model outcomes. We extend the model using optimal control theory to assess multiple intervention strategies targeting diagnosis, treatment initiation, and adherence. A novel dynamic control framework is proposed to achieve the UNAIDS 95-95-95 targets through efficient resource allocation. Numerical simulations validate the analytical results and compare the effectiveness and cost-efficiency of control strategies. Our findings highlight that long-term HIV epidemic control depends critically on prioritizing adherence-focused interventions alongside efforts to expand first-line treatment coverage.


\end{abstract}

\section{Introduction}
Acquired Immunodeficiency Syndrome (AIDS) has persisted as a significant contributor to global mortality and morbidity, predominantly affecting low and middle-income countries \cite{verm2014}. AIDS is caused by the Human Immunodeficiency Virus (HIV), an RNA retrovirus that primarily targets the immune system of the host and weakens its ability to defend against other opportunistic infections. The World Health Organization (WHO) classifies HIV/AIDS as a global epidemic, with a total of 88.4 million infections and 42.3 million deaths from HIV-related illnesses since the start of the epidemic. In 2023, an estimated 39.9 million individuals were living with HIV worldwide, including 1.3 million new infections, with 0.63 million people having succumbed to AIDS-related complications \cite{fact2024}. These figures highlight the substantial burden that HIV/AIDS continues to put on the public health systems, worldwide.

The progression of HIV infection is typically divided into three stages: the acute phase, the chronic phase, and the final progression to AIDS, which is identified by a significant depletion of CD4+ T-cells \cite{hern2013}. The acute phase, which generally lasts for a few weeks, is characterized by a rapid increase in viral load, making individuals highly infectious during this short window \cite{mill2010}. Following this, the immune system initiates a response, leading to a decline in viral load and a transition to the chronic stage, where individuals often remain asymptomatic. Although the acute phase is highly infectious, most new HIV transmissions occur during the chronic stage, particularly in long-term partnerships \cite{gurs2023}. Without antiretroviral treatment, CD4+ T-cell counts gradually declines from healthy levels (around $800 \text{ to }1000\text{ cells/mm}^{3}$) to below $200\text{ cells/mm}^{3}$, eventually leading to the onset of AIDS, over a span of 5 to 10 years \cite{cai2009}.

Despite global efforts to reduce HIV transmission, late diagnosis continues to be a critical area of concern, as it is associated with higher risk of onward transmission, increased morbidity and mortality, and higher treatment costs as compared to timely diagnosis \cite{late2020}. Late diagnosis is characterized by a CD4+ T-cell count $\leq 350\text{ cells/mm}^{3}$ within 91 days of diagnosis \cite{crox2022}. According to an estimate, approximately 5.4 million people were living with HIV without the knowledge of their disease status in 2023, which is a concerning figure, given its critical impact on the further transmission of the disease within their communities \cite{fact2024}. The leading cause of late HIV diagnosis is insufficient testing, which is influenced by a complex network of societal, systemic, and individual factors. Structural or interpersonal discrimination based on sexual orientation, ethnicity, or age creates barriers to routine testing, health-seeking behaviours, and testing in response to symptoms of the disease \cite{karv2022}. Implementation of effective interventions to increase timely HIV testing can be achieved by identifying barriers in our local communities. These interventions include public information campaigns, wider availability of free at-home testing, and healthcare initiatives aimed at reducing HIV-related stigma and discrimination \cite{boar2024}. Self-imposed behavioral modifications, driven by psychological fear and increased awareness through media and public health campaigns, play a critical role in reducing HIV transmission globally. \cite{ghos2018, xue2022}.

Ongoing efforts towards diagnosis, treatment, and prevention have progressed significantly, but HIV is still not a fully curable disease. However, along with early diagnosis of HIV infection, timely introduction of treatment like Antiretroviral Therapy (ART) have been pivotal in reducing mortality rates among individuals living with HIV. With the increasing accessibility to medical facilities, global ART coverage was around $77\%$ by the end of 2023 \cite{fact2024}. ART suppress viral replication significantly, which strengthens the immune system of the infected population. The rate of virologic failure has significantly decreased with the modern ART regimens, which has helped in reducing the life expectancy gap between people living with HIV (PLWH) and the general population. Virologic failure of ART is characterized by two consecutive assessments of viral load above 200 copies/ml \cite{cutr2020}. Virologic failure may result from one or more of three following factors: patient-related factors, viral factors, and treatment-related factors. Sub-optimal adherence to ART is the most common cause of virologic failure, which is related to the behaviour of patient \cite{cutr2020}. Psychosocial factors, including housing instability, limited access to good healthcare system, issues with drug side effects, treatment costs, and high burden of medicines, are major contributors to poor adherence to ART. Depending on the treatment regimen, research suggests that maintaining an adherence rate of 80 to 85\% or higher is typically sufficient for effective HIV viral suppression \cite{byrd2019}.  With the increasing accessibility of ART, maintaining optimal drug adherence is essential for controlling the development of drug resistance.

HIV drug resistance (HIVDR) remains a persistent clinical and public health concern, as it can develop in newly infected patients (who access treatment) and be transmitted from infected individuals to others, posing a serious threat to treatment efficacy and disease control \cite{litt2002, cutr2020}. Drug-resistant mutations typically develop during treatment due to ongoing viral replication, which can result from either sub-optimal drug concentrations caused by poor adherence (patient-related factor) or from an inappropriate combination of ART regimens that fail to effectively suppress viral levels, even with good adherence (treatment-related factor). A study on drug-resistant mutations found a positive correlation between poor adherence to treatment and the number of drug-resistant HIV-infected individuals in the Leningrad region \cite{shch2022}. Risk factors associated with virologic failure due to drug resistance are higher viral loads or lower CD4+ T-cell concentration, before initiation of treatment, particularly with less potent ART regimens. Also, failure of viral suppression can often be attributed to pharmacological factors, such as drug-drug or drug-food interactions, which result in sub-optimal pharmacokinetics and prevent sufficient serum concentrations of the antiretroviral agent \cite{calc2017}. 

To address virologic failure associated with first-line introductory ART regimens, it is recommended to switch to second-line ART regimens as alternative therapeutic options, particularly in the presence of drug resistance \cite{pane2024}. After successful identification of the factors contributing to virologic failure, the next step is to determine the most effective ART strategy for the future. Key factors in selecting subsequent ART regimen(s) include  a thorough assessment of the adherence status, investigation of previous ART or pre-treatment drug history, and evaluation of past and current HIV drug resistance tests \cite{pane2024}. Multiple large randomized controlled trials, particularly in resource-constrained settings where NNRTI-based regimens have been used as first-line therapy, have examined various options for second-line regimen combinations. These studies recommend prescribing at least two or three fully active ART drugs, with the potential addition of partially active drugs to leverage their immunologic and virologic benefits \cite{seco2013, abou2019, pane2024}. These studies further emphasized that second-line regimens should include agents with a high resistance barrier, such as boosted protease inhibitors (PIs) or dolutegravir (DTG), to minimize the risk of further resistance development. While second-line regimens can restore viral suppression and prevent disease progression, they present challenges, including potential toxicity, increased complexity and cost of the treatment \cite{cutr2020}. Despite these limitations, second-line therapy remains essential for maintaining viral load and optimizing long-term treatment outcomes in patients with first-line ART failure.

To effectively evaluate epidemic progression and the impact of biomedical or other intervention strategies, it is crucial to understand the evolutionary behavior of viruses and their mutant strains within an epidemiological framework. The interplay of ecological and evolutionary mechanisms during the course of infection plays a central role in shaping viral evolution at the population level \cite{lyth2013}. Mathematical modeling has proven to be a powerful tool for integrating various biological and epidemiological processes to identify the key determinants of disease dynamics. Several modeling frameworks have been employed in this context, including compartmental models \cite{kerm1927, tuck1998, ren2018}, data-driven phenomenological approaches \cite{chow2016, pell2018}, and statistical inference techniques involving parameter estimation through simulation-based methods \cite{shre2011,stoc2017}. Many researchers have adopted complex network models to capture disease transmission patterns arising from interactions among individuals \cite{liu2013, gupt2021}. In this study, we use a compartmental modeling framework to capture the dynamics of HIV by dividing the total population into distinct classes that represent various stages of infection. Compartmental models have been extensively used in the literature to investigate HIV transmission dynamics at the population level \cite{ande1986, may1987, shar2008, cai2009, silv2017, ghos2018, akud2018, poon2022, xue2022, gurs2023, poon2024}. 

Anderson and May \cite{ande1986, may1987} were among the first to explore HIV transmission within communities, highlighting the influence of different biological processes in shaping the early dynamics of the epidemic following the infection. In \cite{cai2009}, the authors studied an HIV/AIDS treatment model that distinguishes between two stages of infection: asymptomatic and symptomatic. They further compared this framework with a time-delay model, where the initiation of treatment is assumed to take effect after a delay. A comprehensive global analysis of a Susceptible-Infected-Chronic-AIDS (SICA) compartmental model was carried out in \cite{silv2017}. The study concluded that substantial increase in treatment coverage and transmission reduction were essential to meet the UNAIDS 2030 target of ending AIDS. A simplified SI‑type model to assess how both media-driven awareness and self‑imposed psychological fear influence HIV/AIDS dynamics was proposed in \cite{ghos2018}. In this study, the authors have concluded that while fear reduces transmission, media-based awareness is significantly more effective in reducing the disease burden than fear alone. Further, Gurski et al. \cite{gurs2023} developed a dynamic compartmental model that captured staged HIV transmission within both casual and long-term partnerships, and incorporated treatment processes alongside partnership dynamics. Their analysis concluded that infection rates vary substantially by partner type and disease stage, indicating that long-term partnerships contribute uniquely to the overall transmission dynamics.

Recent studies have applied multi-strain mathematical models to examine how ecological interactions and evolutionary mechanisms influence the transmission dynamics of infectious diseases \cite{shar2008, fudo2020, kudd2021, maka2022, poon2022, poon2024}. Interactions among multiple pathogen strains, driven by immune responses, fitness trade-offs, and transmission dynamics, play a crucial role in shaping pathogen population behavior across spatial and temporal scales. These interactions influence strain persistence and genetic diversity \cite{maka2022}. In \cite{shar2008}, authors analyzed a multi-strain HIV model for both wild-type and drug-resistant strains under ART, and found that resistant strains can outcompete wild-type strains or coexist depending on their relative fitness. Various recent studies have demonstrated that, in the long term, drug-sensitive strains often face competitive exclusion from the population in the presence of emerging drug-resistant strains \cite{fudo2020, kudd2021, poon2022, poon2024}. Recently, Poonia and Chakrabarty \cite{poon2022, poon2024} developed two-strain HIV models to study drug-sensitive and drug-resistant infection mechanism by explicitly incorporating treatment adherence. Their study examined how varying ART coverage and adherence levels to ART influence the transmission dynamics of multiple HIV strains. The findings highlight that both treatment adherence and transmission rates play critical roles in determining strain coexistence or competitive exclusion. Also, enhancing adherence was shown to significantly reduce the emergence of drug-resistant infections and the overall disease burden.

Effective disease management requires minimizing the disease burden while keeping control costs low. The optimal control theory is a widely used framework to achieve this balance \cite{lenh2007, okos2013, akud2018, agus2019, peni2020, yusu2023}. Okosun et al.  \cite{okos2013} developed an optimal control model to evaluate the combined impact of condom use, screening of unaware HIV-positive individuals, and treatment in a population with ongoing susceptible immigration. They concluded that undiagnosed HIV infections impose substantial burden on community-level costs. Moreover, implementing a combination of all three interventions was identified as the most cost-effective strategy for reducing HIV transmission. In \cite{akud2018}, the authors analyzed an optimal control model incorporating time-dependent ART allocation and demonstrated that prioritizing treatment for individuals in the early stages of infection is most effective in reducing new infections and total infection-years. In contrast, allocating resources to later-stage patients resulted in minimizing overall costs and HIV-related mortality. Peni et al. \cite{peni2020} formulated a nonlinear model predictive control framework with temporal-logic constraints to manage COVID-19, using a discrete-time compartmental model. Their findings emphasized the importance of early intervention, continuous monitoring of the susceptible population, and the flexibility to adapt control measures in response to the evolving dynamics of the outbreak. 

At the $20^{th}$ International AIDS Conference in 2014, UNAIDS introduced the 90-90-90 global targets as a strategic framework for HIV epidemic control \cite{unai2014}. These goals aimed for 90\% of all individuals living with HIV to be diagnosed, 90\% of those diagnosed to receive ART, and 90\% of individuals on ART to achieve viral suppression by the end of 2020.  Furthermore, UNAIDS has proposed the more ambitious 95-95-95 targets to be achieved by 2030, which are expected to result in approximately 86\% overall viral suppression. Achieving these goals in an optimal manner remains a significant challenge, particularly for low- and middle-income countries with a high prevalence of HIV. Xue et al. \cite{xue2022} incorporated the 90-90-90 framework by mathematically quantifying each step-diagnosis, treatment uptake, and viral suppression, within their compartmental model, and formulated the model with behavioral fear effects. They concluded that, while implementing the 90-90-90 targets significantly reduces new infections, it would still require approximately 26 years to eliminate new HIV cases entirely. In addition, the optimal control to reduce the development and transmission of drug-resistant HIV strains introduces additional complexities, especially in the context of treatment strategies involving first- and second-line treatments. 

Despite extensive research on HIV modeling, a comprehensive mathematical investigation that captures the multi-strain dynamics of HIV, including the emergence of drug-resistant strains due to sub-optimal adherence to first-line treatment and the subsequent switch to second-line treatment following the diagnosis of resistance, remains limited. This study addresses this gap by incorporating distinct compartments for undiagnosed and diagnosed infected individuals, as well as for drug-sensitive and drug-resistant strains. Also, we propose a structured mechanism for treatment switching, once the development of drug resistance is identified. Furthermore, we integrate this framework within an optimal control setting. We propose a range of control strategies aimed at enhancing HIV diagnosis rates, expanding treatment accessibility, improving drug adherence among individuals receiving ART, and reducing the overall infected population, with particular focus on the control of drug-resistant HIV infections. Along with the equilibrium point analysis for understanding the long-term dynamics of the model, a comprehensive local and global sensitivity analysis is conducted to identify the most influential parameters governing the disease dynamics. In addition, a detailed cost-effectiveness analysis offers critical insights for public health experts, helping them to allocate resources efficiently and focus on interventions that can optimally reduce the HIV transmission. To the best of our knowledge, no existing compartmental model simultaneously captures these critical aspects, making our study a significant step toward better understanding and controlling the spread of multiple strains of HIV in a community.

The main objective of this study is to develop and analyze a comprehensive multi-strain HIV transmission model that captures the emergence of drug-resistant strains due to sub-optimal adherence to first-line treatment, incorporates treatment switching to second-line therapy upon diagnosis of resistance, and evaluates the effectiveness and cost-efficiency of various intervention strategies. To achieve this, we introduce a compartmental model that captures the transmission dynamics of drug-sensitive and drug-resistant HIV strains in Section \ref{03:sec:mathematica_model}. Section \ref{03:sec:equilibrium_points} presents a detailed theoretical analysis of the model, including its biological feasibility and the existence and stability of various equilibrium points. Parameter values and initial conditions for the state variables are computed in Section \ref{03:sec:parameter_estimation}, drawing upon data from existing literature. Section \ref{03:sec:sensitivity_analysis} provides an extensive sensitivity analysis for different model outcomes using Partial Rank Correlation Coefficients (PRCC) and Sobol's method to identify the most influential parameters in the short-term and long-term dynamics. In Section \ref{03:sec:optimal_control_analysis}, we formulate and analyze an optimal control problem aimed at minimizing the infection burden while optimizing the cost of multiple interventions.  Section \ref{03:sec:numerical_simulation} presents a comprehensive numerical simulation study, illustrating the analytical findings and exploring a set of time-dependent control strategies designed to reduce disease spread. This section also includes an adjoint-based sensitivity analysis to evaluate the optimal allocation of additional public health resources. To further understand the role of each control variable, a control contribution analysis is conducted using Shapley values from cooperative game theory, highlighting both individual and synergistic effects. Finally, we conclude by summarizing the key findings and discussing their biological and public health implications in Section \ref{03:sec:conclusion}.

\section{The mathematical modeling framework}{\label{03:sec:mathematica_model}}
To construct a multi-strain mathematical model that represents the transmission dynamics of HIV with multiple treatment options and drug adherence, we use the compartmental modeling framework based on different stages of the infection. The total sexually active population is categorized into eight mutually exclusive compartments: susceptible individuals $(S)$, undiagnosed and infected individuals with the drug-sensitive strain $(I_{SU})$, diagnosed and infected individuals with the drug-sensitive strain $(I_{SD})$, individuals receiving first-line treatment $(T_{1})$, undiagnosed and infected individuals with the drug-resistant strain $(I_{RU})$, diagnosed and infected individuals with the drug-resistant strain $(I_{RD})$, individuals receiving second-line treatment $(T_{2})$, and individuals who have progressed to AIDS $(A)$. 
The transitions of population among these compartments with time are governed by a system of coupled nonlinear ordinary differential equations presented as follows:
\begin{eqnarray}
&S^{\prime}&=\lambda-\alpha_S S(I_{SU}+c I_{SD})-\alpha_R S (I_{RU}+ c I_{RD}) -\mu S, \nonumber \\
&I_{SU}^{\prime}&=\alpha_S S (I_{SU}+c I_{SD})-k_S I_{SU}-\theta_1 I_{SU}-\mu I_{SU},
\nonumber \\
&I_{SD}^{\prime}&= k_{S} I_{SU}-\beta_{S} \eta_{1} I_{SD}-\theta_{1} (1-\eta_{1}) I_{SD}
-\mu I_{SD}, 
\nonumber \\
&T_{1}^{\prime}&=\beta_{S} \eta_{1} I_{SD}-\gamma (1-\epsilon) T_{1}-\theta_{2} \epsilon T_{1}-\mu T_{1},
\nonumber \\
&I_{RU}^{\prime}&=\alpha_R S (I_{RU}+ c I_{RD})+\gamma (1-\epsilon) T_{1}-k_{R} I_{RU}-\theta_{1} I_{RU}-\mu I_{RU},
\nonumber \\
&I_{RD}^{\prime}&=k_{R} I_{RU}-\beta_{R} \eta_{2} I_{RD}-\theta_{1}(1-\eta_2) I_{RD}-\mu I_{RD},
\nonumber\\
&T_{2}^{\prime}&=\beta_{R} \eta_{2} I_{RD}-\theta_{2} T_{2}-\mu T_{2},
\nonumber \\
&A^{\prime}&=\theta_1 I_{SU}+\theta_{1} (1-\eta_{1}) I_{SD}+\theta_{2} \epsilon T_{1}+\theta_{1} I_{RU}+\theta_{1}(1-\eta_2) I_{RD}+\theta_{2} T_{2}-(\mu+\mu_d) A 
\label{3_system_main}
\end{eqnarray}
with the initial condition $\displaystyle{\left(S(0), I_{SU}(0), I_{SD}(0), T_{1}(0), I_{RU}(0), I_{RD}(0), T_{2}(0), A(0)\right)}$ to lie within the biologically feasible region $\mathcal{R}$, defined as
$$\mathcal{R}=\left\{(S,I_{SU},I_{SD},T_{1},I_{RU},I_{RD},T_{2},A)\in \mathbb{R}^{8}_{+} : 0\leq S+I_{SU}+I_{SD}+T_{1}+I_{RU}+I_{RD}+T_{2}+A\leq \frac{\lambda}{\mu}\right\},$$
where $\displaystyle{\mathbb{R}^{8}_{+}}$ refers to the non-negative orthant along with its lower-dimensional faces. 

The susceptible individuals are recruited into the population at a constant rate $\lambda$ as new individuals become sexually active. These individuals can become infected through effective contact with both drug-sensitive and drug-resistant infected individuals. We assume that effective contact results in the transmission of the same strain to newly infected individuals. The force of infection, expressed as a function of both susceptible and infected individuals, is a critical component of epidemiological modeling. Based on the principle of mass action, we define the force of infection from drug-sensitive individuals as proportional to $\alpha_S (I_{SU}+c I_{SD})$, and from drug-resistant individuals to $\alpha_R (I_{RU}+ c I_{RD})$. The parameter $c\in [0,1]$ captures the reduction in transmission rate of diagnosed individuals compared to undiagnosed ones. This reduction reflects the behavioral change following diagnosis, as individuals who are aware of their HIV positive status are less likely to engage in risk-prone behavior. Once infected, susceptible individuals move into undiagnosed compartments $(I_{SU} \text{ or } I_{RU})$ based on whether the infection is acquired from a drug-sensitive or drug-resistant individual. Undiagnosed individuals may get tested and become diagnosed at rates $k_{S}$ $(\text{for } I_{SU})$ and $k_{R}$ $(\text{for } I_{RU})$, or may progress directly to the AIDS stage at rate $\theta_{1}$ if not tested in time, where $\theta_{1}$ is the conversion rate of untreated infected population to the population in AIDS stage. While diagnosis helps in earlier detection of infection, its impact on transmission remains limited unless it is promptly followed by an effective treatment initiation. We assume that first-line treatment coverage is limited to a fraction $(\eta_{1})$ of the diagnosed drug-sensitive infected population, due to resource constraints. This fraction initiates first-line treatment at rate $\beta_{S}$, while the remaining $(1-\eta_{1})$ progresses to the AIDS stage at rate $\theta_{1}$. 

Further, we assume that only a fraction $\epsilon$ of individuals receiving first-line treatment are adherent, while the remaining fraction $(1-\epsilon)$ are non-adherent. Individuals under treatment are considered non-infectious, as effective treatment suppresses viral load to levels that prevent onward transmission. However, the non-adherent individuals are assumed to develop drug resistance and, unlike adherent individuals, do not progress directly to the AIDS stage. Instead, they transition to the undiagnosed drug-resistant infected compartment $(I_{RU})$ at a rate $\gamma$. In contrast, adherent individuals progress to the AIDS stage at a reduced rate $\theta_{2}$, as treatment suppresses but does not cure the infection. Note that the progression rate to AIDS for untreated infected individuals $(\theta_{1})$ is significantly higher than that of those receiving optimal treatment $(\theta_{2})$. After the diagnosis of drug-resistant infection, the second-line treatment is given to only a fraction $\eta_2$, reflecting constraints in resource availability, at a rate $\beta_{R}$. The remaining fraction of diagnosed drug-resistant infected individuals progresses to the AIDS stage at a rate $\theta_{1}$. For those receiving second-line treatment, we assume full adherence due to the advanced stage of infection and enhanced medical supervision. Consequently, these individuals progress to the AIDS stage at a rate $\theta_{2}$. Individuals in the AIDS stage are assumed to be sufficiently aware of their condition to not contribute to further transmission. A natural death rate $\mu$ applies uniformly across all compartments. Since these treatments cannot cure the infection but only delay disease progression, all infected individuals are assumed to eventually progress to the AIDS stage unless they die naturally beforehand. Consequently, we consider disease-induced mortality to occur exclusively in the AIDS compartment, with a rate $\mu_d$. All model parameters and their biological interpretations are summarized in Table \ref{03:tab:parameter_description}, and a schematic representation of the system \eqref{3_system_main} is presented in Figure \ref{03:fig:flow_diagram}.

\begin{figure}[htbp]
\begin{center}
\tikzstyle{block} = [rectangle, minimum width=1.5cm, minimum height=1.1cm, draw, thick, text centered, font=\large, rounded corners]
\tikzstyle{arrow} = [thick,->,>=latex]

\colorlet{susceptiblecolor}{green!50}
\colorlet{sensitiveundiagnosedcolor}{red!80}
\colorlet{resistantundiagnosedcolor}{red!80}
\colorlet{sensitivediagnosedcolor}{red!40}
\colorlet{resistantdiagnosedcolor}{red!40}
\colorlet{treatment1color}{yellow!60}
\colorlet{treatment2color}{yellow!60}
\colorlet{aidscolor}{gray!25}

    \begin{tikzpicture}[thick,->,>=latex,node distance=3.5cm]
        \node (box1) [block, fill=susceptiblecolor] {$S$};
        \node (box2) [block, left of=box1, xshift=-1cm, yshift=-3.5cm, fill=sensitiveundiagnosedcolor]  {$I_{SU}$};
        \node (box5) [block, right of=box1, xshift=1cm, yshift=-3.5cm, fill=resistantundiagnosedcolor] {$I_{RU}$};
        \node (box3) [block, below of=box2, fill=sensitivediagnosedcolor] {$I_{SD}$};
        \node (box4) [block, right of=box3,  xshift=1cm, fill=treatment1color] {$T_{1}$};
        \node (box6) [block, below of=box5, fill=resistantdiagnosedcolor] {$I_{RD}$};
        \node (box7) [block, below of=box6, yshift=-0.2cm, fill=treatment2color] {$T_2$};
        \node (box8) [block, below of=box4, yshift=-0.2cm, fill=aidscolor] {$A$};
        
        \draw[dashed,-] (box1) -- node[anchor=south, xshift = -0.0cm, rotate=19] {} (box3);
        \draw[dashed,-] (box1.north) node[anchor=south, xshift = -5.0cm, rotate=100] {} to[bend left=90](box6.east);
        \draw[arrow] (box1) -- node[anchor=south, xshift = -0.0cm, rotate=38] {$\alpha_S (I_{SU}+c I_{SD})$} (box2);
        \draw[->] (box1) -- node[anchor=south, rotate=-38] {$\alpha_R (I_{RU}+c I_{RD})$} (box5);
        \draw[->] (box2) -- node[anchor=south, xshift=-0.3cm] {$k_S$} (box3);
        \draw[->] (box2.west) node[anchor=south, xshift=-0.6cm,yshift=-4.3cm] {$\theta_1$} to[out=215, in=210] (box8.west);
        \draw[->] (box3) -- node[anchor=south] {$\beta_S \eta_1$} (box4);
        \draw[->] (box3) -- node[anchor=south, rotate=-38] {$\theta_1 (1-\eta_1)$} (box8);
        \draw[->] (box4) -- node[anchor=south, rotate=36] {$\gamma (1-\epsilon)$} (box5);
        \draw[->] (box4) -- node[anchor=south, xshift=0.3cm] {$\theta_2 \epsilon$} (box8);
        \draw[->] (box5) -- node[anchor=south, xshift=0.3cm] {$k_R$} (box6);
        \draw[->] (box5) -- node[anchor=south, xshift=0.3cm, yshift=-0.3cm] {$\theta_1$} (box8);
        \draw[->] (box6) -- node[anchor=south, xshift=0.5cm] {$\beta_R \eta_2$} (box7);
        \draw[->] (box6) -- node[anchor=south, rotate=40] {$\theta_1 (1-\eta_2)$} (box8);
        \draw[->] (box7) -- node[anchor=south] {$\theta_2$} (box8);
          
        \draw[->] (box1.south) -- node[anchor=north, yshift=-0.5cm] {$\mu$} ++(0,-1) ;
        \draw[<-] (box1.north) -- node[anchor=south, yshift=0.5cm] {$\lambda$} ++(0,1);
        \draw[->] (box2.north) -- node[anchor=south, yshift=0.5cm] {$\mu$} ++(0,1);
        \draw[->] (box3.south) -- node[anchor=north, yshift=-0.5cm] {$\mu$} ++(0,-1);
        \draw[->] (box4.north) -- node[anchor=south, yshift=0.5cm] {$\mu$} ++(0,1);
        \draw[->] (box5.north) -- node[anchor=south, yshift=0.5cm] {$\mu$} ++(0,1);
        \draw[->] (box6.east) -- node[anchor=east, xshift=1cm] {$\mu$} ++(1,0);
        \draw[->] (box7.south) -- node[anchor=north, yshift=-0.5cm] {$\mu$} ++(0,-1);
        \draw[->] (box8.south) -- node[anchor=north, yshift=-0.5cm] {$\mu+\mu_d$} ++(0,-1);   
        
    \end{tikzpicture}
    \caption{Schematic representation of the system (\ref{3_system_main}). Solid arrows represent direct transitions between compartments resulting from effective contact, while dashed arrows indicate the presence of effective contacts that contribute to infection risk, without resulting in transitions to the contacted compartment.}
    \label{03:fig:flow_diagram}
\end{center}
\end{figure}
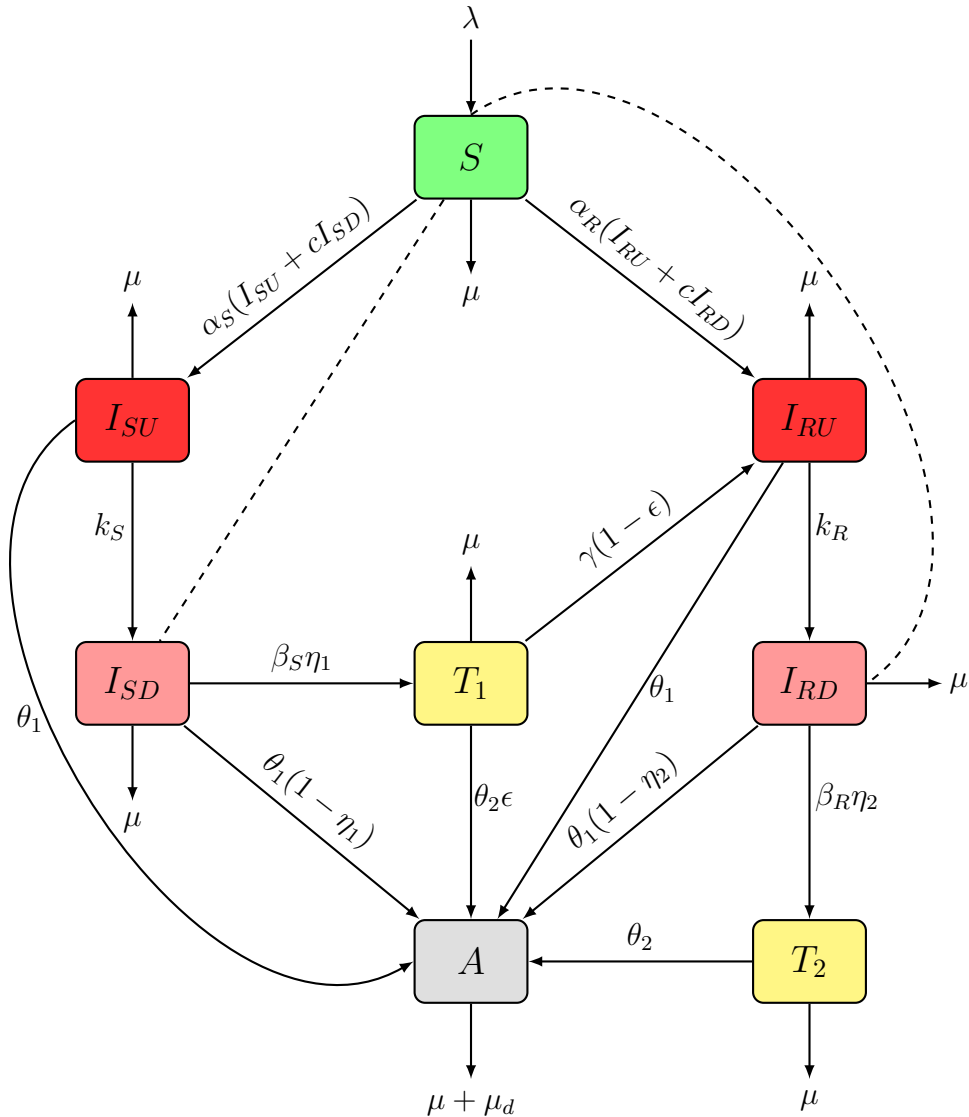

\begin{table}[htbp]
\centering
\begin{tabular}{>{\centering\arraybackslash}m{2cm} p{15.5cm}}
\toprule
\textbf{Parameter} & \textbf{Biological Description} \\
\midrule
$\lambda$     & Recruitment rate of susceptible individuals \\
$\alpha_S$    & Transmission rate per effective contact between a susceptible individual and a drug-sensitive infected individual \\
$\alpha_R$    & Transmission rate per effective contact between a susceptible individual and a drug-resistant infected individual \\
$k_S$         & Diagnosis rate of infection among drug-sensitive infected individuals \\
$k_R$         & Diagnosis rate of infection among drug-resistant infected individuals \\
$c$           & Relative infectiousness of diagnosed individuals compared to undiagnosed individuals \\
$\theta_1$    & Progression rate of infection in untreated infected individuals to AIDS    stage \\
$\theta_2$    & Progression rate of infection in optimally adherent treated infected individuals to AIDS stage \\
$\beta_S$     & First-line treatment initiation rate for diagnosed drug-sensitive infected individuals \\
$\beta_R$     & Second-line treatment initiation rate for diagnosed drug-resistant infected individuals \\
$\gamma$      & Rate at which non-adherent individuals develop drug resistance \\
$\eta_1$      & Proportion of diagnosed drug-sensitive infected individuals receiving first-line treatment \\
$\eta_2$      & Proportion of diagnosed drug-resistant infected individuals receiving second-line treatment \\
$\epsilon$    & Proportion of individuals receiving first-line treatment who are adherent \\
$\mu$         & Natural death rate \\
$\mu_d$       & Disease-induced death rate \\
\bottomrule
\end{tabular}
\caption{Model parameters and their biological descriptions for system (\ref{3_system_main}).}
\label{03:tab:parameter_description}
\end{table}

\section{Mathematical analysis of the system}{\label{03:sec:equilibrium_points}}
To understand the long-term behavior of the system and its sensitivity to changes in epidemiological parameters, we identify the equilibrium points and analyze their local stability. This section explores the parametric conditions under which the disease either dies out or persists in the population. In addition, we examine whether the coexistence of both drug-sensitive and drug-resistant strains is possible, and if so, determine the conditions that enable such coexistence. But first we establish a preliminary result regarding the non-negativity and boundedness of the model solutions, which ensures that the system is biologically well-posed. Since each state variable in the model represents a population density at a given time, it is required that the solutions remain non-negative and bounded for all time, provided the initial conditions are non-negative. To guarantee this biological feasibility, we present the following theorem.
\begin{theorem}
    The biologically feasible region $\mathcal{R}$ is a positively invariant set for the system \eqref{3_system_main}. This implies that any solution $\left(S(t),I_{SU}(t),I_{SD}(t),T_{1}(t),I_{RU}(t),I_{RD}(t),T_{2}(t),A(t)\right)$ of the system, starting from an initial condition $(S(0), I_{SU}(0), \allowbreak I_{SD}(0), T_{1}(0), I_{RU}(0), I_{RD}(0), T_{2}(0), A(0)) \in \mathcal{R}$, will remain within the region $\mathcal{R}$ for all $t > 0$. 
\end{theorem}
\begin{proof}
    The proof of this theorem follows a similar approach to that of Theorem 1 in Section 3 of \cite{poon2024}.
\end{proof}

In the system \eqref{3_system_main}, the first seven equations are independent of the AIDS compartment ($A$). The eighth equation, corresponding to the dynamics of $A$, can be solved separately and yields a positive value for $A$ whenever the other state variables are positive. Therefore, we focus on the following reduced system for further analysis.
\begin{eqnarray}
&S^{\prime}&=\lambda-\alpha_S S(I_{SU}+c I_{SD})-\alpha_R S (I_{RU}+ c I_{RD}) -\mu S \nonumber \\
&I_{SU}^{\prime}&=\alpha_S S (I_{SU}+c I_{SD})-A_{1} I_{SU}
\nonumber \\
&I_{SD}^{\prime}&=k_{S} I_{SU}-A_{2} I_{SD}
\nonumber \\
&T_{1}^{\prime}&=\beta_{S} \eta_{1} I_{SD}- A_{3} T_{1}
\nonumber \\
&I_{RU}^{\prime}&=\alpha_R S (I_{RU}+ c I_{RD})+\gamma (1-\epsilon) T_{1}-A_{4} I_{RU}
\nonumber \\
&I_{RD}^{\prime}&=k_{R} I_{RU}-A_{5} I_{RD}
\nonumber\\
&T_{2}^{\prime}&=\beta_{R} \eta_{2} I_{RD}-A_{6} T_{2}
\label{3_system_reduced}
\end{eqnarray}
where $A_{1}=k_S + \theta_1 + \mu,~ A_{2}=\beta_{S} \eta_{1}+\theta_{1} (1-\eta_{1})+\mu,~ A_{3}=\gamma (1-\epsilon)+\theta_{2} \epsilon+\mu,~\allowbreak A_{4}=k_R + \theta_1 + \mu,~ \allowbreak A_{5}=\beta_{R} \eta_{2}+\theta_{1}(1-\eta_2)+\mu,~ \allowbreak A_{6}=\theta_2 + \mu.$

\subsection{Disease elimination and basic reproduction numbers}{\label{03:subsec:equilibrium_point_0}}
The elimination of the infection from the system is contingent on the existence of a disease free equilibrium state and its stability. In order to obtain the disease free equilibrium (DFE) point of the system (\ref{3_system_reduced}), we set all variables to zero which are either infectious $(I_{SU},~ I_{SD},~ I_{RU},~ I_{RD})$ or infected $(T_{1},~T_{2})$ and find the steady state of non-infected variable $(S)$. As a result, the DFE point is given by,
\begin{eqnarray}
E^{(0)}&=&\left(S^{(0)},I_{SU}^{(0)}, I_{SD}^{(0)}, T_{1}^{(0)}, I_{RU}^{(0)}, I_{RD}^{(0)}, T_{2}^{(0)}\right) \nonumber \\
&=&\left(\frac{\lambda}{\mu},0, 0, 0, 0, 0, 0\right)\nonumber
\end{eqnarray}
This indicates that the whole population becomes susceptible to HIV infection once DFE point is achieved. Note that, although, this is an ideal situation which is most likely to be unachievable in short term, we still need further mathematical investigation to derive some conditions under which the solutions of the system (\ref{3_system_reduced}) converge to the DFE point. For that, first we determine the basic reproduction number for the system as well as individual viral strains, which will provide crucial insights for the dynamics of the presented model. Since we have four types of infectious individuals in the system, the reduced infection subsystem linearized about the DFE point is as follows
$$\dot{x}=(F+V)x,$$ 
where $\displaystyle{x=\left(I_{SU}, I_{SD}, I_{RU}, I_{RD} \right)^{T} \in \mathbb{R}^4_{+} }$. Also, $F$ and $V$ are Jacobian matrices at point $\displaystyle{E^{(0)}}$, given by
$$F=\begin{pmatrix}
\frac{\lambda \alpha_{S}}{\mu} & \frac{\lambda c \alpha_{S}}{\mu} & 0 & 0\\
0 & 0 & 0 & 0 \\
0 & 0 & \frac{\lambda \alpha_{R}}{\mu} & \frac{\lambda c \alpha_{R}}{\mu} \\
0 & 0 & 0 & 0
\end{pmatrix}, \quad 
V=\begin{pmatrix}
A_{1} & 0 & 0 & 0\\
-k_{S} & A_{2} & 0 & 0 \\
0 & 0 & A_{4} & 0 \\
0 & 0 & -k_{R} & A_{5}
\end{pmatrix}$$
which correspond to the transmission of the virus, describing the production of newly infected individuals and the transition of infectiousness, describing the change in status of infected individuals, respectively. Therefore, the next generation matrix $(FV^{-1})$ is obtained as,
$$FV^{-1}=\begin{pmatrix}
\frac{\lambda \alpha_{S}\left(A_{2}+c k_{S}\right)}{\mu A_{1} A_{2}} & \frac{\lambda c \alpha_{S}}{\mu A_{2}} & 0 &0\\
0 & 0 & 0 & 0 \\
0 & 0 & \frac{\lambda \alpha_{R}\left(A_{5}+c k_{R}\right)}{\mu A_{4} A_{5}} & \frac{\lambda c \alpha_{R}}{\mu A_{5}}\\
 0 & 0 & 0 & 0 
\end{pmatrix}.$$
The basic reproduction number $\displaystyle{(R_{0})}$ for system (\ref{3_system_reduced}) is the spectral radius or dominant eigenvalue of the next generation matrix $FV^{-1}$. Therefore, $\displaystyle{R_{0}=\max\{R_{0}^{(S)},R_{0}^{(R)}\}}$, where
$$R_{0}^{(S)}=\frac{\lambda \alpha_{S}\left(A_{2}+c k_{S}\right)}{\mu A_{1} A_{2}} =\frac{\lambda \alpha_{S}\left(\beta_{S} \eta_{1}+\theta_{1} (1-\eta_{1})+\mu+c k_{S}\right)}{\mu (k_S + \theta_1 + \mu) (\beta_{S} \eta_{1}+\theta_{1} (1-\eta_{1})+\mu)},$$
and 
$$R_{0}^{(R)}=\frac{\lambda \alpha_{R}\left(A_{5}+c k_{R}\right)}{\mu A_{4} A_{5}}=\frac{\lambda \alpha_{R}\left(\beta_{R} \eta_{2}+\theta_{1}(1-\eta_2)+\mu+c k_{R}\right)}{\mu (k_{R}+\theta_{1}+\mu) (\beta_{R} \eta_{2}+\theta_{1}(1-\eta_2)+\mu)}.$$
Additionally, we define the ``relative basic reproduction number'' as the ratio of $R_{0}^{(S)}$ and $R_{0}^{(R)}$. This quantity signifies the strength of the dominance of the sensitive strain over the resistant strain in terms of transmission potential. The relative basic reproduction number is denoted as $R_{0}^{(SR)}$, which is given by:
$$R_{0}^{(SR)}=\frac{\alpha_{S} A_{4} A_{5} \left(A_{2}+c k_{S}\right)}{\alpha_{R} A_{1} A_{2} \left(A_{5}+c k_{R}\right)} =\frac{\alpha_{S}\left(\beta_{S} \eta_{1}+\theta_{1} (1-\eta_{1})+\mu+c k_{S}\right) (k_{R}+\theta_{1}+\mu) (\beta_{R} \eta_{2}+\theta_{1}(1-\eta_2)+\mu)}{\alpha_{R} \left(\beta_{R} \eta_{2}+\theta_{1}(1-\eta_2)+\mu+c k_{R}\right) (k_S + \theta_1 + \mu) (\beta_{S} \eta_{1}+\theta_{1} (1-\eta_{1})+\mu) }.$$

The basic reproduction number $\displaystyle{R_{0}^{(S)} (\text{or } R_{0}^{(R)})}$ quantifies the average number of secondary infections caused by a single individual infected with the drug-sensitive (or drug-resistant) strain in a fully susceptible population. Each basic reproduction number can be decomposed into two additive components corresponding to transmission from undiagnosed and diagnosed infected individuals. In $\displaystyle{R_{0}^{(S)}}$, the term $\displaystyle{\frac{\lambda \alpha_{S}}{\mu A_{1}}}$ captures the contribution from drug-sensitive undiagnosed individuals $(I_{SU})$, where $\lambda \alpha_{S}$ reflects the inflow due to new infections and $\displaystyle{\frac{1}{\mu A_{1}}}$ shows the effective infectious period, influenced by the removal through progression to other compartments or natural death. The additional term $\displaystyle{\frac{\lambda \alpha_{S} c k_{S}}{\mu A_{1} A_{2}}}$ quantifies the contribution from drug-sensitive diagnosed individuals $(I_{SD})$. Similarly, $\displaystyle{R_{0}^{(R)}}$ includes contributions from the undiagnosed $\displaystyle{\left(\frac{\lambda \alpha_{R}}{\mu A_{4}}\right)}$ and diagnosed  $\displaystyle{\left(\frac{\lambda \alpha_{R} c k_{R}}{\mu A_{4} A_{5}}\right)}$ individuals infected with the drug-resistant strain, representing their respective transmission and removal dynamics. 

Biologically, these basic reproduction numbers are shaped by various model parameters related to epidemiological, demographic, and behavioral factors. A higher recruitment rate of the susceptible population increases the pool of individuals available for infection, contributing to higher infection levels in the population. The transmission rates $\alpha_{S}$ and $\alpha_{R}$, along with the reduction factor $c$, directly linked to the disease burden. Diagnosis rates $k_{S}$ and $k_{R}$ play a mixed role as they help identify infected individuals early, thereby reducing transmission, but also increase the diagnosed population, which may continue to transmit the infection at a reduced level. As a result, diagnosis contributes both to limiting and sustaining transmission; however, its overall impact remains beneficial in reducing the spread of the disease. The natural death rate $\mu$  and the the progression rate to AIDS $\theta_{1}$ shorten the infectious period, indirectly reducing the average number of secondary cases. Similarly, parameters $\beta_{S}, \beta_{R}, \eta_{1}$ and $\eta_2$ influence the initiation of effective treatment of diagnosed individuals, which in turn reduces overall infectiousness and slows disease progression. In the subsequent analysis, we examine the significance of the basic reproduction numbers in determining the long-term dynamics of the system (\ref{3_system_main}), as the existence and stability of equilibrium points depend on these thresholds.

To determine the local behavior of the system (\ref{3_system_reduced}) around an equilibrium point, we write the corresponding linearly approximated system as,
$$\dot{X}=J\left(E^{(*)}\right) X,$$
where $X=\left( S,~ I_{SU},~ I_{SD},~ T_{1},~ I_{RU},~ I_{RD},~ T_{2}\right)^T$,  $E^{(*)}=\left(S^{(*)},I_{SU}^{(*)}, I_{SD}^{(*)}, T_{1}^{(*)}, I_{RU}^{(*)}, I_{RD}^{(*)}, T_{2}^{(*)}\right)$, and $J\left(E^{(*)}\right)$ is the Jacobian matrix of the system (\ref{3_system_reduced}) at the equilibrium point $E^{(*)}$, which is given by,
\begin{equation} \label{03:eqn:jacobian_matrix}
\begin{adjustbox}{max width=0.9\textwidth,center}
$J\left(E^{(*)}\right) =
\begin{pmatrix}
-\alpha_S \left(I_{SU}^{(*)} + c I_{SD}^{(*)}\right) - \alpha_R \left(I_{RU}^{(*)} + c I_{RD}^{(*)}\right) - \mu & -\alpha_S S^{(*)} & -c \alpha_S S^{(*)} & 0 & -\alpha_R S^{(*)} & -c \alpha_R S^{(*)} & 0 \\
\alpha_S \left(I_{SU}^{(*)} + c I_{SD}^{(*)}\right) & \alpha_S S^{(*)} - A_1 & c \alpha_S S^{(*)} & 0 & 0 & 0 & 0 \\
0 & k_S & -A_2 & 0 & 0 & 0 & 0 \\
0 & 0 & \beta_S \eta_1 & -A_3 & 0 & 0 & 0 \\
\alpha_R \left(I_{RU}^{(*)} + c I_{RD}^{(*)}\right) & 0 & 0 & \gamma (1 - \epsilon) & \alpha_R S^{(*)} - A_4 & c \alpha_R S^{(*)} & 0 \\
0 & 0 & 0 & 0 & k_R & -A_5 & 0 \\
0 & 0 & 0 & 0 & 0 & \beta_R \eta_2 & -A_6
\end{pmatrix}$
\end{adjustbox}
\end{equation} 

In order to determine the local stability of the DFE, we calculate the Jacobian matrix $J\left(E^{(0)}\right)$, which has the following characteristic equation,
\begin{align*}
   p_{0}(x)= (x+\mu)(x+A_{3})(x+A_{6})&\left(\frac{\mu x^2 +(-\lambda \alpha_{S} + \mu A_{1} + \mu A_{2})x+(\mu A_{1} A_{2}-\lambda \alpha_{S} A_{2} - c \lambda k_{S} \alpha_{S})}{\mu }\right)\\
    &\left(\frac{\mu x^2 +(-\lambda \alpha_{R} + \mu A_{4} + \mu A_{5})x+(\mu A_{4} A_{5}-\lambda \alpha_{R} A_{5} - c \lambda k_{R} \alpha_{R})}{\mu }\right)=0.
\end{align*}
Clearly, this equation always has three roots with a negative real part $(\mu,  A_{3},  A_{6})$. In addition, the remaining four roots will have a negative real part if $R_{0}^{(S)}<1$ and $R_{0}^{(R)}<1$. From the above discussion, we can propose the following result about the existence and local stability of $E^{(0)}$.

\begin{theorem}{\label{03:thm:eqilibrium_point_0}}
For the system (\ref{3_system_reduced}), the disease free equilibrium point $(E^{(0)})$ exists trivially. Further, $E^{(0)}$ is locally asymptotically stable if $R_{0}<1$, otherwise it is unstable.
\end{theorem}

\subsection{Disease persistence}{\label{03:subsec:equilibrium_point_1_2}}
The disease persists in the system if at least one of the drug-sensitive or drug-resistant infected populations maintains a positive level for a long time, which is also referred to as an endemic state. In this subsection, we will discuss all possible endemic states and analyze the conditions necessary for their existence and stability. 
Based on these conditions, we can determine which parameter to focus on to reduce or eliminate the infection from the system.

The system (\ref{3_system_reduced}) has two possible scenarios for disease persistence: one in which only the drug-resistant infected population persists and another in which both drug-sensitive and drug-resistant infected populations coexist. 
\subsubsection{Drug-resistant strain endemic equilibrium point}{\label{03:subsubsec:equilibrium_point_1}}
To obtain the drug-resistant strain endemic equilibrium point, we set $I_{SU}=0$ and solve the resulting system of equations derived from the system (\ref{3_system_reduced}). The corresponding equilibrium point is given by,
$$E^{(1)} = \left(S^{(1)},I_{SU}^{(1)}, I_{SD}^{(1)}, T_{1}^{(1)}, I_{RU}^{(1)}, I_{RD}^{(1)}, T_{2}^{(1)}, A^{(1)}\right)$$
where, $S^{(1)}= \frac{A_4 A_5}{(A_5 + c k_{R}) \alpha_{R}},~ I_{SU}^{(1)} = I_{SD}^{(1)} = T_{1}^{(1)} = 0, ~ I_{RU}^{(1)} = \frac{\mu A_{5} \left(R_{0}^{(R)}-1 \right)}{\alpha_{R} \left(A_{5}+ck_{R}\right)}, ~ I_{RD}^{(1)} = \frac{\mu k_{R} \left(R_{0}^{(R)}-1 \right)}{\alpha_{R} \left(A_{5}+ck_{R}\right)}, ~ T_{2}^{(1)} = \frac{\mu k_{R} \beta_{R} \eta_{2} \left(R_{0}^{(R)}-1 \right)}{\alpha_{R} A_{6} \left(A_{5}+ck_{R}\right)}.$

Note that at the equilibrium state, $I_{SD} = T_{1} =0$, since $I_{SU}$ is the only source compartment for these populations. Clearly, the drug-resistant strain endemic equilibrium point exists if and only if $R_{0}^{(R)}>1$. This condition suggests that the transmission and diagnosis rates of the resistant strain, and the availability of second-line therapy play a critical role in determining the existence of $E^{(1)}$. However, the local stability of $E^{(1)}$ may still be influenced by parameters associated with the sensitive strain, even though these compartments are absent in the equilibrium state. To assess the local stability of $E^{(1)}$, we compute the Jacobian matrix  $J\left(E^{(1)}\right)$ using (\ref{03:eqn:jacobian_matrix}). The corresponding characteristic equation is then given by,
$$p_{1}(x)=(x+A_{3})(x+A_{6})(x^2+B_{1} x+B_{2})(x^3+C_1 x^2 + C_2 x+C_3)=0,$$
where $\displaystyle{B_{1}=A_{1} + A_{2}-\frac{\alpha_{S} A_{4} A_{5}}{\alpha_{R} (A_{5}+c k_{R})}}$, $\displaystyle{B_{2}=A_{1} A_{2}-\frac{\alpha_{S} A_{4} A_{5} (A_{2}+c k_{S})}{\alpha_{R} (A_{5}+c k_{R})}}$, $\displaystyle{C_{1}=A_{5} + \frac{c k_{R} A_{4}}{A_{5}+c k_{R}}} + \frac{\lambda \alpha_{R} (A_{5} + c k_{R})}{A_{4} A_{5}},$ $\displaystyle{C_{2}=\frac{-\mu A_{4}^{2} A_{5}^{2} + \lambda \alpha_{R} (A_{4}+A_{5}) (A_{5} + c k_{R})^{2}}{A_{4} A_{5} (A_{5} + c k_{R})}},$ and $\displaystyle{C_{3}=-\mu A_{4} A_{5} + \lambda \alpha_{R} (A_{5}+c k_{R})}.$ Clearly, $p_{1}(x)$ has two real and negative roots as $A_{3}$ and $A_{6}$. According to the Routh–Hurwitz criterion \cite{yang2002}, $E^{(1)}$  is locally asymptotically stable if the following conditions are satisfied: $B_{i}>0$ for $i=1,2;$ $C_{i}>0$ for $i=1,2,3;$ and $C_{1} C_{2}>C_{3}.$\\
Let's consider $\displaystyle{R_{0}^{(SR)} < 1}$, then we have
\begin{align*}
    B_{1} &= A_{1} + A_{2} - \frac{A_{1} A_{2}}{A_{2} + c k_{S}} \frac{R_{0}^{(S)}}{R_{0}^{(R)}} \\
    &= A_{1}\left(1-\frac{A_{2}}{A_{2}+c k_{S}}  R_{0}^{(SR)} \right) + A_{2} > 0, \\
    B_{2} &= A_{1} A_{2} \left(1 - R_{0}^{(SR)}\right)>0.
\end{align*}
Similarly, 
\begin{align*}
    C_{1} &= >0, \\
    C_{2} &= -\frac{\mu A_{4} A_{5}}{A_{5}+c k_{R}} + \frac{\lambda \alpha_{R} (A_{5}+c k_{R})}{A_{4}} + \frac{\lambda \alpha_{R} (A_{5}+c k_{R})}{A_{5}}\\
          &= \frac{-\mu A_{4} A_{5}^{2} + \lambda \alpha_{R} (A_{5}+c k_{R})^2}{A_{5}(A_{5} + c k_{R})} + \frac{\lambda \alpha_{R} (A_{5}+c k_{R})}{A_{4}}\\
          &>\frac{\mu A_{4} A_{5} \left(R_{0}^{(R)}-1 \right)}{A_{5} + c k_{R}} > 0, \qquad \left(\because R_{0}^{(R)}>1 \right)\\
    C_{3} &= \mu A_{4} A_{5} \left(R_{0}^{(R)}-1 \right) > 0.\\
\end{align*}
Also, 
\begin{align*}
    C_{1} C_{2}-C_{3} &= \lambda \alpha_{R} c k_{R} -\lambda \alpha_{R} \mu -\frac{c k_{R} \mu A_{4}^{2} A_{5}}{(A_{5} + c k_{R})^{2}} + \frac{\lambda \alpha_{R} c k_{R} A_{4}}{A_{5}} + \frac{\lambda^{2} \alpha_{R}^{2} (A_{5} + 2 c k_{R})}{A_{4} A_{5}} + \frac{c k_{R} \mu A_{4} A_{5}}{A_{5} + c k_{R}} \\
    &~~~~+ \frac{\lambda^{2} \alpha_{R}^{2} (A_{5} + c k_{R})^{2}}{A_{4}^{2} A_{5}} + \frac{\lambda \alpha_{R} (A_{5}^{4} + c k_{R} A_{5}^{3} + \lambda \alpha_{R} c^{2} k_{R}^{2})}{A_{4} A_{5}^{2}}\\
    & > -\lambda \alpha_{R} \mu -\frac{c k_{R} \mu A_{4}^{2} A_{5}}{(A_{5} + c k_{R})^{2}} + \frac{\lambda \alpha_{R} c k_{R} A_{4}}{A_{5}} + \frac{\lambda^{2} \alpha_{R}^{2} (A_{5} + 2 c k_{R})}{A_{4} A_{5}}\\
    & > \lambda \alpha_{R} \mu \left(R_{0}^{(R)} - 1 \right) + \frac{c k_{R} \mu A_{4}^{2} A_{5}}{(A_{5} + c k_{R})^{2}} \left(R_{0}^{(R)} - 1 \right) > 0.
\end{align*}

Based on the above discussion, we propose the following result about the existence and local stability of the drug-resistant strain endemic equilibrium point $E^{(1)}$. 

\begin{theorem}{\label{03:thm:eqilibrium_point_1}}
For the system (\ref{3_system_reduced}), the drug-resistant strain endemic equilibrium point $E^{(1)}$ exists if and only if $\displaystyle{R_{0}^{(R)}>1}$. Further, $E^{(1)}$ is locally asymptotically stable if $\displaystyle{R_{0}^{(SR)} < 1}$, otherwise it is unstable.
\end{theorem}

\subsubsection{Co-existence endemic equilibrium point}{\label{03:subsubsec:equilibrium_point_2}}
Here, we investigate the simultaneous persistence of drug-sensitive and drug-resistant infected populations by analyzing the co-existence endemic equilibrium of the system (\ref{3_system_reduced}). To identify this equilibrium, we set the right-hand side of the system (\ref{3_system_reduced}) to zero and solve the resulting system of algebraic equations, which yields:

$$E^{(*)} = \left(S^{(*)},I_{SU}^{(*)}, I_{SD}^{(*)}, T_{1}^{(*)}, I_{RU}^{(*)}, I_{RD}^{(*)}, T_{2}^{(*)}\right),$$
where 
\begin{align*}
    S^{(*)} &= \frac{A_1 A_2}{(A_2 + c k_{S}) \alpha_{S}},\\
    I_{SU}^{(*)} &= \frac{\mu A_{1}^{2} A_{2}^{2} A_{3} \left(R_{0}^{(S)}-1 \right) \left(R_{0}^{(SR)}-1 \right)}{\alpha_{S} A_{1} (A_{2}+c k_{S})\left( A_{1} A_{2} A_{3} \left(R_{0}^{(SR)}-1 \right)+\gamma k_{S} \beta_{S} \eta_{1} (1-\epsilon)\right)},\\
    I_{SD}^{(*)} &= \frac{\mu  k_{S} A_{1}^{2} A_{2}^{2} A_{3} \left(R_{0}^{(S)}-1 \right) \left(R_{0}^{(SR)}-1 \right)}{\alpha_{S} A_{1} A_{2} (A_{2}+c k_{S})\left( A_{1} A_{2} A_{3} \left(R_{0}^{(SR)}-1 \right)+\gamma k_{S} \beta_{S} \eta_{1} (1-\epsilon)\right)},\\
    T_{1}^{(*)} &=  \frac{\mu  k_{S} \beta_{S} \eta_{1} A_{1}^{2} A_{2}^{2} \left(R_{0}^{(S)}-1 \right) \left(R_{0}^{(SR)}-1 \right)}{\alpha_{S} A_{1} A_{2} (A_{2}+c k_{S})\left( A_{1} A_{2} A_{3} \left(R_{0}^{(SR)}-1 \right)+\gamma k_{S} \beta_{S} \eta_{1} (1-\epsilon)\right)},\\
    I_{RU}^{(*)} &= \frac{\gamma k_{S} \beta_{S} \eta_{1} \mu A_{5} (1-\epsilon) \left(R_{0}^{(S)}-1 \right)}{\alpha_{R} (A_{5}+c k_{R})\left( A_{1} A_{2} A_{3} \left(R_{0}^{(SR)}-1 \right)+\gamma k_{S} \beta_{S} \eta_{1} (1-\epsilon)\right)},\\
    I_{RD}^{(*)} &= \frac{\gamma k_{S} k_{R} \beta_{S} \eta_{1} \mu (1-\epsilon) \left(R_{0}^{(S)}-1 \right)}{\alpha_{R} (A_{5}+c k_{R})\left( A_{1} A_{2} A_{3} \left(R_{0}^{(SR)}-1 \right)+\gamma k_{S} \beta_{S} \eta_{1} (1-\epsilon)\right)},\\
    T_{2}^{(*)} &=  \frac{\gamma k_{S} k_{R} \beta_{S} \beta_{R} \eta_{1} \eta_{2} \mu (1-\epsilon) \left(R_{0}^{(S)}-1 \right)}{A_{6} \alpha_{R} (A_{5}+c k_{R})\left( A_{1} A_{2} A_{3} \left(R_{0}^{(SR)}-1 \right)+\gamma k_{S} \beta_{S} \eta_{1} (1-\epsilon)\right)}.
\end{align*}
The co-existence endemic equilibrium point remains in the biologically feasible region if and only if $R_{0}^{(S)}>1$ and $R_{0}^{(SR)}>1$. These existence conditions suggest that coexistence is feasible only when the drug-sensitive strain maintains dominance in overall incidence relative to the drug-resistant strain.

Now, in order to determine the local asymptotic stability of the equilibrium point $\displaystyle{E^{(*)}}$, we compute the Jacobian matrix $\displaystyle{J\left(E^{(*)}\right)}$. The corresponding characteristic equation of this Jacobian takes the form
\begin{equation}{\label{03:eqn:characteristic_equation_2}}
    p_{*}(x)=(x+A_{6})q_{*}(x)=0,
\end{equation}
where $q_{*}(x)$ is a sixth-degree polynomial given by:
$$q_{*}(x)=x^6 + a_5 x^5 + a_4 x^4 + a_3 x^3 + a_2 x^2 + a_1 x + a_0.$$
The factor $(x+A_{6})$ indicates an eigenvalue as $-A_6 ~(<0)$, contributing to stability of $E^{(*)}$. The remaining six eigenvalues of $\displaystyle{J\left(E^{(*)}\right)}$ correspond to the roots of the polynomial $q_{*}(x)$. However, analytical investigation to determine these roots and derive explicit conditions for local stability is intractable, as the extraction of the coefficients $a_0, \ldots, a_5$ is infeasible due to the highly complex and nested structure of the characteristic polynomial. Therefore, we will investigate the stability of the equilibrium point $E^{(*)}$ numerically for a specific parameter set using the Routh-Hurwitz criterion \cite{yang2002} in Section \ref{03:sec:numerical_simulation}. 

The Routh array for the polynomial \( q_{*}(x) = x^6 + a_5 x^5 + a_4 x^4 + a_3 x^3 + a_2 x^2 + a_1 x + a_0 \) is:
\[
\begin{array}{c|cccc}
s^6 & 1 & a_4 & a_2 & a_0 \\
s^5 & a_5 & a_3 & a_1 & 0 \\
s^4 & b_1 & b_2 & b_3 & 0 \\
s^3 & c_1 & c_2 & 0 & 0 \\
s^2 & d_1 & d_2 & 0 & 0 \\
s^1 & e_1 & 0 & 0 & 0 \\
s^0 & f_1 & 0 & 0 & 0 \\
\end{array}
\]
where the elements are defined as:
\begin{align*}
b_1 &= \frac{ a_5 a_4 - a_3 }{ a_5 }, \quad
b_2 = \frac{ a_5 a_2 - a_1 }{ a_5 }, \quad
b_3 = a_0,  \quad
c_1 = \frac{ b_1 a_3 - a_5 b_2 }{ b_1 },  \quad
c_2 = \frac{ b_1 a_1 - a_5 b_3 }{ b_1 }, \\
d_1 &= \frac{ c_1 b_2 - b_1 c_2 }{ c_1 }, \quad
d_2 = a_0, \quad
e_1 = \frac{ d_1 c_2 - c_1 a_0 }{ d_1 },\quad
f_1 = a_0.
\end{align*}
According to the Routh-Hurwitz criterion, the equilibrium point is locally asymptotically stable if and only if all elements in the first column of the Routh array are positive. Therefore, the local stability of $E^{(*)}$ holds under the condition:
\begin{equation}{\label{03:eqn:stability_condition_2}}
    a_5 > 0, \quad b_1 > 0, \quad c_1 > 0, \quad d_1 > 0, \quad e_1 > 0, \quad a_0 > 0
\end{equation}
Based on these results, we we propose the following result about the existence and local stability of the co-existence endemic equilibrium point $E^{(*)}$.
\begin{theorem}{\label{03:thm:eqilibrium_point_2}}
For the system (\ref{3_system_reduced}), the co-existence endemic equilibrium point $E^{(*)}$ exists if and only if $\displaystyle{R_{0}^{(R)}>1}$ and $\displaystyle{R_{0}^{(SR)} > 1}$. Further, $E^{(*)}$ is locally asymptotically stable under the parametric conditions provided in (\ref{03:eqn:stability_condition_2}), otherwise it is unstable.
\end{theorem}

The existence and stability conditions of these equilibrium points highlight the critical role of certain parameters in deciding the long-term dynamics of the system. For example,  a lower recruitment rate $\lambda$ reduces the influx of susceptible individuals, leading to a decline in disease burden or even complete elimination over the long time horizon. Similarly, increasing awareness among diagnosed individuals (i.e., reducing $c$) significantly contributes to disease control. The relative transmission dynamics of the two strains also plays a key role. Relatively higher transmission rates of the drug-sensitive strain or lower rates for the drug-resistant strain can create a possibility for the coexistence of sensitive and resistant strains, even at a higher treatment coverage. The effects of other parameters involve complex inter-dependencies, which will be explored in detail through a comprehensive sensitivity analysis in Section \ref{03:sec:sensitivity_analysis}.

\begin{remark}
The existence and stability conditions of the equilibrium points $E^{(0)}, E^{(1)}$ and $E^{(*)}$ indicate that the system does not admit multiple regions of asymptotic stability, and no two stable equilibrium states can coexist. This implies the absence of bistability in the system.
\end{remark}

\section{Epidemiological parameters and initial state values}{\label{03:sec:parameter_estimation}}
In this Section, we compute the epidemiological parameters of the system (\ref{3_system_main}), using various studies, including meta-analysis, reviews, and  reports. To increase the applicability of the proposed model across different scenarios, we incorporate studies that account for diverse population groups based on factors such as risk behavior, geographic location, and economic status etc. Using these studies, we determine the baseline parameter values along with their ranges of variation. Further, we estimate the initial conditions for each state variable in the system (\ref{3_system_main}), which will be used for most of numerical simulations. 

In \cite{poon2022}, the recruitment rate for India is estimated at 25.31 million per year, calculated based on factors such as births, immigration, emigration, and other processes contributing to new addition of individuals into the sexually active population. The transmission rate $\alpha_{S}$ $(\alpha_{R})$ is the rate of infection per effective contact between susceptible and drug-sensitive (drug-resistant) infected populations. Typically, the infection rate is defined by the change in the infected population resulting from interactions between one unit of both susceptible and infected individuals over a unit of time. For the drug-sensitive infected population, we consider the transmission rate $\alpha_{S}=0.000047134$ million$^{-1}$ year$^{-1}$ \cite{poon2022}. A mutation in the wild strain of a virus can either increase or decrease its virulence \cite{ande1991}, making it essential to consider both scenarios where the transmission rate of the drug-resistant infected population is either higher or lower than that of the drug-sensitive infected population. Therefore, we vary the transmission rate associated with drug-resistant strain within an appropriate range of 0.00001 to 0.00007 million$^{-1}$ year$^{-1}$, setting the baseline value at $\alpha_{R}=0.00002$ million$^{-1}$ year$^{-1}$. A recent systematic review and meta-analysis study of high- and upper-middle-income countries estimated that people live with undiagnosed HIV infection for an average of 3 years before receiving a diagnosis \cite{gbad2022}. However, despite improved HIV testing facilities, the average time from infection to diagnosis can still range widely, from 0.69 to 10.15 years, depending on different interacting groups and regional factors \cite{wand2009, hall2015, dail2017}. Therefore, we set the baseline value for the diagnosis rate of infection in drug-sensitive infected population as $k_{S}=0.33$ year$^{-1}$. To assess various scenarios influenced by $k_{S}$, we can adjust this parameter within a range of 0.01 to 1.45 per year. For the undiagnosed drug-resistant infected population, we assume that the diagnosis of their resistant status takes less time as compared to the diagnosis of the drug-sensitive infected population, as they are already under medical supervision. Hence, we assume that the diagnosis rate of infection in drug-sensitive infected population as $k_{R}=$5 year$^{-1}$, with a range of variation as 1 to 15 per year. This corresponds to the undiagnosed period for drug-resistant strain, which varies from approximately 3 weeks to 1 year. Undiagnosed individuals are estimated to transmit HIV at a rate three to seven times higher than diagnosed individuals \cite{hall2012}. This variation depends on factors such as the number of at-risk sexual partners, retention in treatment, and viral suppression among diagnosed individuals. Thus, we choose $c=0.25$ as the baseline value, indicating a 75\% reduction in HIV transmissibility after diagnosis. To account for all possible scenarios, we allow $c$ to vary between 0 and 1.

Further, in the absence of treatment, the natural progression of HIV infection to AIDS usually takes about $8 \text{ to }10$ years \cite{poor2016}. Accordingly, we assume a mean conversion time of $1 / \theta_{1}=10$ years (i.e., $\theta_{1}=0.1$ year$^{-1}$) for an untreated HIV infection case into the AIDS stage. In contrast, when treatment is provided, an adherent patient advances to the AIDS stage far more slowly than a non-adherent or untreated patient. ART significantly increases the life expectancy of AIDS patients, with estimates ranging from 29 to 37.3 years, depending on the age of the patient at which treatment begins \cite{teer2017}. Therefore, we choose a suitable parameter range of 0.025 to 0.05 per year for $\theta_{2}$, representing an increase of 20 to 40 years in life expectancy for HIV-infected patients as a result of ART. The parameters $\beta_{S}$ and $\beta_{R}$ are associated with the expected time taken to initiate appropriate treatment after the diagnosis of infection for drug-sensitive and drug-resistant individuals, respectively. A systematic review and meta-analysis suggest that this duration can vary from 5 to 40 days, influenced by factors such as age and gender of patient, and the economic conditions of the country \cite{tao2023}. Another study indicates a broader range, from 20 to 108 days \cite{nico2021}. Consequently, we set these parameters between 4 and 20 per year. A study has shown that the average time for the development of drug resistance in non-adherent patients ($1/\gamma$) is around 6 months after the initiation of treatment \cite{stad2013}. Based on this, we vary the parameter $\gamma$ between 0.5 and 2 per year, considering different levels of sub-optimality in adherence to the treatment. The United Nations World Population Prospects 2024 \cite{unit2024} reports that the crude death rate (deaths per 1,000 people per year) ranges from 6.9 to 9.3 across different geographic regions, depending on economic status. Based on this, we assume the natural death rate $(\mu)$ can vary between 0.0069 and 0.0093 per year. Regarding the disease induced death rate, the survival probability of a patient declines from 0.48 to 0.26 and further to 0.18 over 2, 4, and 6 years, respectively, following the onset of AIDS \cite{poor2016}. Accordingly, we set the parameter $\mu_{d}$ to lie within the range of 0.16 and 1 per year, representing a potential time frame of 1 to 6 years for patient mortality after AIDS onset. 

By the end of 2023, ART coverage among the HIV infected population ranged from $49\%$ in the Middle East and North Africa to $83\%$ in Eastern and Southern Africa, with a global average of $77\%$ \cite{fact2024}. Therefore, we consider the baseline values for the parameters $\eta_{1}$ and $\eta_{2}$ as 0.77 and 0.9, respectively.
In a recent study, authors have found that only $53\%$ population is optimally adherent to ART, based on the assumption that the proportion of days covered (PDC) for the study group is greater than 0.95 \cite{li2024}. Furthermore, if we assume a PDC of 0.90 as the threshold for developing drug resistance in non-adherent patients, approximately $34\%$ of patients are non-adherent to ART. Thus, we set the parameter $\epsilon=0.66$. For the initial values of each state variable, 

We set the initial population sizes as follows: 625 million for the susceptible group, 2.2 million for drug-sensitive infected individuals, 1.5 million for those under treatment, and 1 million for drug-resistant infected individuals, following \cite{poon2022}. A recent study indicates that approximately one-third of infected individuals are unaware of their infection status in underdeveloped and developing countries \cite{mahm2024}. Based on this, we allocate 0.7 million to the drug-sensitive undiagnosed group, 1.5 million to the drug-sensitive diagnosed group, 0.3 million to the drug-resistant undiagnosed group, and 0.7 million to the drug-resistant diagnosed group. Additionally, among the 1.5 million individuals under treatment, we assume that 1.2 million are receiving first-line treatment, while 0.3 million are receiving the second-line treatment. Since individuals who progress to the AIDS stage have a short survival period, the population in this stage remains relatively small, estimated at around 0.25 million, as considered in \cite{poon2024}. However, we also vary these initial conditions in our numerical simulations to explore different possible scenarios in dynamics of the proposed model. For each parameter, Table \ref{03:tab:parameter_values} provides the unit, baseline value, range of variation, and corresponding references.

\begin{table}[htbp]
\centering
\begin{tabular}{@{}lllll@{}}
\toprule
\textbf{Parameter} & \textbf{Unit} & \textbf{Baseline Value} & \textbf{Range of Variation} & \textbf{Reference(s)} \\
\midrule
$\lambda$      & million year$^{-1}$ & 25.31 &    [15,\ 30]             & \cite{poon2022} \\
$\alpha_S$     & million$^{-1}$ year$^{-1}$ & $4.7134 \times 10^{-5}$ & $[4 \times 10^{-5},\ 3 \times 10^{-4}]$ & \cite{poon2022} \\
$\alpha_R$     & million$^{-1}$ year$^{-1}$ & $2.0 \times 10^{-5}$ & $[1 \times 10^{-5},\ 3 \times 10^{-3}]$ & \cite{poon2022} \\
$k_{S}$        & year$^{-1}$ & 0.33 & [0.01,\ 1.45]           & \cite{hall2015, dail2017, gbad2022} \\
$k_{R}$        & year$^{-1}$ & 5.00 & [1,\ 15]                & \cite{hall2015, dail2017, gbad2022} \\
$c$            & unitless     & 0.25 & [0,\ 1]                 & \cite{hall2012} \\
$\theta_{1}$   & year$^{-1}$ & 0.10 & [0.05,\ 0.5]            & \cite{poor2016} \\
$\theta_{2}$   & year$^{-1}$ & 0.03 & [0.025,\ 0.05]          & \cite{teer2017} \\
$\beta_S$      & year$^{-1}$ & 10.00 & [4,\ 20]               & \cite{nico2021, tao2023} \\
$\beta_R$      & year$^{-1}$ & 15.00 & [4,\ 20]               & \cite{nico2021, tao2023} \\
$\gamma$       & year$^{-1}$ & 1.00 & [0.5,\ 2]               & \cite{stad2013} \\
$\eta_1$       & unitless     & 0.77 & [0,\ 1]                 & \cite{fact2024} \\
$\eta_2$       & unitless     & 0.90 & [0,\ 1]                 & Assumed \\
$\epsilon$     & unitless     & 0.66 & [0,\ 1]                 & \cite{li2024} \\
$\mu$          & year$^{-1}$ & 0.007 & [0.0069,\ 0.0093]      & \cite{unit2024} \\
$\mu_d$        & year$^{-1}$ & 0.50 & [0.16,\ 1]              & \cite{poor2016} \\
\bottomrule
\end{tabular}
\caption{Parameter values used in the system (\ref{3_system_main}) along with their units, baseline values, range of variation, and references.}
\label{03:tab:parameter_values}
\end{table}

\section{Sensitivity analysis}{\label{03:sec:sensitivity_analysis}}
In the model (\ref{3_system_main}), the state variables have a complex and non-linear relationship with the epidemiological parameters. The uncertainty in the model inputs and outputs may often attributes to the ambiguous behavior of the model. In such situations, the sensitivity analysis is a tool that identifies the role of uncertainties in the input factors of the model on the variations in different model outputs \cite{pian2016}. The input factors of the model are quantities such as initial values of state variables, model parameters, model parametrization, etc. which can be controlled before the model execution. On the other hand, model outputs, often known as quantities of interest (QoI), refer to variables dependent on model inputs. In epidemiological models, the QoI can be time dependent, such as growth rate of infection, values of state variables or cummulative infections over time, and time independent, such as the basic reproduction numbers, peak infection time and magnitude. Since the model (\ref{3_system_main}) does not exhibit bi-stable behavior, initial conditions have an insignificant influence on time-dependent QoIs in the longer horizon. However, a small variation might be observed in such model outputs at an early phase for different initial conditions. Therefore, we consider epidemiological parameters as uncertain input quantity for sensitivity analysis. 

The sensitivity analysis methods are generally divided into two categories based on the approach used to explore the input space: Local methods, which evaluate the influence of inputs at a single point, and Global methods, which examine sensitivity at multiple points simultaneously. In this section, we aim to address the following questions: (A). What is the local as well as global impact of each epidemiological parameter on various QoIs? (B). Is there any difference in the contribution of these parameters in the uncertainties of the short- and long-term dynamics of the model (\ref{3_system_main})? To thoroughly answer these questions, we conducted both local and global sensitivity analysis, allowing for a detailed evaluation of parameter influence on the dynamics of the system.

\subsection{Local sensitivity analysis}{\label{03:subsec:local_sensitivity_analysis}}
The initial transmission of drug-sensitive and drug-resistant strains of HIV is largely governed by their respective basic reproduction numbers \cite{chit2008}. In addition, we have observed the importance of the ratio of these basic reproduction numbers in the existence and stability of different equilibrium points in section \ref{03:sec:equilibrium_points}. Therefore, we choose the basic reproduction number of both strains and their ratio as our QoIs to assess the local impact of each parameter on the short-term dynamics of the model (\ref{3_system_main}). Although local sensitivity analysis (LSA) does not provide any information about the interactions between parameters and their combined effect on model output, it works as an initial screening tool to identify the most influential parameters. The normalized forward sensitivity index is a commonly used local sensitivity measure that quantifies the relationship between relative changes in a parameter and the corresponding variations in a variable. If the model output $(u)$ is a differentiable function of the model parameter $(\alpha)$, the normalized forward sensitivity index $SI_{\alpha}^{u}$ is defined as:
$$SI_{\alpha}^{u}:=\frac{\alpha}{u}\times \dfrac{\partial u}{\partial \alpha}$$

We calculated the local sensitivity index of $R_{0}^{(S)}$, $R_{0}^{(R)}$ and $R_{0}^{(SR)}$ corresponding to each parameter using their baseline values from Table \ref{03:tab:parameter_values}. These indices are presented in Table \ref{03:tab:local_sensitivity}. Each index represents the ratio of the percentage change in the model output in response to a given percentage change in the model input. The negative sign shows an inverse relationship between the model input and the output. The highlighted values in bold indicate the significant parameters with a sensitivity index magnitude greater than 0.5. Note that the parameters not listed in Table \ref{03:tab:local_sensitivity} have no impact on any of the given QoIs. From this local sensitivity analysis, we observe that $\lambda$ and $\alpha_{S}$ are the most influential parameters in positively driving $R_{0}^{(S)}$, while $\lambda$ and $\alpha_{R}$ play a similar role in influencing $R_{0}^{(R)}$. In contrast, $R_{0}^{(S)}$ is most negatively sensitive to $\mu$ and $k_{S}$, whereas $\mu$ and $k_{R}$ have the most negative impact on $R_{0}^{(R)}$.
Except for $\lambda$, every listed parameter influences $R_{0}^{(SR)}$. Among them, $\alpha_{S}$ and $k_{R}$ have the strongest positive impact, while $\alpha_{R}$ and $k_{S}$ show a negative influence, highlighting a competitive interaction between these two strains. 

\begin{table}[htbp]
    \centering
    \begin{adjustbox}{max width=\textwidth}
    \begin{tabular}{@{}lccccccccccccc@{}}
        \toprule
        \textbf{QoI} & $\lambda$ & $\alpha_S$ & $\alpha_R$ & $c$ & $k_S$ & $k_R$ & $\theta_{1}$ & $\beta_S$ & $\beta_R$ & $\eta_{1}$ & $\eta_{2}$ & $\mu$ \\
        \midrule
        $R_{0}^{(S)}$  & \textbf{1.000} & \textbf{1.000} & --     & 0.011   & \textbf{-0.745} & --     & -0.229   & -0.011   & --     & -0.010   & --     & \textbf{-1.016} \\
        $R_{0}^{(R)}$  & \textbf{1.000} & --     & \textbf{1.000} & 0.085   & --     & \textbf{-0.894} & -0.020   & --     & -0.085   & --     & -0.084   & \textbf{-1.001} \\
        $R_{0}^{(SR)}$ & --     & \textbf{1.000} & \textbf{-0.952} & -0.074  & \textbf{-0.745} & \textbf{0.894} & -0.209   & -0.011   & 0.085   & -0.010   & 0.084   & -0.015 \\
        \bottomrule
    \end{tabular}
    \end{adjustbox}
    \caption{Normalized forward sensitivity indices of $R_{0}^{(S)}$, $R_{0}^{(R)}$, and $R_{0}^{(SR)}$ with respect to key model parameters. Bold values indicate indices with absolute values greater than 0.5, highlighting parameters with substantial influence on the respective QoIs.}
    \label{03:tab:local_sensitivity}
\end{table}

Although LSA gives an initial screening of some most influential parameters at their baseline values and has low computational cost, it is best suited for linear models where these sensitivity results can be extrapolated across the entire parameter space. For a non-linear model, where state variables and model parameters may interact with each other, LSA may produce misleading conclusions \cite{qian2020}. Therefore, we further extend our analysis by conducting a detailed global sensitivity analysis of the proposed nonlinear model (\ref{3_system_main}) to assess the impact of interactions among state variables and model parameters on different model outputs. 

\subsection{Global sensitivity analysis}{\label{03:subsec:global_sensitivity_analysis}}
Global sensitivity analysis (GSA) methods typically determine sensitivity indices by applying an appropriate averaging method to model output values for multiple input points in the parameter space. This single averaged value quantifies the overall impact of the input factor on output uncertainty. These indices are often not computable analytically and are instead estimated using numerical approximations based on sampling-based methods. Various sampling procedures are available in the literature for generating input samples across the parameter space, including random sampling \cite{fish2013}, Latin hypercube sampling (LHS) \cite{helt2003}, and quasi-random sampling \cite{cafl1998} etc. The selection of an appropriate sampling technique depends on the specific requirements of the GSA method. In this study, we apply two different GSA methods based on their respective use cases for different QoIs.

\subsubsection{Partial rank correlation coefficient method}{\label{03:subsubsec:PRCC}}
The partial rank correlation coefficient (PRCC) is a correlation and regression based GSA method that quantifies the monotonic relationship between the model parameters and outputs. PRCC efficiently captures monotonic dependencies, including underlying nonlinear relationships, by transforming data into rank-ordered form and removing linear effects of other parameters. A rank-based transformation ensures the robustness of the results, particularly when nonlinear relationships exist between input parameters and outputs. The PRCC gives the sensitivity measure as the linear correlation between the residuals $(\mathbf{X_{j}}-\hat{X}_{j})$ and $(\mathbf{Y}-\hat{Y})$. Here, $\mathbf{X_{j}}$ is the rank-transformed $j^{th}$ input parameter and $\mathbf{Y}$ denotes the rank-transformed output state variable. The terms $\hat{X}_{j}$ and $\hat{Y}$ are derived from following linear regression models based on $k$ samples:
$$\hat{X}_{j}=C_{0}+\sum_{i=1,i\neq j}^{k}C_i \mathbf{X_i}, \qquad \text{and} \qquad  \hat{Y}=B_{0}+\sum_{i=1,i\neq j}^{k}B_i \mathbf{Y_i}.$$ 
A PRCC index of $1$ or $-1$ indicates that the output is a perfectly monotonic increasing or decreasing function of the inputs, respectively. In contrast, an index of 0 suggests that the output fluctuates between increasing and decreasing at successive data points.

Similar to Section \ref{03:subsec:local_sensitivity_analysis}, we select $R_{0}^{(S)}$, $R_{0}^{(R)}$ and $R_{0}^{(SR)}$ as QoIs to access the impact of parameters on early dynamics of the model. Each of these functions is strictly monotonic with respect to all parameters within the specified range of variation given in Table \ref{03:tab:parameter_values}, a necessary condition to ensure the robustness of the PRCC method. To generate the sample space, we use LHS with $N=5000$, which is more efficient than random sampling (see \ref{03:appendix:sec:convergence_analysis}). Here, $N$ represents the number of samples from the feasible range of parameters. LHS requires fewer simulations because it achieves a faster convergence rate. As we do not have any information on the distribution of parameters, we assume a uniform distribution for each parameter. For the generated sample space, the QoIs are first computed and then converted into their rank values along with the parameters. The PRCC values for each parameter are subsequently determined using the procedure described in \cite{mari2008}. These values, corresponding to different QoIs, are displayed in Figure \ref{03:fig:gsa_prcc_heatmap} using a heatmap representation. In Figure \ref{03:fig:gsa_prcc_heatmap}, the color intensity of each block quantifies the influence of a parameter on the corresponding output. The parameters $\alpha_{S} \text{ and } \lambda$ exhibit the strongest positive influence on $R_{0}^{(S)}$, whereas $k_{S}, \theta_{1}, \text{ and } \mu$ show a significant negative impact on $R_{0}^{(S)}$. Similarly, variations in parameters $\alpha_{R}, c \text{ and } \lambda$ lead to a directly proportional change in $R_{0}^{(R)}$, while parameters $k_{R}$ and $\eta_{2}$ influence it indirectly. Also, $R_{0}^{(SR)}$ exhibits a direct dependence on $\alpha_{S}, k_{R} \text{ and } \eta_{2}$, and $\alpha_{R}, k_{S} \text{ and } c$ contribute inversely to its variations.

The sensitivity analysis is meaningful only with an assessment of its sampling variability. Quantifying the uncertainty in output values over a given parameter range is a computationally expensive task, as it involves multiple iterations of the entire process. Therefore, we applied the bootstrapping technique for resampling to access the confidence interval (CI) for the PRCC values efficiently. We used the percentile method for constructing CIs to eliminate the effect of skewness in bootstrap distributions, as moment-based methods may lead to poor estimates in such cases \cite{arch1997}. However,  the percentile method requires a larger number of bootstrap resamples to achieve reliable estimates. As suggested in \cite{arch1997}, we performed 2000 bootstrap iterations by random resampling with replacement from the original dataset of 5000 points. For each bootstrap sample, we computed the PRCC value for all three outputs and constructed 95\% CIs using the percentile approach. Figure \ref{03:fig:gsa_prcc_barplot} presents the mean PRCC values for each parameter along with their confidence intervals for different outputs. Parameters with confidence intervals containing $0$ are considered to be non-significant. To justify the selection of number of sample points and bootstrap resamples, a detailed convergence test for $3$ parameters with highest PRCC values (magnitude) for each output is presented in Figure \ref{03:fig:gsa_prcc_convergence} (see \ref{03:appendix:sec:convergence_analysis}). 

\begin{remark}
PRCC method eliminates the effect of any hidden correlation between parameters by partialling out their shared variance that might arise due to randomness in the sampling procedure. However, certain parameters that do not appear in the expressions for these reproduction numbers still have nonzero PRCC values, though their influence remains negligible.
\end{remark}

\begin{figure}
\centering
\includegraphics[width=1\linewidth]{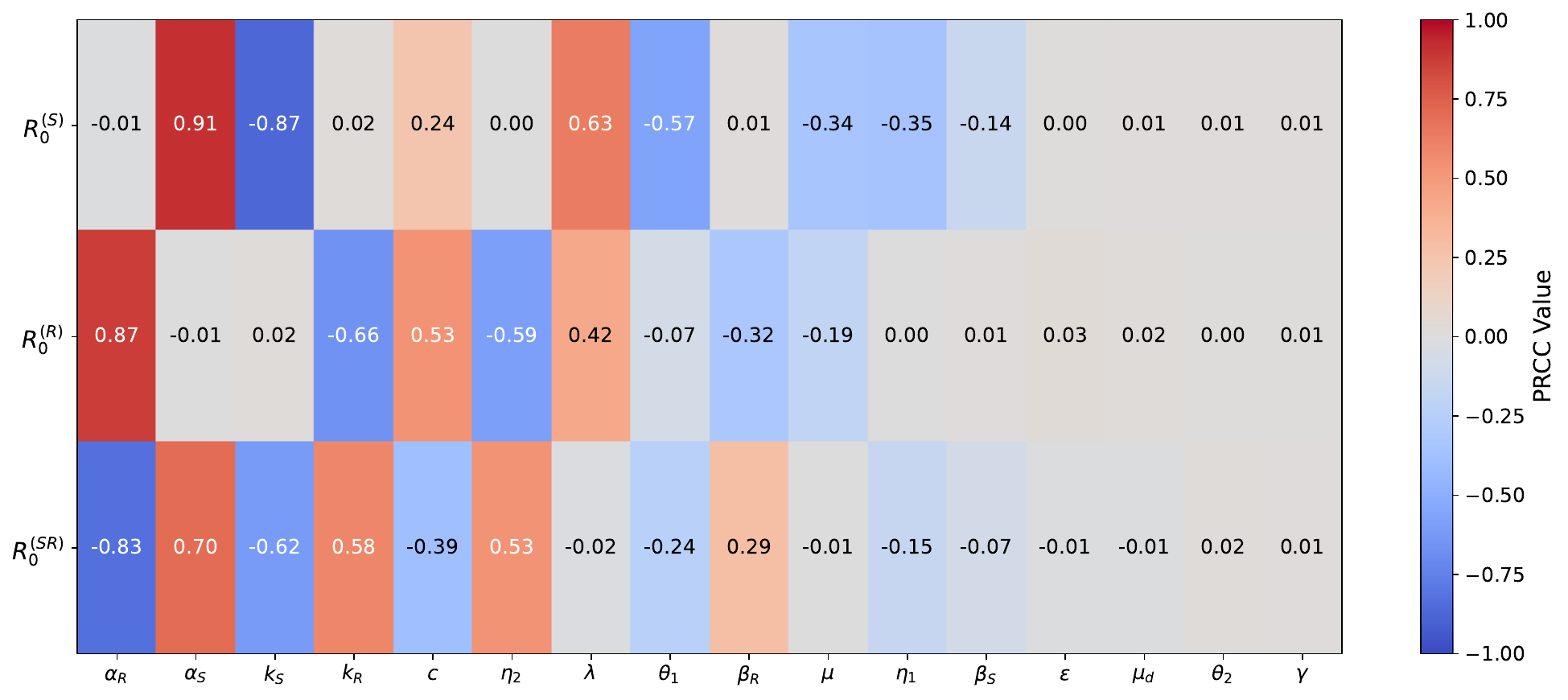}
\caption{Heatmap illustrating  PRCC values computed using 5,000 parameter sets generated using LHS for parameters influencing the basic reproduction numbers of sensitive strain $(R_{0}^{(S)})$, resistant strain $(R_{0}^{(R)})$, and their ratio $(R_{0}^{(SR)})$. The color intensity represents the strength of sensitivity, with red indicating positive correlation  and blue indicating negative correlation. Parameters are arranged in descending order of mean absolute PRCC values to highlight their relative contribution to early model dynamics.}
\label{03:fig:gsa_prcc_heatmap}
\end{figure}

\begin{figure}
\centering
\includegraphics[width=1\linewidth]{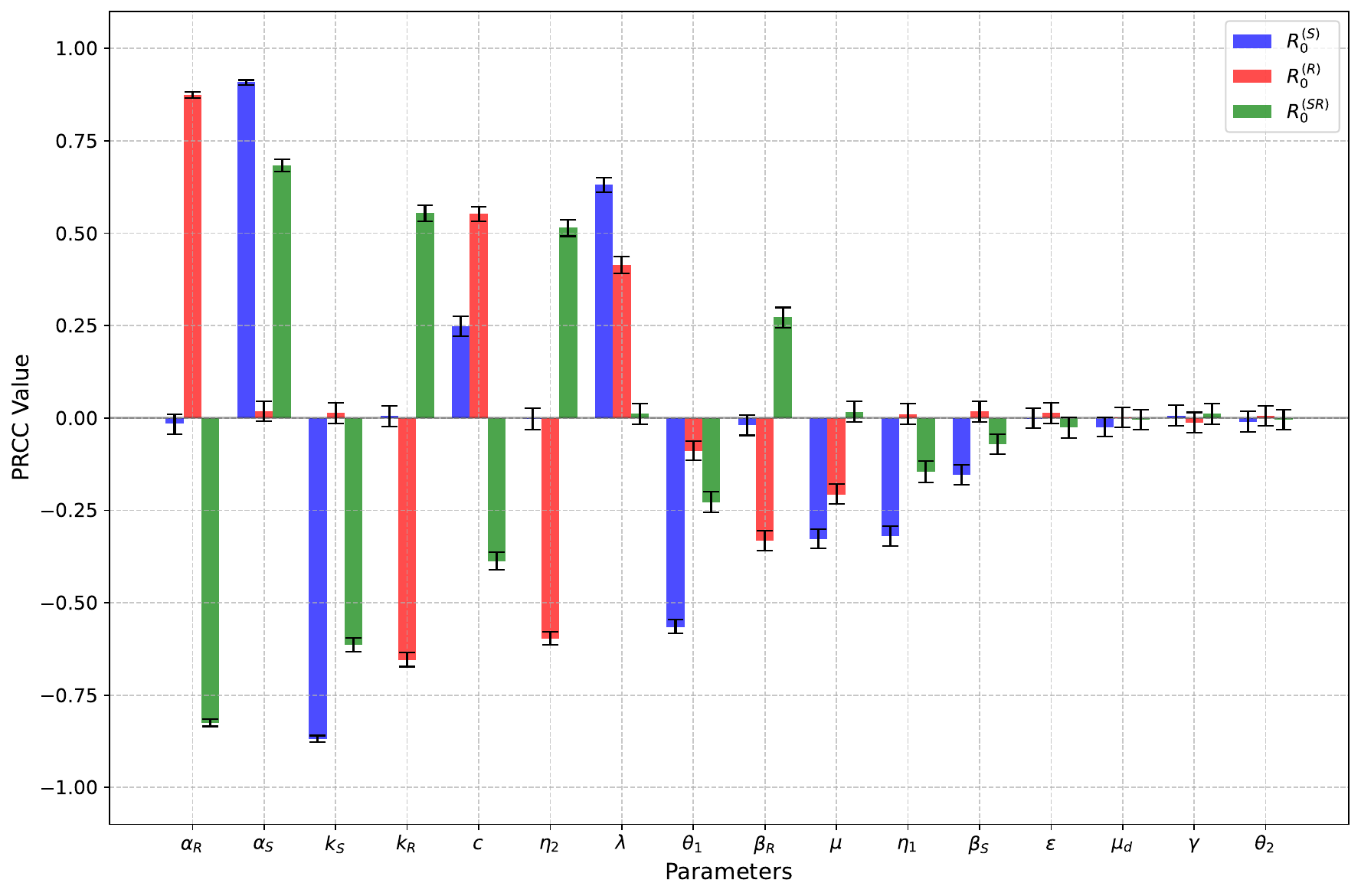}
\caption{Bar plot showing mean PRCC values for parameters influencing the basic reproduction numbers of sensitive strain $(R_{0}^{(S)})$ in blue, resistant strain $(R_{0}^{(R)})$ in red, and their ration $(R_{0}^{(SR)})$ in green, derived from 2000 bootstrap resamples of 5000 original Latin Hypercube samples. Error bars indicate 95\% confidence intervals. Parameters are arranged in descending order of mean absolute PRCC values to highlight their relative contribution to early model dynamics.}
\label{03:fig:gsa_prcc_barplot}
\end{figure}

To examine the long-term impact of these parameters on the model dynamics, we consider the state variable values at a specific time as the quantity of interest. However, directly applying the PRCC method for global sensitivity analysis may lead to inaccurate results, as the state variables do not necessarily exhibit a monotonic relationship with these parameters. In such cases, variance-based sensitivity analysis methods are more appropriate, providing more reliable results with better interpretability. We apply the Sobol method for GSA, which characterizes the total variance in output values into contributions from individual parameters and their interactions.

\subsubsection{The Sobol method}{\label{03:subsubsec:Sobol}}
The Sobol method is a variance-based approach for GSA that quantifies the contribution of individual parameters or a group of parameters to the overall variability of output variables in a non-linear model \cite{sobo2001}. However, this method does not provide any insight into the direction of this impact of the parameters. The Sobol indices are derived from the conditional variance. Given $Y$ as the model output,  the first-order Sobol index for an input parameter $X_{i}$, measuring the direct contribution of $i^{th}$ input to the total variance, is defined as:
$$S_{i}=\frac{\mathbb{V}_{X_{i}}[\mathbb{E}_{X_{\sim i}}[Y|X_{i}]]}{\mathbb{V}[Y]},$$
where $\mathbb{E}$ and $\mathbb{V}$ denote the expectation and variance operators, respectively, and $X_{\sim i}$ represents all parameters except $X_{i}$ \cite{andr2010}. In the numerator, the expectation of $Y$ is evaluated considering all possible values of $X_{\sim i}$ while keeping $X_{i}$ fixed, and the outer variance is computed over all possible values of $X_{i}$.
Another key sensitivity measure is the total-order Sobol index, which captures the overall influence of a parameter $X_{i}$ by incorporating both its direct contribution and indirect effects arising from interactions with other parameters. The total order Sobol index for input parameter $X_{i}$ is defined as \cite{andr2010}: 
$$S_{T_{i}}=\frac{\mathbb{E}_{X_{\sim i}}[\mathbb{V}_{X_{i}}[Y|X_{\sim i}]]}{\mathbb{V}[Y]}.$$
In this study, we computed only the first-order and total-order sensitivity indices, while higher-order indices were not considered. The total-order index is minimal for non-influential parameters, making it an effective tool for parameter screening, whereas the first-order index serves as a good ranking measure, particularly when parameter interactions are non-significant . 

We use the `SALib' library in Python \cite{herm2017} to generate the sample space and compute Sobol' sensitivity indices for different parameters. The sample space is generated using the Saltelli sampler  \cite{andr2010}, which is an extension of the Sobol' sequence \cite{sobo1976}. Sobol' sequence uses a quasi-random sampling algorithm that strategically distributes new samples to the sequence away from the earlier sampled points. A notable advantage of this sampling method is its superior uniformity compared to random or pseudo-random sampling, making it more suitable for accurately representing a uniform distribution of input parameters. To estimate the first- and total-order Sobol' indices, the Saltelli sampler generates $N \times (P+2)$ samples \cite{andr2010}, where $N$ represents the base sample size and $P$ denotes the number of input parameters. Also, it is recommended to choose the base sample size as an integer power of two to ensure optimal uniformity and better convergence \cite{owen2020}.

We define a time-dependent QoI as the total number of infected individuals ($I_{SU}+I_{SD}+I_{RU}+I_{RD}$) at a given time. The time interval is set from 0 to 100, capturing both the short-term and long-term dynamics of the model. To quantify the uncertainty in these time-varying sensitivity indices, the model must be evaluated $N \times (P+2)$ times per experiment, making this approach computationally expensive. On the other hand, applying standard bootstrap resampling to a single parameter set may produce inaccurate results due to initial sampling bias. This bias can significantly affect outcomes, particularly if the system (\ref{3_system_main}) has some bifurcation point within this single parameter set. Therefore, we adopted a two-stage approach that integrates ensemble averaging with bootstrap resampling to quantify uncertainty arising from both parameter variability and sampling bias. This approach addresses these challenges by first conducting multiple independent model evaluations $(n=50)$ with different Saltelli's sample sets (with a base sample size of 1024). This produces ensemble average values at different time points, representing more accurate sensitivity indices by eliminating the effects of single sampling bias. Further, we applied bootstrap resampling (1000 resamples) across these independent runs to generate 95\% confidence interval. Figure \ref{03:fig:gsa_sobol} presents a detailed illustration of the mean first- and total-order Sobol' sensitivity indices for each parameter, with the shaded region representing the 95\% confidence interval, corresponding to the total number of infections at a given time point. 

In Figure \ref{03:fig:gsa_sobol}, we observe that the influence of most parameters increases over time. The initially low sensitivity indices are expected, given the chronic nature of HIV and its prolonged infectious period, resulting in delayed parameter effects on model dynamics. Parameters such as $\theta_{2}, \beta_{S}, \gamma, \epsilon, \mu$ and $\mu_{d}$ have non-significant influence on the output. For most parameters, the large gap between first-order and total-order indices shows that their impact on the total number of infections is not just through their direct effect but significantly through their interactions with other parameters. Also, the direct impact of individual parameters becomes minimal after a certain time, making interactions between parameters the primary source of variance in the model output. Overall, a substantial shift in sensitivity is observed in the later phase compared to the early phase for most parameters.

\begin{figure}
\centering
\includegraphics[width=1\linewidth]{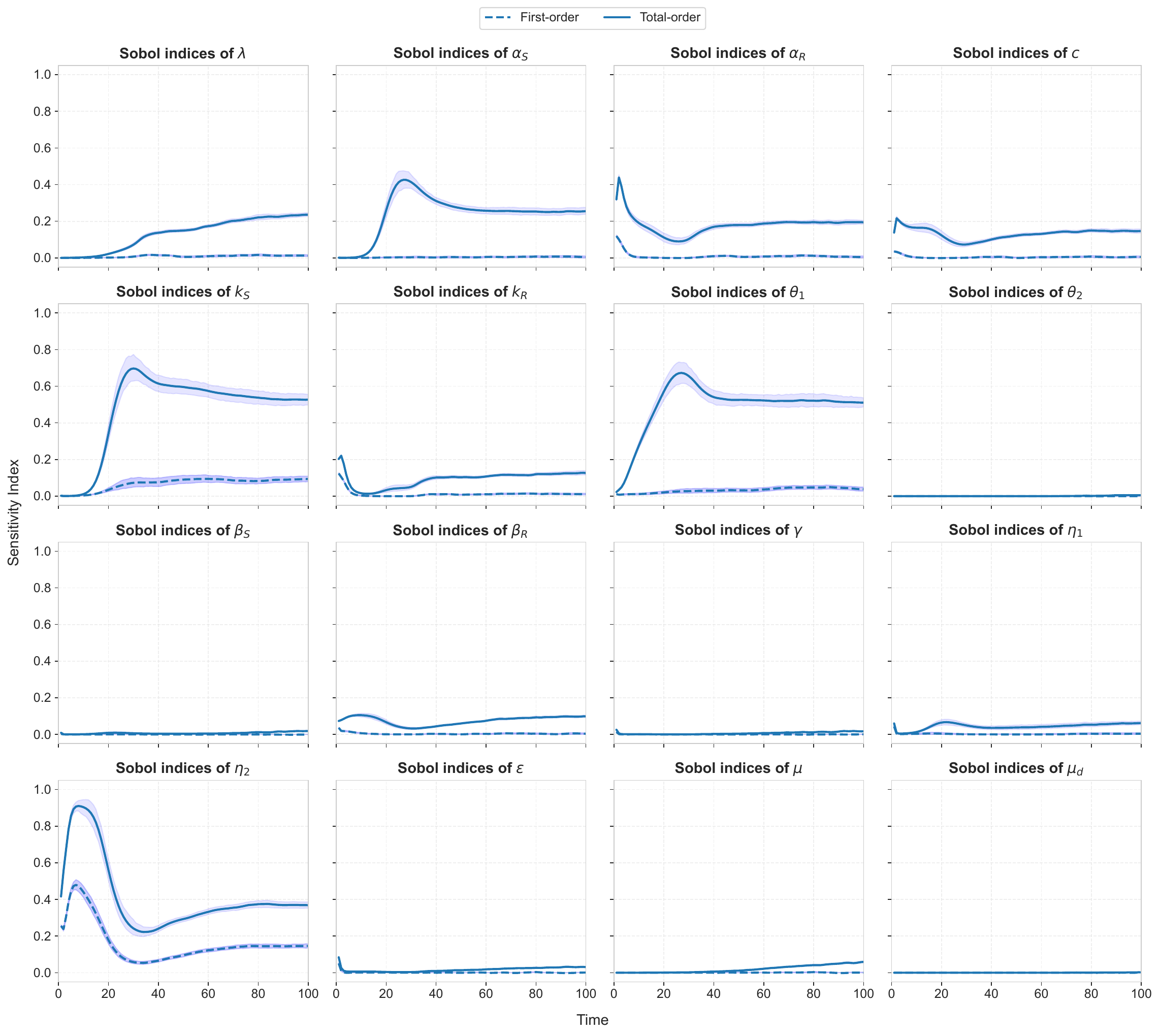}
\caption{Time-dependent global sensitivity analysis of the model (\ref{3_system_main}) using Sobol' indices. Each sub figure shows the first-order (dashed) and total-order (solid) sensitivity indices for each parameter corresponding to the total number of infected individuals. The shaded region represents the 95\% confidence interval.}
\label{03:fig:gsa_sobol}
\end{figure}

\section{Optimal control analysis}{\label{03:sec:optimal_control_analysis}}
In this Section, we explore time-dependent control mechanisms that focus on minimizing the infection burden while optimizing the cost of pharmaceutical interventions such as treatment, drug-adherence, and related factors.  In Section \ref{03:sec:sensitivity_analysis}, we observed that the dynamics of model (\ref{3_system_main}) is highly sensitive to the parameters $k_{S}$ and $k_{R}$, which are related to the diagnosis of infection. However, HIV diagnosis is not typically classified as a pharmaceutical intervention, it serves as a critical foundation for various interventions, both pharmaceutical and non-pharmaceutical. Therefore, we consider these diagnosis rates for drug-sensitive and drug-resistant infected populations, respectively, as controllable parameters. Furthermore, we consider the parameters $\eta_1$, $\eta_2$ and $\epsilon$ as treatment-related controllable measures. 

To construct the optimal control problem, we define the control variables $u_{1}(t), u_{2}(t), u_{3}(t), u_{4}(t)$ and $u_{5}(t)$ as bounded and Lebesgue-measurable functions over the interval $[0,t_{f}]$, where $t_{f}$ is the final time of interventions. The description of these control variables is as follows:\\
$u_{1}(t):$ Intensity of intervention efforts allocated to diagnosing drug-sensitive infected population at time $t$,\\
$u_{2}(t):$ Intensity of intervention efforts allocated to diagnosing drug-resistant infected population at time $t$,\\
$u_{3}(t):$ Intensity of the treatment intervention efforts to initialize first-line therapy among diagnosed drug-sensitive infected individuals at time $t$,\\
$u_{4}(t):$  Intensity of the treatment intervention efforts to initialize first-line therapy among diagnosed drug-sensitive infected individuals at time $t$,\\
$u_{5}(t):$ Intensity of supportive interventions implemented to improve drug-adherence among individuals receiving first-line therapy at time $t$.

Our primary objective is to determine the optimal control profile $(u^*_{1}(t), u^*_{2}(t), u^*_{3}(t), u^*_{4}(t), u^*_{5}(t))$ that minimizes the total associated cost. Note that the optimal allocation of resources is determined by the formulation of the cost function. To achieve the objective described above, we define the following cost functional to be minimized:
\begin{eqnarray}{\label{03:eqn:control_functional}}
J[\mathcal{U}(t)] &=& \int_{0}^{t_{f}} \Big[ W_{1} I_{SU}(t) + W_{2} I_{SD}(t) + W_{3} I_{RU}(t) + W_{4} I_{RD}(t) - W_{5} T_{1}(t) - W_{6} T_{2}(t) \nonumber \\ 
    && \quad + \frac{1}{2} \big( w_{1} u_{1}^{2}(t) + w_{2} u_{2}^{2}(t) + w_{3} u_{3}^{2}(t) + w_{4} u_{4}^{2}(t) + w_{5} u_{5}^{2}(t) \big) \Big] dt
\end{eqnarray}
where $\mathcal{U}(t)=\{u_{1}(t), u_{2}(t), u_{3}(t), u_{4}(t), u_{5}(t) : 0\leq u_{i}(t) \leq 1 \text{ for } i=1,2,3, 4,5; t\in [0,t_{f}]\}$ is the control set. The integrand $L(I_{SU}, I_{SD}, I_{RU}, I_{RD}, T_{1}, T_{2}, u_{1}(t), u_{2}(t), u_{3}(t), u_{4}(t), u_{5}(t)) $ represents the current cost at time $t$. The total infected population is defined as the sum of individuals in compartments $I_{SU}, I_{SD}, I_{RU}$ and $I_{RD}$. While individuals in compartments $T_1, T_2$ and $A$ are also infected, they do not contribute to the overall disease burden since they are not involved in further transmission. However, treatment-related terms are included in the objective functional to ensure a balanced optimization approach that simultaneously minimizes the untreated infected population while incentivizing treatment enrollment. The positive weight constants $W_{i}~ (i=1,2,3,4)$ represent the relative importance of reducing specific infected classes in controlling disease spread and minimizing its burden, $W_{5}$ and $W_{6}$ represent the relative importance of maximizing treatment coverage, particularly for vulnerable population with drug-resistant strain, and $w_{j}~(j=1,2,3,4,5)$ denote the relative costs of each control intervention, ensuring proper scaling and balance within the cost functional. The impact of efforts to improve HIV diagnosis, treatment availability, and drug adherence is expected to have a nonlinear relationship with the outcomes. Therefore, we use quadratic cost terms for each control variable to appropriately capture these effects, as in \cite{lenh2007, okos2013, agus2019}. We seek an optimal control solution $\mathcal{U}^{*}(u^*_{1}(t), u^*_{2}(t), u^*_{3}(t), u^*_{4}(t), u^*_{5}(t))$ such that $J(\mathcal{U}^*(t))=\min\{J(\mathcal{U}(t)\}$, subject to the system 

\begin{eqnarray}
S^{\prime}&=&\lambda-\alpha_S S(I_{SU}+c I_{SD})-\alpha_R S (I_{RU}+ c I_{RD}) -\mu S, \nonumber \\
I_{SU}^{\prime}&=&\alpha_S S (I_{SU}+c I_{SD})-k_{S_{max}}\mathbf{u_{1}(t)} I_{SU}-\theta_1 I_{SU}-\mu I_{SU},
\nonumber \\
I_{SD}^{\prime}&=& k_{S_{max}} \mathbf{u_{1}(t)} I_{SU}-\beta_{S} \mathbf{u_{3}(t)} I_{SD}-\theta_{1} (1-\mathbf{u_{3}(t)}) I_{SD}-\mu I_{SD}, 
\nonumber \\
T_{1}^{\prime}&=&\beta_{S} \mathbf{u_{3}(t)} I_{SD}-\gamma (1-\mathbf{u_{5}(t)}) T_{1}-\theta_{2} \mathbf{u_{5}(t)} T_{1}-\mu T_{1},
\nonumber \\
I_{RU}^{\prime}&=&\alpha_R S (I_{RU}+ c I_{RD})+\gamma (1-\mathbf{u_{5}(t)}) T_{1}-k_{R_{max}}\mathbf{u_{2}(t)} I_{RU}-\theta_{1} I_{RU}-\mu I_{RU},
\nonumber \\
I_{RD}^{\prime}&=& k_{R_{max}} \mathbf{u_{2}(t)} I_{RU}-\beta_{R} \mathbf{u_{4}(t)} I_{RD}-\theta_{1}(1-\mathbf{u_{4}(t)}) I_{RD}-\mu I_{RD},
\nonumber\\
T_{2}^{\prime}&=& \beta_{R} \mathbf{u_{4}(t)} I_{RD}-\theta_{2} T_{2}-\mu T_{2},
\nonumber \\
A^{\prime}&=& \theta_1 I_{SU}+\theta_{1} (1-\mathbf{u_{3}(t)}) I_{SD}+\theta_{2} \mathbf{u_{5}(t)} T_{1}+\theta_{1} I_{RU}+\theta_{1}(1-\mathbf{u_{4}(t)}) I_{RD}+\theta_{2} T_{2}-(\mu+\mu_d) A 
\label{03:system:control}
\end{eqnarray}
with non-negative initial conditions. Note that we have substituted the parameters $k_{S}, k_{R}, \eta_{1}, \eta_{2}$ and $\epsilon$ with their respective maximum values within the specified ranges to normalize the newly introduced control variables within the interval $[0,1].$

\subsection{Existence of an optimal control}{\label{03:subsec:existence_optimal_control}}
In this subsection, we establish the presence of optimal control functions that effectively minimize the cost functional within a finite time frame.

\begin{theorem}{\label{03:thm:existence_optimal_control}}
For the controlled state system (\ref{03:system:control}) with non-negative initial conditions, there exists an optimal control set $(u^*_{1}, u^*_{2}, u^*_{3}, u^*_{4}, u^*_{5})$ in $\mathcal{U}$ such that $J(u^*_{1}, u^*_{2}, u^*_{3}, u^*_{4}, u^*_{5})=\min\{J(u_{1}, u_{2}, u_{3}, u_{4}, u_{5}): (u_{1}, u_{2}, u_{3}, u_{4}, u_{5})\in \mathcal{U}\}.$
\end{theorem}
\begin{proof}
The existence of an optimal control is ensured if the following three conditions, as stated in Theorem 4.1 (Chapter III ) of \cite{flem1975}, are satisfied.
\begin{itemize}
        \item[i.] The solution set of the system (\ref{03:system:control}) with control functions in $\mathcal{U}$ is non-empty. 
        \item[ii.] The control set $\mathcal{U}$ is closed and convex, and the state system can be written as a linear function of the control variables, with coefficients that depend on both time and state variables.
        \item[iii.]  Integrand $L$ of Eq. (\ref{03:eqn:control_functional}) is convex on the control set $\mathcal{U}$. Furthermore, there exists a continuous function $g$ such that $L(I_{SU}, I_{SD}, I_{RU}, I_{RD}, T_{1}, T_{2}, u_{1}, u_{2}, u_{3}, u_{4}, u_{5}) \geq g(u_{1}, u_{2}, u_{3}, u_{4}, u_{5})$ where $|(u_{1}, u_{2}, u_{3}, u_{4}, u_{5})|^{-1} g(u_{1}, u_{2}, u_{3}, u_{4}, u_{5}) \rightarrow \infty$ whenever $|(u_{1}, u_{2}, u_{3}, u_{4}, u_{5})| \rightarrow \infty$. Here $|.|$ represents the Euclidean norm.
\end{itemize}
The controlled state system (\ref{03:system:control}) with non-negative initial conditions, along with the cost functional (\ref{03:eqn:control_functional}) defined on the control set $\mathcal{U}$ satisfies all three conditions (see \ref{03:appendix:sec:optimal_control_analysis} for a detailed proof), which shows the existence of an optimal control.
\end{proof}

\subsection{Characterization of optimal control variables}{\label{03:subsec:characterization_optimal_control}}
We next analytically characterize the trajectories of the optimal control variables for the control system (\ref{03:eqn:control_functional})-(\ref{03:system:control}). Also, we apply the Pontryagin’s Maximum Principle \cite{pont2018} to establish the necessary conditions for optimal control. This reformulates the problem (\ref{03:eqn:control_functional})-(\ref{03:system:control}) into a pointwise minimization of the Hamiltonian $H$ over the control variables $(u_{1}, u_{2}, u_{3}, u_{4}, u_{5})$. We define the Hamiltonian from the control system (\ref{03:system:control}) and the cost functional (\ref{03:eqn:control_functional}) as follows: 
\begin{eqnarray}
    H &=& L+ \sigma_{1} S^{\prime} + \sigma_{2} I_{SU}^{\prime} + \sigma_{3} I_{SD}^{\prime} + \sigma_{4} T_{1}^{\prime} +\sigma_{5} I_{RU}^{\prime} + \sigma_{6} I_{RD}^{\prime} + \sigma_{7} T_{2}^{\prime} + \sigma_{8} A^{\prime} \nonumber \\
    &=& \left( W_{1} I_{SU} + W_{2} I_{SD} + W_{3} I_{RU} + W_{4} I_{RD}  - W_{5} T_{1} - W_{6} T_{2}+ \frac{1}{2} \sum_{i=1}^{5}w_{i} u_{i}^{2} \right) + \sigma_{1} [\lambda-\alpha_S S(I_{SU}+c I_{SD}) \nonumber \\
    && -\alpha_R S (I_{RU}+ c I_{RD}) -\mu S] 
    + \sigma_{2} [\alpha_S S (I_{SU}+c I_{SD})-k_{S_{max}}u_{1}(t) I_{SU}-\theta_1 I_{SU}-\mu I_{SU}] \nonumber \\
    && + \sigma_{3} [k_{S_{max}}u_{1}(t) I_{SU}-\beta_{S} u_{3}(t) I_{SD}-\theta_{1} (1-u_{3}(t)) I_{SD}-\mu I_{SD}] 
    + \sigma_{4} [\beta_{S} u_{3}(t) I_{SD}-\gamma (1-u_{5}(t)) T_{1} \nonumber \\
    && -\theta_{2} u_{5}(t) T_{1}-\mu T_{1}] 
    + \sigma_{5}[\alpha_R S (I_{RU}+ c I_{RD})+\gamma (1-u_{5}(t)) T_{1}-k_{R_{max}}u_{2}(t) I_{RU}-\theta_{1} I_{RU}-\mu I_{RU}] \nonumber \\
    && + \sigma_{6}[k_{R_{max}}u_{2}(t) I_{RU}-\beta_{R} u_{4}(t) I_{RD}-\theta_{1}(1-u_{4}(t)) I_{RD}-\mu I_{RD}] 
    + \sigma_{7} [\beta_{R} u_{4}(t) I_{RD}-\theta_{2} T_{2}-\mu T_{2}] \nonumber \\
    && + \sigma_{8} [\theta_1 I_{SU}+\theta_{1} (1-u_{3}(t)) I_{SD}+\theta_{2} u_{5}(t) T_{1}+\theta_{1} I_{RU}+\theta_{1}(1-u_{4}(t)) I_{RD}+\theta_{2} T_{2}-(\mu+\mu_d) A ],
\end{eqnarray}
where $\sigma_{i}(t) $'s (for $i=1,2,\ldots,8$) represent the adjoint functions corresponding to the state variables, which will be determined in the subsequent analysis.

\begin{theorem}{\label{03:thm:characterization_optimal_control}}
Let $(u^*_{1}, u^*_{2}, u^*_{3}, u^*_{4}, u^*_{5})\in \mathcal{U}$ be an optimal control profile with the corresponding optimal state solution $(S^{*}, I^{*}_{SU}, I^{*}_{SD}, T^{*}_{1}, I^{*}_{RU}, I^{*}_{RD}, T^{*}_{2}, A^{*})$ subject to the control problem (\ref{03:eqn:control_functional})-(\ref{03:system:control}). Then there exist adjoint functions $\sigma_{i} $'s (for $i=1,2,\ldots,8$) which satisfy,
    \begin{eqnarray}{\label{03:eqn:adjoint_system}}
        \sigma^{\prime}_{1} &=& \alpha_{S}(I_{SU}+c I_{SD})(\sigma_{1}-\sigma_{2})+ \alpha_{R}(I_{RU}+c I_{RD})(\sigma_{1}-\sigma_{5}) + \mu \sigma_{1}, \nonumber \\
        \sigma^{\prime}_{2} &=& -W_{1} + \alpha_{S} S(\sigma_{1}-\sigma_{2}) + k_{S_{max}} u_{1}(\sigma_{2}-\sigma_{3}) + \theta_{1} (\sigma_{2}-\sigma_{8}) + \mu \sigma_{2},  \nonumber \\
        \sigma^{\prime}_{3} &=& -W_{2} + c \alpha_S (\sigma_{1}-\sigma_{2}) + \beta_{S} u_{3} (\sigma_{3}-\sigma_{4}) + \theta_{1} (1-u_{3})(\sigma_{3}-\sigma_{8}) + \mu \sigma_{3},   \nonumber \\
        \sigma^{\prime}_{4} &=&  W_{5} + \gamma (1-u_{5}) (\sigma_{4}-\sigma_{5}) + \theta_{2} u_{5} (\sigma_{4}-\sigma_{8}) + \mu \sigma_{4} \nonumber \\
        \sigma^{\prime}_{5} &=& -W_{3} + \alpha_{R} S (\sigma_{1}-\sigma_{5}) + k_{R_{max}} u_{2} (\sigma_{5}-\sigma_{6}) + \theta_{1} (\sigma_{5}-\sigma_{8}) + \mu \sigma_{5}, \nonumber \\
        \sigma^{\prime}_{6} &=& -W_{4} + c \alpha_{R} S (\sigma_{1}-\sigma_{5}) + \beta_{R} u_{4} (\sigma_{6}-\sigma_{7}) + \theta_{1} (1-u_{4}) (\sigma_{6}-\sigma_{8}) + \mu \sigma_{6}, \nonumber \\
        \sigma^{\prime}_{7} &=& W_{6} + \theta_{2} (\sigma_{7}-\sigma_{8}) + \mu \sigma_{7}, \nonumber \\
        \sigma^{\prime}_{8} &=& (\mu + \mu_{d}) \sigma_{8},
    \end{eqnarray}
with transversality conditions $\sigma_{i}(t_f)=0 \quad \forall i=1,2,\ldots,8.$ Further, the solutions of the optimal control variables are given as,
    \begin{eqnarray}
    {\label{03:eqn:characterization_optimal_control_1}}
        u^{*}_{1} &=& \max \left\{ 0, \min \left\{\frac{k_{S_{max}} (\sigma_{2}-\sigma_{3}) I^{*}_{SU}}{w_{1}} , 1 \right\} \right\},  \\ 
    {\label{03:eqn:characterization_optimal_control_2}}
        u^{*}_{2} &=& \max \left\{ 0, \min \left\{\frac{k_{R_{max}}(\sigma_{5}-\sigma_{6}) I^{*}_{RU}}{w_{2}} , 1 \right\} \right\}, \\
    {\label{03:eqn:characterization_optimal_control_3}}
        u^{*}_{3} &=& \max \left\{ 0, \min \left\{\frac{\left((\sigma_{3}-\sigma_{4})\beta_{S} + (\sigma_{8}-\sigma_{3})\theta_{1} \right) I^{*}_{SD}}{w_{3}} , 1 \right\} \right\}, \\
    {\label{03:eqn:characterization_optimal_control_4}}
        u^{*}_{4} &=& \max \left\{ 0, \min \left\{\frac{\left((\sigma_{6}-\sigma_{7})\beta_{R} + (\sigma_{8}-\sigma_{6})\theta_{1} \right) I^{*}_{RD}}{w_{4}} , 1 \right\} \right\}, \\
    {\label{03:eqn:characterization_optimal_control_5}}
         u^{*}_{5} &=& \max \left\{ 0, \min \left\{\frac{\left((\sigma_{5}-\sigma_{4})\gamma + (\sigma_{4}-\sigma_{8})\theta_{2} \right) T^{*}_{1}}{w_{5}} , 1 \right\} \right\}.
    \end{eqnarray}
\end{theorem}
\begin{proof}
Given the existence of an optimal solution, the Pontryagin's Maximum Principle provides the system of differential equations for the adjoint variables, derived by differentiating the Hamiltonian with respect to the state variables and evaluating it at the optimal control. Therefore, the adjoint system is determined by solving the following set of differential equations, subject to the specified boundary conditions:
    $$\frac{d \sigma_{1}}{dt}=-\frac{dH}{dS}, ~ \frac{d \sigma_{2}}{dt}=-\frac{dH}{dI_{SU}}, \ldots ,\frac{d \sigma_{8}}{dt}=-\frac{dH}{dA}, \quad \sigma_{i}(t_f)=0 \quad \forall i=1,2,\ldots,8,$$
which results in the adjoint system (\ref{03:eqn:adjoint_system}) for given control problem.

To characterize the optimal controls, we differentiate the Hamiltonian function with respect to the control variables within the interior of the control set $\mathcal{U}$. Applying Pontryagin's Maximum Principle, we derive the following optimality conditions:
\begin{eqnarray}{\label{03:eqn:optimality_conditions}}
        \dfrac{dH}{d u_{1}}&=&w_{1}u^{*}_{1}-k_{S_{max}}(\sigma_{2}-\sigma_{3}) I^{*}_{SU}=0, \nonumber \\
        \dfrac{dH}{d u_{2}}&=&w_{2}u^{*}_{2}-k_{R_{max}}(\sigma_{5}-\sigma_{6}) I^{*}_{RU}=0, \nonumber \\
        \dfrac{dH}{d u_{3}}&=&w_{3}u^{*}_{3}-\beta_{S}(\sigma_{3}-\sigma_{4}) I^{*}_{SD}-\theta_{1}(\sigma_{8}-\sigma_{3}) I^{*}_{SD}=0, \nonumber \\
        \dfrac{dH}{d u_{4}}&=&w_{4}u^{*}_{4}-\beta_{R}(\sigma_{6}-\sigma_{7}) I^{*}_{RD}-\theta_{1}(\sigma_{8}-\sigma_{6}) I^{*}_{RD}=0, \nonumber \\
        \dfrac{dH}{d u_{5}}&=&w_{5}u^{*}_{5}-\gamma(\sigma_{5}-\sigma_{4}) T^{*}_{1}-\theta_{2}(\sigma_{4}-\sigma_{8}) T^{*}_{1}=0.
\end{eqnarray}
Solving the system (\ref{03:eqn:optimality_conditions}) for $(u^*_{1}, u^*_{2}, u^*_{3}, u^*_{4}, u^*_{5})$ and using the bounds of control set $\mathcal{U}$ yield the optimal control characterizations given in (\ref{03:eqn:characterization_optimal_control_1})-(\ref{03:eqn:characterization_optimal_control_5}) for the control problem (\ref{03:eqn:control_functional})-(\ref{03:system:control}). 
\end{proof}

A comprehensive discussion on various control strategies, their epidemiological outcomes, and the corresponding cost-effectiveness analysis is presented in Section \ref{03:sec:numerical_simulation}. 

\section{Numerical simulations}{\label{03:sec:numerical_simulation}}
Dynamic modelling of complex systems depends on numerical simulations to effectively illustrate and support the theoretical findings of the model. In this Section, we first validate the analytical results on the existence and stability of equilibrium points through numerical simulations. We then present a comprehensive evaluation of various optimal control strategies, including a dynamic control approach, and their effects on the model outcomes. A detailed cost-effectiveness analysis is carried out, considering both infection reduction and treatment coverage expansion objectives. Additionally, we perform an adjoint-based sensitivity analysis to assess how supplementary resources can be optimally allocated to maximize public health benefits. Finally, to capture both the individual and synergistic contributions of control variables, we include a control contribution analysis using Shapley values, a concept from cooperative game theory.

\subsection{Existence and stability of equilibrium points}{\label{03:subsec:numerical_simulation:existence_stability_EP}}

To illustrate the existence and stability of all equilibrium points, we fix the majority of parameters in the system (\ref{3_system_reduced}) as outlined in Section \ref{03:sec:parameter_estimation}. To explore various scenarios, we vary the parameters $\alpha_{S}$ and $\alpha_{R}$ within their feasible ranges, as specified in Table \ref{03:tab:parameter_values}, since these parameters critically influence the basic reproduction numbers. The initial condition in each time series plot is set as suggested in Section  \ref{03:sec:parameter_estimation}. First, we set $\alpha_{S}=0.000047$ and $\alpha_{R}=0.00002$. For this set of parameters, the basic reproduction numbers are computed as $R_{0}^{(S)}=0.39 (<1), ~R_{0}^{(R)}=0.02(<1)$ and $R_{0}^{(SR)}=25.41(>1)$. According to Theorem \ref{03:thm:eqilibrium_point_0}, this configuration admits a unique disease-free equilibrium point, $E^{(0)} = (3615.7143, 0, 0, 0, 0, 0, 0)$. The eigenvalues of the Jacobian matrix $J(E^{(0)})$ are $-13.5276, -7.7319, -5.024, -0.3668, -0.2652, -0.037$ and $-0.007$. Since all eigenvalues have a negative real part, $E^{(0)}$ is locally asymptotically stable. This behavior is captured in the time series plot presented in Figure \ref{03:fig:numerical_simulation_time_series_E0}.

Next, we consider the parameter set with $\alpha_{S}=0.00005$ and $\alpha_{R}=0.003$. The corresponding basic reproduction numbers are $R_{0}^{(S)}=0.42 (<1), ~R_{0}^{(R)}=2.32(>1)$ and $R_{0}^{(SR)}=0.18(<1)$. These values satisfy the conditions for the existence and local stability of the drug-resistant strain endemic equilibrium point, as described in Theorem \ref{03:thm:eqilibrium_point_1}. The equilibrium point in this case is given by $E^{(1)} = (1558.2339, 0, 0, 0, 2.8201, 1.0432, 380.6183)$ and the eigenvalues of the Jacobian matrix $J(E^{(1)})$ are $-0.0008 \pm 0.2137i,$ $-13.9495, -7.7309, -0.3668, -0.3582$ and $-0.037$. Since all eigenvalues have negative real parts, the local asymptotic stability of $E^{(1)}$ is confirmed, which is further supported by the plot shown in Figure \ref{03:fig:numerical_simulation_time_series_E1}. Note that the disease-free equilibrium point also exists for this parameter set, but it is unstable. 

We further explore the system dynamics by considering a new set of parameters with $\alpha_{S}=0.0003$ and $\alpha_{R}=0.0001$. The basic reproduction numbers are $R_{0}^{(S)}=2.51 (>1), ~R_{0}^{(R)}=0.08(<1)$ and $R_{0}^{(SR)}=32.43(>1)$. Under these conditions, both the co-existence endemic equilibrium point $E^{(*)}$ and the DFE $E^{(1)}$ exist. The $E^{(*)}$ is given by $(1441.2443,\allowbreak 34.0749,\allowbreak 1.4547, 30.5372, 2.0977, 0.7759, 283.1167)$. To examine the local stability of this equilibrium, we compute the characteristic polynomial of the Jacobian matrix $J(E^{(*)})$, which is given by:
\begin{align*}
    p_{*}(x)&= (x+0.037)q_{*}(x) \\
            &= (x+0.037)(x^{6}+26.5988 x^5 + 219.924 x^4 + 598.411 x^3 + 201.237 x^2 + 6.0169 x + 0.8754)
\end{align*}
The Routh array for the polynomial $q_{*}(x)$ is:
\[
\begin{array}{c|cccc}
s^6 & 1 & 219.924 & 201.237 & 0.8754 \\
s^5 & 26.5988 & 598.411 & 6.0169 & 0 \\
s^4 & 197.4262 & 201.0108 & 0.8754 & 0 \\
s^3 & 571.3298 & 5.8989 & 0 & 0 \\
s^2 & 198.9724 & 0.8754 & 0 & 0 \\
s^1 & 3.3852 & 0 & 0 & 0 \\
s^0 & 0.8754 & 0 & 0 & 0 \\
\end{array}
\]
Note that all the elements in the first column of the Routh array are positive, thereby satisfying the local stability conditions outlined in Theorem \ref{03:thm:eqilibrium_point_2}. The roots of the characteristic equation $p_{*}(x)=0$ are $-0.0086 \pm 0.0673 i$, $-0.037,\allowbreak -13.5381, -7.7346, -4.9418$ and $-0.3671$. Therefore, the equilibrium point $E^{(*)}$ is locally asymptotically stable, as further supported by the time series plot in Figure \ref{03:fig:numerical_simulation_time_series_E2}, which illustrates the system's convergence to the co-existence state over time.

\begin{figure}[t!]
\begin{center}
\subfloat[$\alpha_{S}=0.000047$, $\alpha_{R}=0.00002$]{\includegraphics[height=6.4cm,width=9cm]{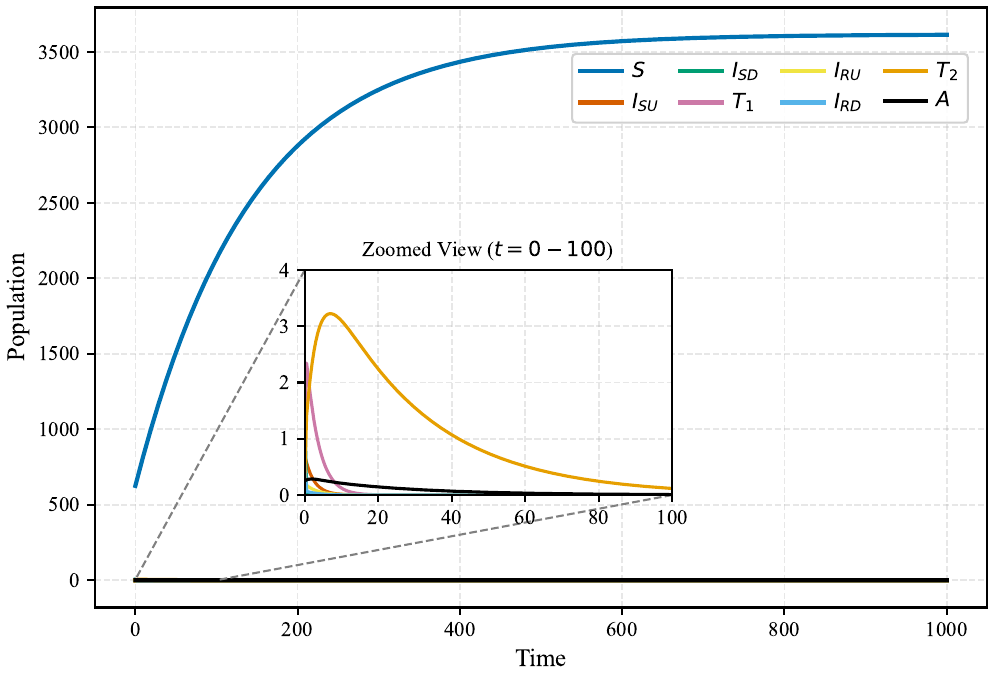}\label{03:fig:numerical_simulation_time_series_E0}}
\quad
\subfloat[$\alpha_{S}=0.00005$, $\alpha_{R}=0.003$]{\includegraphics[height=6.5cm,width=9.2cm]{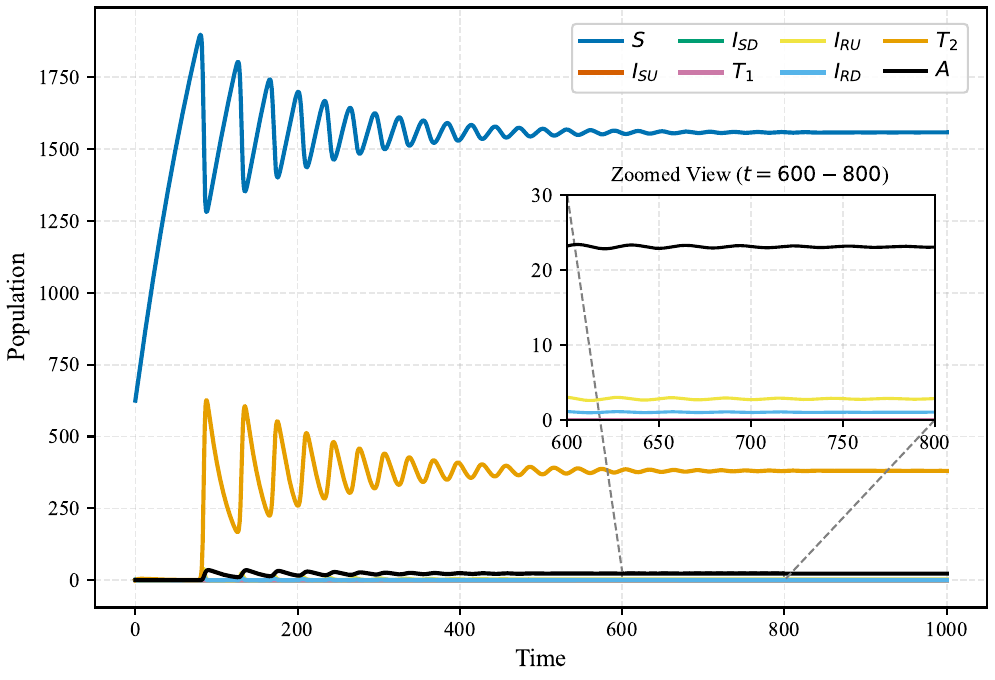}\label{03:fig:numerical_simulation_time_series_E1}}\\
\end{center}
\begin{center}
\subfloat[$\alpha_{S}=0.0003$, $\alpha_{R}=0.0001$]{\includegraphics[height=6.5cm,width=9.2cm]{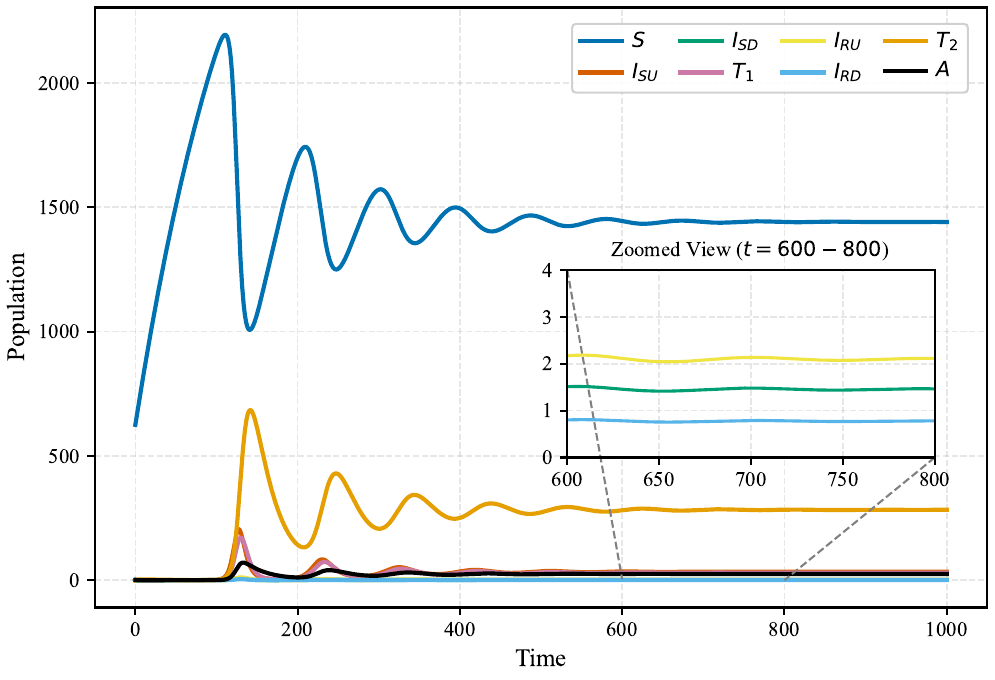}\label{03:fig:numerical_simulation_time_series_E2}}
\end{center}
\caption{Time series plots for the system (\ref{3_system_main}) showing convergence of solutions to (a) the disease-free equilibrium point, (b) the drug-resistant strain endemic equilibrium point, and (c) the coexistence endemic equilibrium point. }\label{03:fig:numerical_simulation_time_series}
\end{figure}	

\subsection{Control strategies}{\label{03:subsec:numerical_simulation:optimal_control}}
In this discussion, we examine the impact of various control strategies for the proposed optimal control problem (\ref{03:eqn:control_functional})-(\ref{03:system:control}), with the goal of achieving the 95-95-95 targets set by UNAIDS \cite{unai2014}. Numerical simulations are performed by solving the control system (\ref{03:eqn:control_functional})-(\ref{03:system:control}) along with the corresponding adjoint system (\ref{03:eqn:adjoint_system}) and the control characterization (\ref{03:eqn:characterization_optimal_control_1})-(\ref{03:eqn:characterization_optimal_control_5}) through an iterative scheme. We employ the forward-backward sweep method, starting with an initial estimate for the control variables, and solve the state system using the classical fourth-order Runge-Kutta method in forward direction over the simulated time. Subsequently, with the state trajectories determined, the adjoint system is solved backward in time using the transversality conditions $\sigma_{i}(t_{f})=0 ~ \forall i=1,2,\dots,8$, again applying the Runge-Kutta fourth-order integration. The control variables are updated at each iteration using a convex combination of the previous control values and those derived from the characterization equations to improve numerical stability. This iterative procedure is repeated until successive iterations give negligible differences in the computed values of control variables (for details see \cite{lenh2007}). 

In general, binary control activation has traditionally been employed in the literature to differentiate between intervention strategies, where a specific control or a combination of controls is activated while all non-focal controls are set to zero \cite{lenh2007, okos2013, yusu2023}. Although this method effectively isolates the individual or combined effects of interventions, it imposes artificial constraints that poorly reflect real-world public health policy implementation. In practice, interventions are rarely applied in isolation. To address this limitation, we introduce a weight-varying mechanism in the cost functional to distinguish among various strategies for practical advantages. A similar approach has been used in \cite{agus2019, peni2020, yusu2023}. This approach enables a more effective allocation of resources that accounts for budget and operational constraints. It also ensures that non-focal interventions are maintained at a reduced but non-zero baseline level, rather than being completely excluded. As a result, the complementary benefits of multiple concurrent interventions are preserved that are typically lost under the binary control activation framework. The resulting optimal control profiles represent more balanced strategies, prioritizing key interventions while maintaining essential supportive measures, thereby offering a more realistic application to real-world public health decision-making. 

We propose a set of control strategies that target key epidemiological components, including HIV infection diagnosis, initiation and adherence to treatment, and management of drug resistance. The descriptions of these strategies are as following:

\vspace{0.2cm}

\noindent
\textbf{Strategy A - Diagnosis focused:} This strategy prioritizes finding infected individuals before they can contribute to further transmission, making it particularly effective when early diagnosis substantially reduces onward spread. It is a suitable option in settings where maintaining diagnostic infrastructure is more feasible and cost-effective than providing treatment, particularly in low-income countries. The early diagnosis is prioritized by considering low weights for the costs associated with control $u_{1}$ and $u_{2}$ while treatment-related costs are maintained at a moderate level.

\vspace{0.2cm}

\noindent
\textbf{Strategy B - Treatment focused:} Prioritizing treatment-related interventions that enhance treatment availability can be particularly effective in the mature phase of the epidemic, where a substantial proportion of infected individuals are already diagnosed. In some cases, a high prevalence of drug-resistant infections also necessitates the adoption of treatment-focused strategies. We consider lower values for weight constants $w_{3}$ and $w_{4}$ to illustrate this strategy. 

\vspace{0.2cm}

\noindent
\textbf{Strategy C - Adherence focused:} Once a substantial proportion of the infected population initiates the treatment, it becomes crucial to shift the focus towards adherence-related interventions to prevent the emergence of drug resistance. Consequently, it reduces the prevalence of drug-resistant infections and might be effective when the second-line treatments are limited and costly. This emphasis is captured by assigning a lower weight to the cost constant $w_{5}$, which prioritizes adherence-enhancing measures.

\vspace{0.2cm}

\noindent
\textbf{Strategy D - Balanced:} In resource-rich settings, where multiple interventions can be simultaneously implemented to rapidly reduce the disease burden, a balanced strategy may be adopted. This involves assigning equal weights to all interventions, ensuring uniform emphasis across diagnosis, treatment, and adherence. This approach is also suitable in the final phase of the epidemic, when comprehensive efforts are needed to eliminate remaining transmission.

\vspace{0.2cm}

\noindent
\textbf{Strategy E - Dynamic control optimization: A framework for 95-95-95 target achievement:} For the optimal achievement of the UNAIDS 95-95-95 targets, we consider a dynamic strategy that combines the key epidemiological components using the model predictive control (MPC) approach, also known as receding horizon control (RHC) \cite{mayn2000, zura2006}. In this approach, we define three key metrics based on the values of control and state variables at a given time: the proportion of diagnosed individuals $M_{1}$, those under treatment $M_{2}$, and those who are virally suppressed $M_{3}$ among all HIV-infected individuals. Note that, earlier defined Strategy A improves $M_1$ by increasing diagnosis rates, Strategy B enhances $M_2$ by improving treatment initiation, and Strategy C raises $M_3$ by supporting treatment adherence. Once the 95-95-95 targets are achieved, a balanced strategy is applied, unless any of the metrics drop below the $95\%$ threshold. For this, the total simulation period is partitioned into $k$ sub-intervals, determined by selected decision points where these metrics are recalculated to get feedback from the model. At each decision point, an optimal control problem is solved over the remaining finite prediction horizon, using the current system state as the initial condition. This optimization generates a sequence of optimal control inputs, and only the first portion of the computed control sequence is applied to the system. The process is repeated at subsequent time steps, continuously updating the control strategy based on the evolving values of metrics $M_{1}$, $M_{2}$ and $M_{3}$, thus creating a feedback loop through the state variables. This modified MPC framework ensures the convergence of the optimal control solution by eliminating instabilities that may arise from the strategy-based variations in the cost functional at different phases of the simulation period. For a detailed description of the computational procedure, see Algorithm \ref{03:alg:numerical_simulation_dynamic_control_optimization}.

\begin{algorithm}
\caption{Dynamic Control Optimization}
\label{03:alg:numerical_simulation_dynamic_control_optimization}
\begin{algorithmic}[1]
\REQUIRE Simulation time $t_f$, decision points $D = \{0, 1, \dots, t_f\}$, model parameters $\theta$, initial conditions $y_0$, strategy weights $W$, control bounds $U_b = [0, 1]^5$, convergence tolerance $\epsilon$, maximum iterations $k_{\text{max}}$, relaxation parameter $\alpha$
\STATE \textbf{Initialize:} Set $t \gets 0$, $y \gets y_0$, and discretize time grid $t_{\text{total}} = \{t_0, t_1, \dots, t_N\}$ over $[0, t_f]$ with $N$ points
\FOR{each decision point $t_d \in D$}
    \STATE Calculate 95-95-95 metrics at $t_d$:
    \STATE \quad $M_1(t_d) \gets \frac{I_{SD} + I_{RD} + T_1 + T_2}{I_{SU} + I_{SD} + I_{RU} + I_{RD} + T_1 + T_2}$
    \STATE \quad $M_2(t_d) \gets \frac{T_1 + T_2}{I_{SD} + I_{RD} + T_1 + T_2 }$
    \STATE \quad $M_3(t_d) \gets \frac{u_5^* T_1 + T_2}{T_1 + T_2}$

    \IF{$M_1(t_d) < 0.95$}
        \STATE Select Strategy A (Diagnosis focused)
    \ELSIF{$M_1(t_d) \geq 0.95$ and $M_2(t_d) < 0.95$}
        \STATE Select Strategy B (Treatment focused)
    \ELSIF{$M_1(t_d) \geq 0.95$ and $M_2(t_d) \geq 0.95$ and $M_3(t_d) < 0.95$}
        \STATE Select Strategy C (Adherence focused)
    \ELSE
        \STATE Select Strategy D (Balanced)
    \ENDIF

    \STATE Set cost functional weights $W$ based on selected strategy

    \STATE Solve optimal control problem from $t_d$ to $t_f$ using backward-forward sweep:
    \FOR{$k = 1$ to $k_{\text{max}}$}
        \STATE Solve state system $\dot{y} = f(y, u^{(k-1)}, \theta)$ from $t_d$ to $t_f$ with $y(0)=y_{0}$
        \STATE Solve adjoint system $\dot{\sigma} = -\nabla_y H(y, u^{(k-1)}, \sigma, \theta, W)$ from $t_f$ to $t_d$ with $\sigma(t_f)=0$
        \STATE Compute $u^{\text{new}}$ using $\frac{\partial H}{\partial u_j} = 0$ for $j = 1, \dots, 5$
        \STATE Update $u^{(k)} \gets (1 - \alpha) u^{(k-1)} + \alpha u^{\text{new}}$
        \IF{$\max_j \| u_j^{(k)} - u_j^{(k-1)} \|_\infty < \epsilon$}
            \STATE Set $u^* \gets u^{(k)}$
            \STATE \textbf{break}
        \ENDIF
    \ENDFOR

    \STATE Implement controls until next decision point:
    \STATE \quad Set $t_{\text{next}} \gets t_d + 1$ if $t_d < t_f$ else $t_f$
    \STATE \quad Simulate system forward to $t_{\text{next}}$ using $u^*(t)$
    \STATE \quad Update $y \gets y(t_{\text{next}})$, $t \gets t_{\text{next}}$
\ENDFOR
\RETURN Dynamic optimal control profile $u^*(t)$, state trajectories $y(t)$.
\end{algorithmic}
\end{algorithm}

For the numerical simulations, we use a representative set of parameters as listed in Table (\ref{03:tab:parameter_values}). The simulation is carried out over a time horizon of $t_{f}=25$ years, which is essential for a chronic infection like HIV, where the infectious period is prolonged and the effects of interventions require several years to be observed. 
The initial population size is $S(0)=625, I_{SU}(0)=0.2, I_{SD}(0)=0.5, T_{1}(0)=1.8, I_{RU}(0)=0.2, I_{RD}(0)=0.1, T_{2}(0)=0.2, A(0)=0.25$. The weight constants associated with the state variables are fixed as $W_{1}=1$, $W_{2}=0.5$, $W_{3}=5$, $W_{4}=2$, $W_{5}=0.25$ and $W_{6}=0.05$. Note that, these weights are assigned based on the relative burden of infection of each state variable within the community. The values of $W_{5}$ and $W_{6}$ reflect the prioritization of treatment coverage, even though this may lead to increased costs. Further, the weight constants associated with the control functions, along with the corresponding strategies, are listed in the Table \ref{03:tab:numerical_simulation_weight_constants}. It is important to note that the weight constants used in the simulations are theoretical in nature and are primarily intended to illustrate the different control strategies proposed in this study. However, they are roughly scaled based on estimated intervention costs. Higher weights are assigned to the drug-resistant infected population compared to the drug-sensitive group, reflecting the increased cost of managing drug-resistant related interventions. Additionally, given the prolonged nature of HIV treatment, the weights associated with treatment-related interventions are chosen to be relatively higher to account for their long-term resource requirements.

\begin{table}[htbp]
  \centering 
  \begin{tabular}{*{6}{c}} 
    \hline
  Strategy & $w_{1}$ & $w_{2}$ & $w_{3}$ & $w_{4}$ & $w_{5}$ \\
    \hline
   Strategy A & 0.005 & 0.005 & 2 & 5 & 1 \\
   Strategy B & 0.5 & 0.5 & 0.02 & 0.05 & 1 \\
   Strategy C & 0.5 & 0.5 & 2 & 5 & 0.01 \\
   Strategy D & 0.5 & 0.5 & 2 & 5 & 0.6 \\
    \hline
  \end{tabular}
  \caption{The weight constants associated with the control variables in different control strategies}
  \label{03:tab:numerical_simulation_weight_constants}
\end{table}

The optimal control profiles for each strategy are shown in Figures \ref{03:fig:numerical_simulation_strategy_A}-\ref{03:fig:numerical_simulation_strategy_95}. These figures also illustrate the state trajectories of infected and treated populations, highlighting the differences in population dynamics under various control strategies. In addition, we included two key epidemiological indicators, the proportions of drug-resistant and undiagnosed individuals within the total infected population, to identify the effectiveness of proposed strategies in achieving their individual objectives. 
In Figure \ref{03:fig:numerical_simulation_strategy_A_1}, the optimal control profile for Strategy A emphasizes diagnosis-focused interventions. Consequently, the proportion of undiagnosed cases declines significantly from the beginning. As the initiation rate of second-line therapy increases to a certain level, the drug-resistant cases starts to fall (see Figure \ref{03:fig:numerical_simulation_strategy_A_2}). Note that, the number of diagnosed drug-resistant individuals rises sharply in the early phase (see Figure \ref{03:fig:numerical_simulation_strategy_A_3}). This is primarily due to an imbalance between the high diagnosis rate, driven by higher values of control $u_{2}$, and the limited availability of second-line treatment. 

\begin{figure}[t!]
\begin{center}
\subfloat[Optimal control profile]{\includegraphics[height=6.4cm,width=9cm]{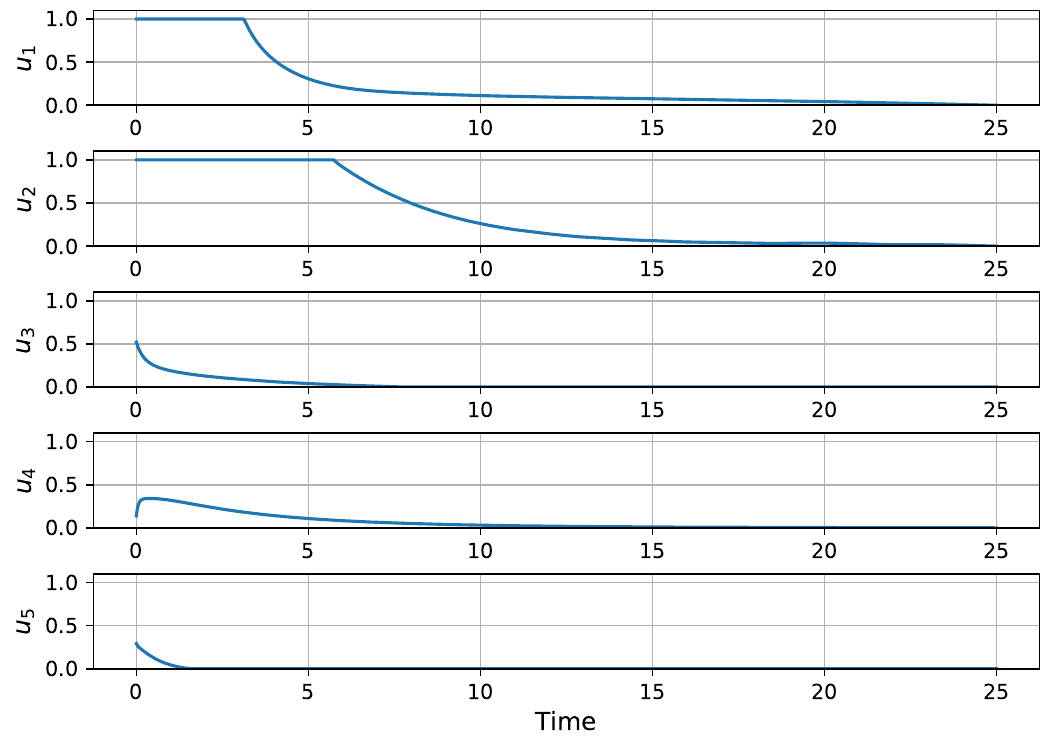}\label{03:fig:numerical_simulation_strategy_A_1}}
\quad
\subfloat[Epidemiological indicators]{\includegraphics[height=6.5cm,width=9.2cm]{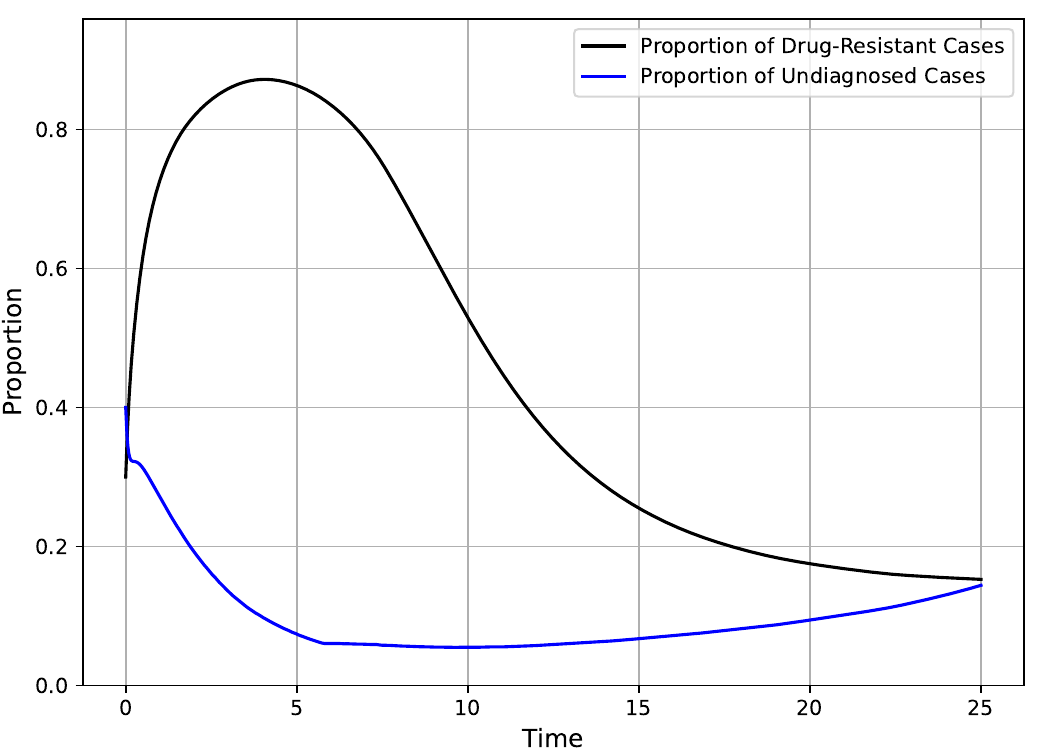}\label{03:fig:numerical_simulation_strategy_A_2}}\\
\end{center}
\begin{center}
\subfloat[Infected population dynamics]{\includegraphics[height=6.5cm,width=9.2cm]{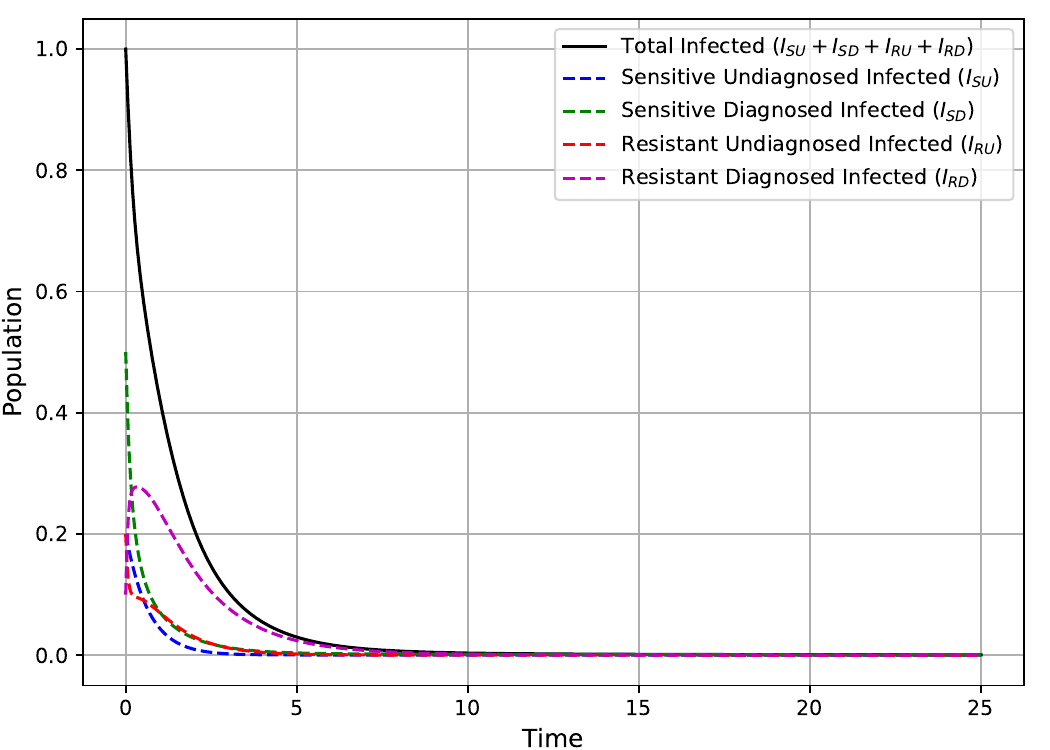}\label{03:fig:numerical_simulation_strategy_A_3}}
\quad 
\subfloat[Treatment coverage dynamics]{\includegraphics[height=6.5cm,width=9.2cm]{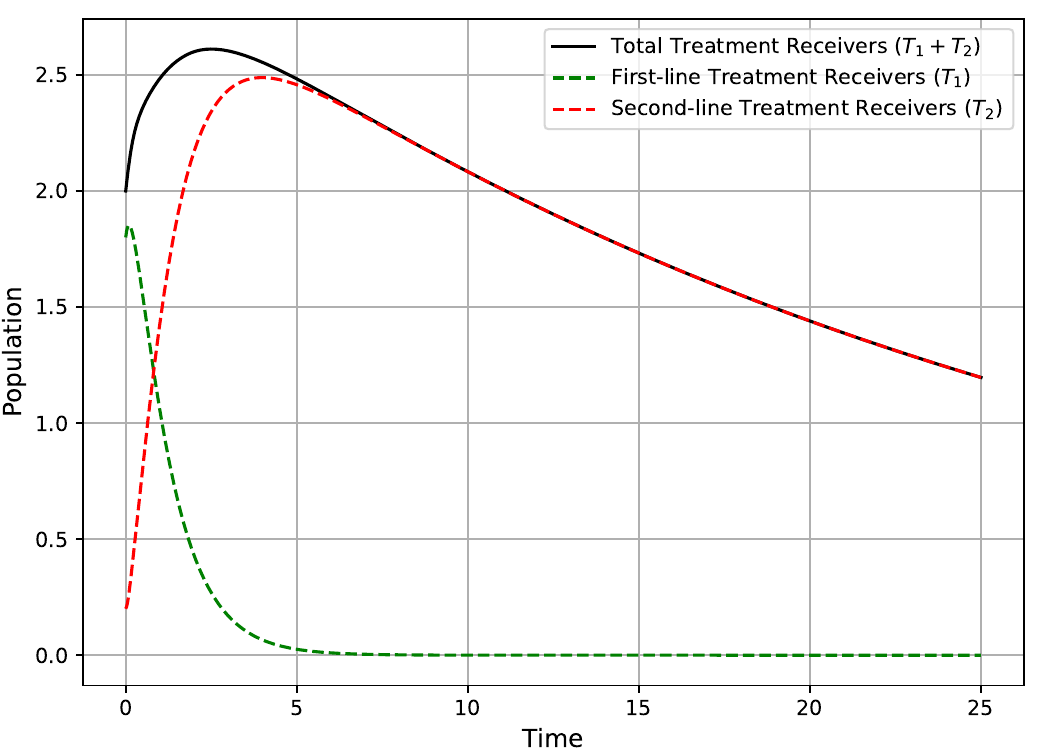}\label{03:fig:numerical_simulation_strategy_A_4}}\\
\end{center}
\caption{State and optimal control dynamics of the optimal control problem (\ref{03:eqn:control_functional})-(\ref{03:system:control}) for Strategy A.}\label{03:fig:numerical_simulation_strategy_A}
\end{figure}	

In Strategy B, the optimal control profile reflects a prioritized allocation of resources toward treatment-related interventions (see Figure \ref{03:fig:numerical_simulation_strategy_B}). Although this strategy emphasizes treatment, it also focuses on diagnosis in the early phase of the simulation, making the infected population eligible to receive treatment. The proportion of undiagnosed cases remains higher compared to Strategy A (see Figure \ref{03:fig:numerical_simulation_strategy_B_2}). As observed in Figure \ref{03:fig:numerical_simulation_strategy_B_3}, however, all infected populations continuously declines following the implementation of this control strategy, makes it more efficient in reducing disease burden relative to Strategy A. As in Strategy A (see Figure \ref{03:fig:numerical_simulation_strategy_A_4}), the limited efforts in adherence-enhancing interventions reduce efficacy of first-line treatment, leading to the development of drug resistance among some patients and their subsequent transition to second-line treatment (see Figure \ref{03:fig:numerical_simulation_strategy_B_4}). 

\begin{figure}[t!]
\begin{center}
\subfloat[Optimal control profile]{\includegraphics[height=6.4cm,width=9cm]{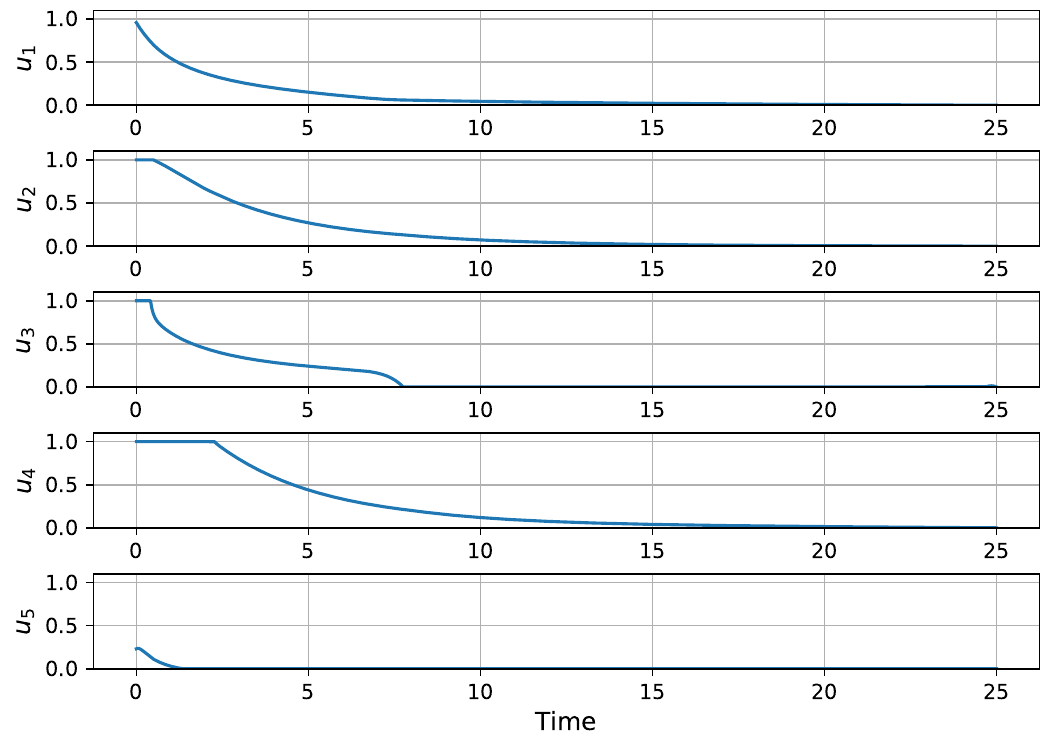}\label{03:fig:numerical_simulation_strategy_B_1}}
\quad
\subfloat[Epidemiological indicators]{\includegraphics[height=6.5cm,width=9.2cm]{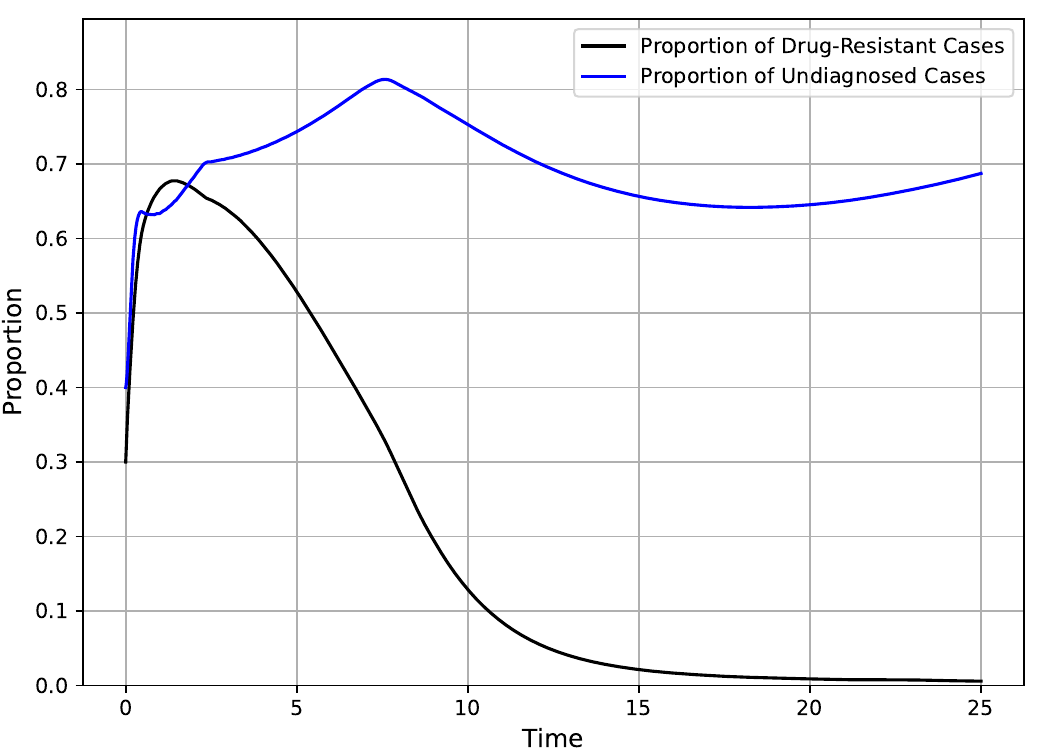}\label{03:fig:numerical_simulation_strategy_B_2}}\\
\end{center}
\begin{center}
\subfloat[Infected population dynamics]{\includegraphics[height=6.5cm,width=9.2cm]{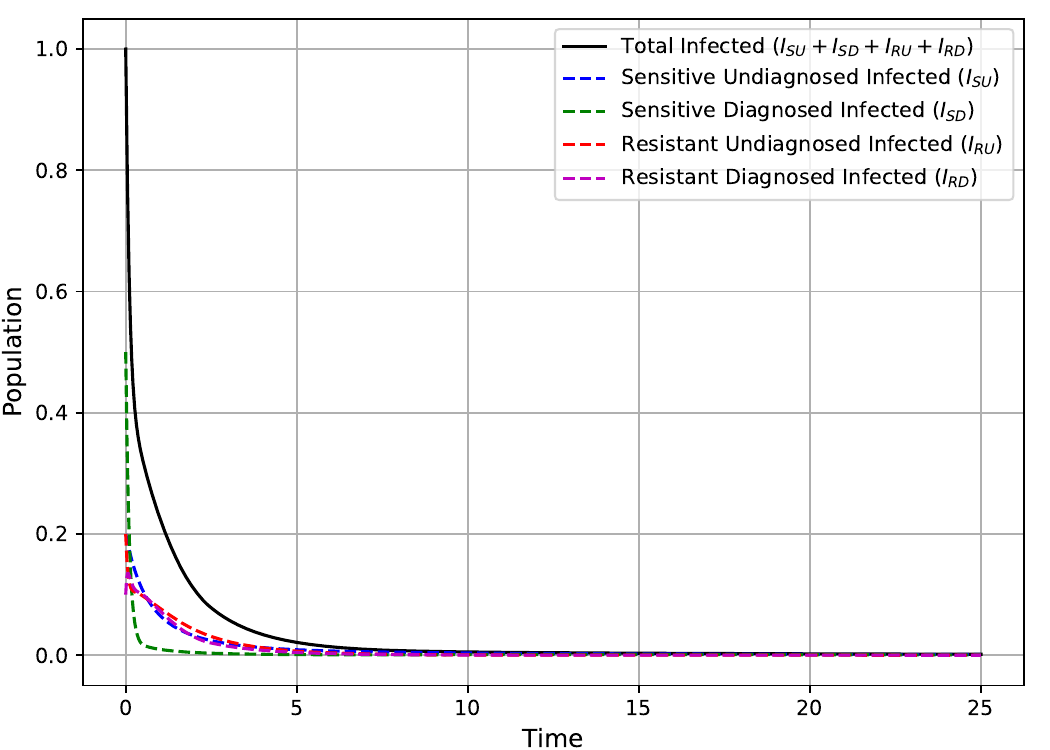}\label{03:fig:numerical_simulation_strategy_B_3}}
\quad 
\subfloat[Treatment coverage dynamics]{\includegraphics[height=6.5cm,width=9.2cm]{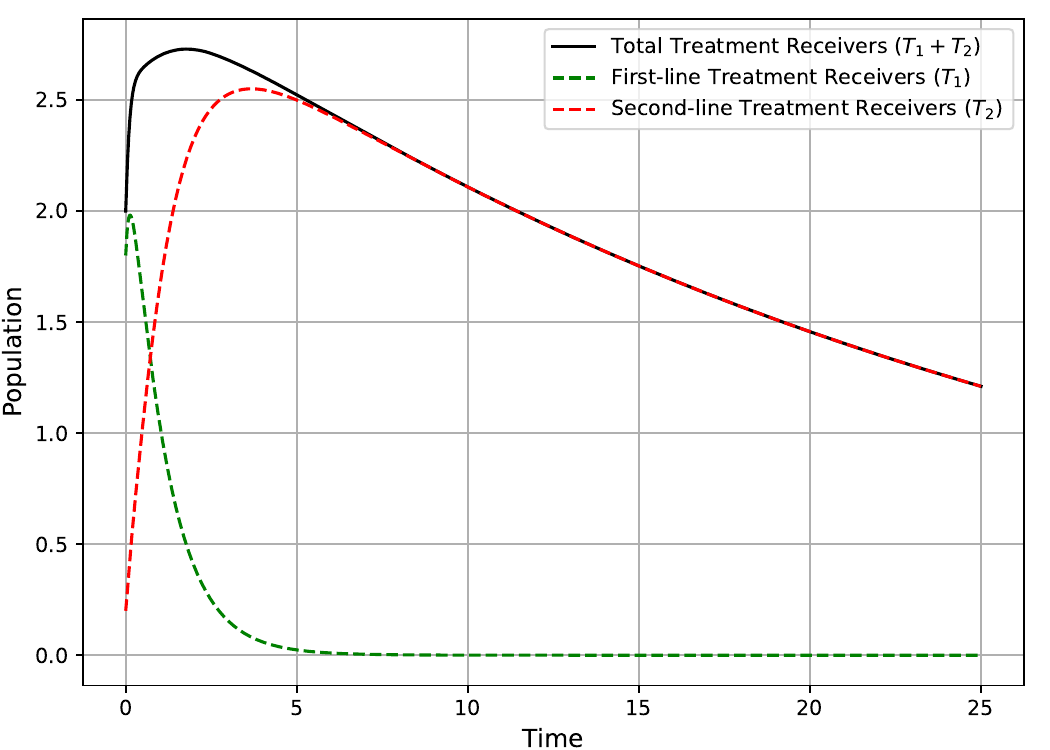}\label{03:fig:numerical_simulation_strategy_B_4}}\\
\end{center}
\caption{State and optimal control dynamics of the optimal control problem (\ref{03:eqn:control_functional})-(\ref{03:system:control}) for Strategy B.}\label{03:fig:numerical_simulation_strategy_B}
\end{figure}	

The optimal controls for Strategy C are presented in Figure \ref{03:fig:numerical_simulation_strategy_C_1} highlighting the priority to adherence-related control variable $u_{5}$. While this strategy leads to a sharp decline in all infected compartments, it does not reduce the proportion of undiagnosed population (see Figure \ref{03:fig:numerical_simulation_strategy_C_2}, Figure \ref{03:fig:numerical_simulation_strategy_C_3}). However, higher intensity of control $u_{5}$ reduces the development of drug resistance. As a result, after an initial short-term rise, the requirement for second-line therapy decreases. In contrast to Strategies A and B, this approach results in fewer second-line treatment initiations (see Figure \ref{03:fig:numerical_simulation_strategy_C_4}).  

\begin{figure}[t!]
\begin{center}
\subfloat[Optimal control profile]{\includegraphics[height=6.4cm,width=9cm]{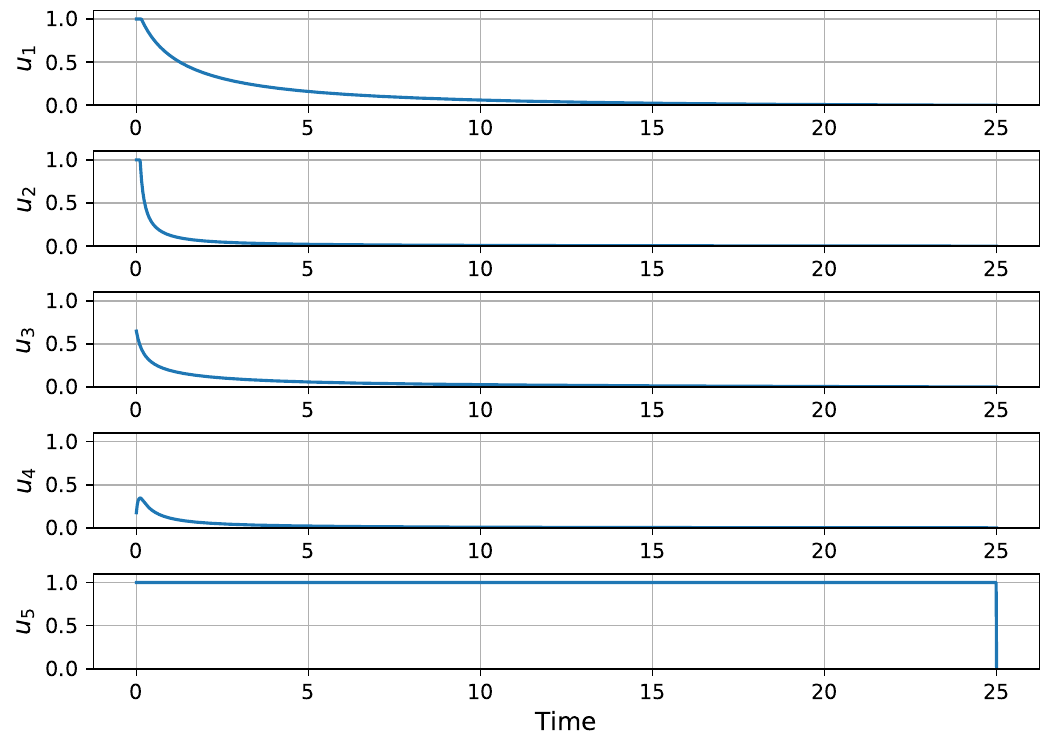}\label{03:fig:numerical_simulation_strategy_C_1}}
\quad
\subfloat[Epidemiological indicators]{\includegraphics[height=6.5cm,width=9.2cm]{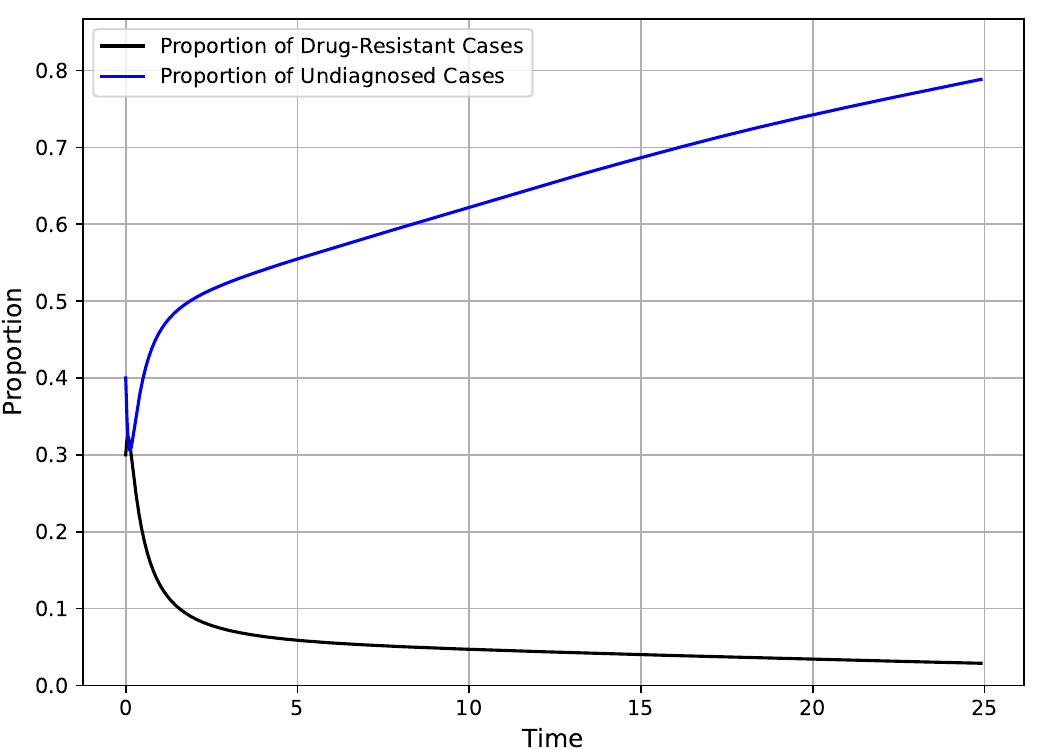}\label{03:fig:numerical_simulation_strategy_C_2}}\\
\end{center}
\begin{center}
\subfloat[Infected population dynamics]{\includegraphics[height=6.5cm,width=9.2cm]{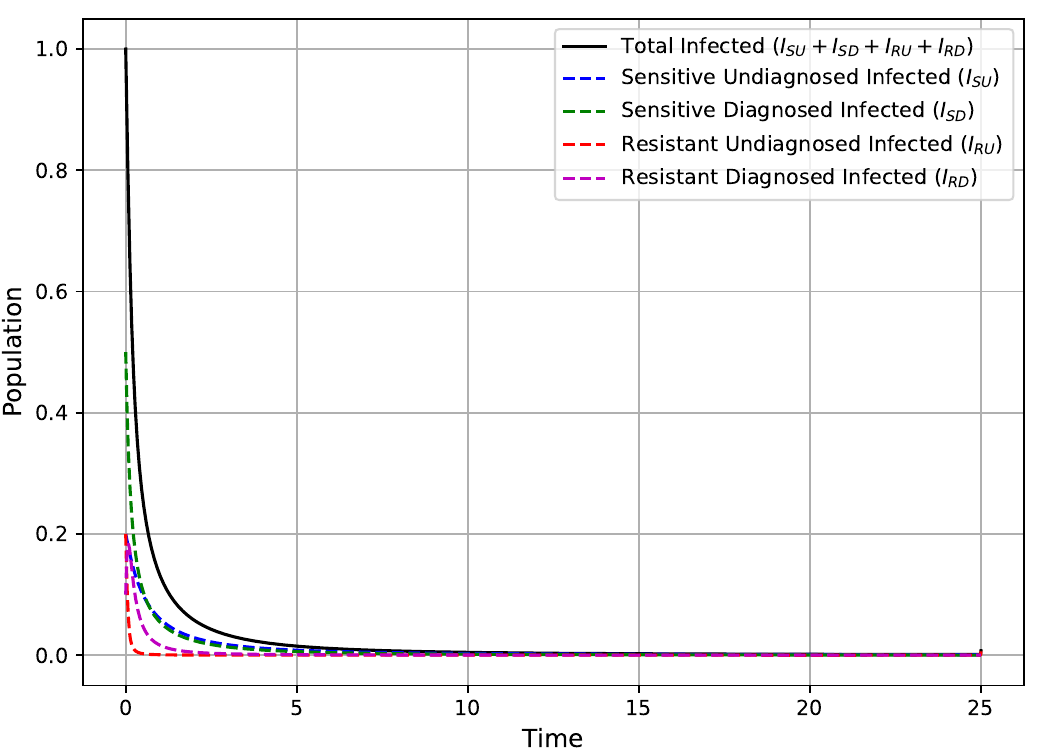}\label{03:fig:numerical_simulation_strategy_C_3}}
\quad 
\subfloat[Treatment coverage dynamics]{\includegraphics[height=6.5cm,width=9.2cm]{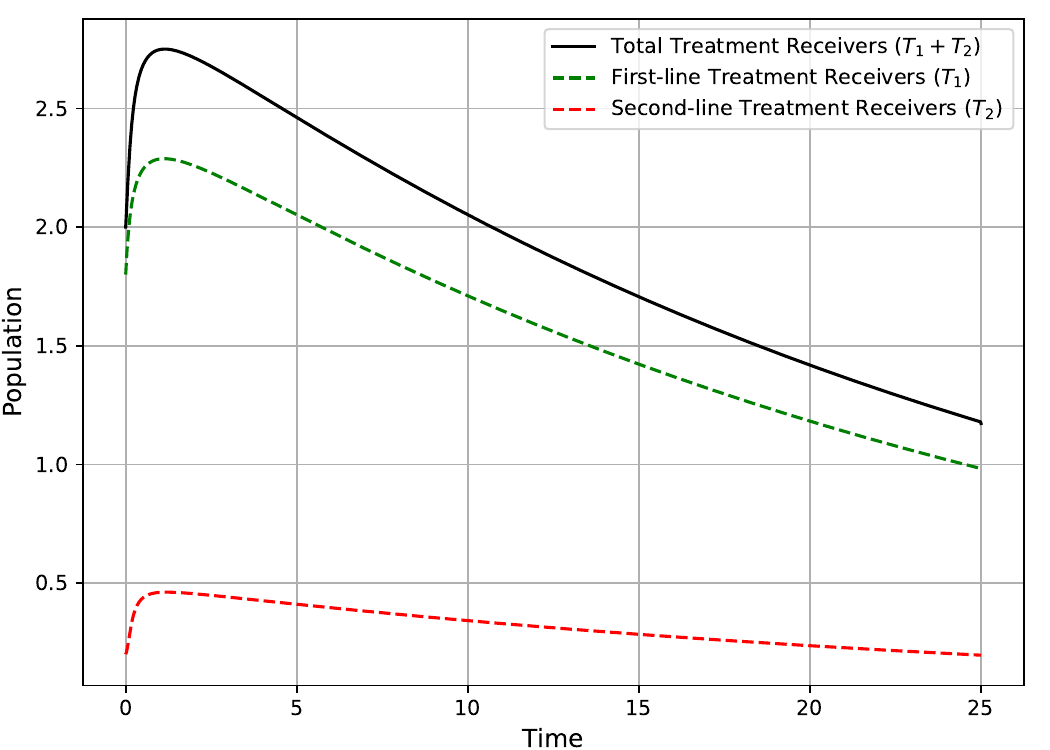}\label{03:fig:numerical_simulation_strategy_C_4}}\\
\end{center}
\caption{State and optimal control dynamics of the optimal control problem (\ref{03:eqn:control_functional})-(\ref{03:system:control}) for Strategy C.}\label{03:fig:numerical_simulation_strategy_C}
\end{figure}	

In Figure \ref{03:fig:numerical_simulation_strategy_D_1}, the optimal control profile shows a balanced approach by targeting all control variables at different stages of the simulation. In the early phase, all controls are applied simultaneously to maximize the intervention impact. As the number of individuals receiving first-line treatment falls below a certain threshold, the adherence-related control loses its relevance and reduces to zero to achieve the optimality. At the same time, these additional efforts are redirected toward improving the diagnosis of drug-resistant infections and expanding second-line treatment coverage, in response to the sudden increase in drug-resistant cases resulting from non-adherence to first-line treatment (see Figure \ref{03:fig:numerical_simulation_strategy_D_2}, Figure \ref{03:fig:numerical_simulation_strategy_D_3}). Accordingly, a significant change in treatment coverage dynamics is observed (see Figure \ref{03:fig:numerical_simulation_strategy_D_4}).

\begin{figure}[t!]
\begin{center}
\subfloat[Optimal control profile]{\includegraphics[height=6.4cm,width=9cm]{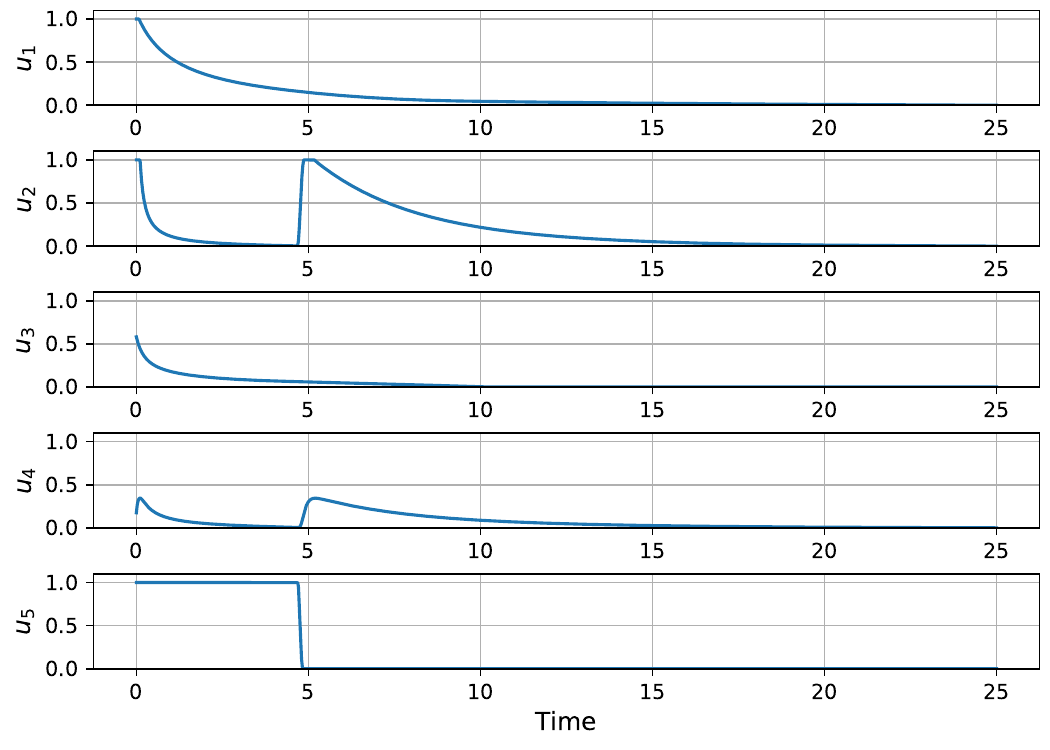}\label{03:fig:numerical_simulation_strategy_D_1}}
\quad
\subfloat[Epidemiological indicators]{\includegraphics[height=6.5cm,width=9.2cm]{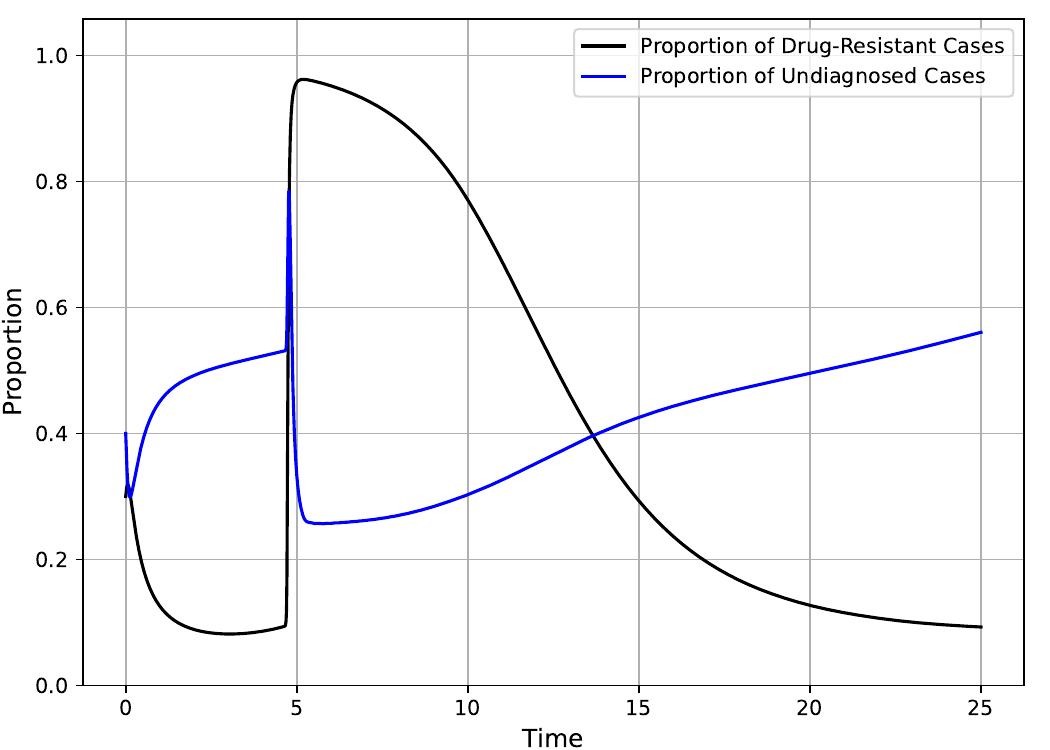}\label{03:fig:numerical_simulation_strategy_D_2}}\\
\end{center}
\begin{center}
\subfloat[Infected population dynamics]{\includegraphics[height=6.5cm,width=9.2cm]{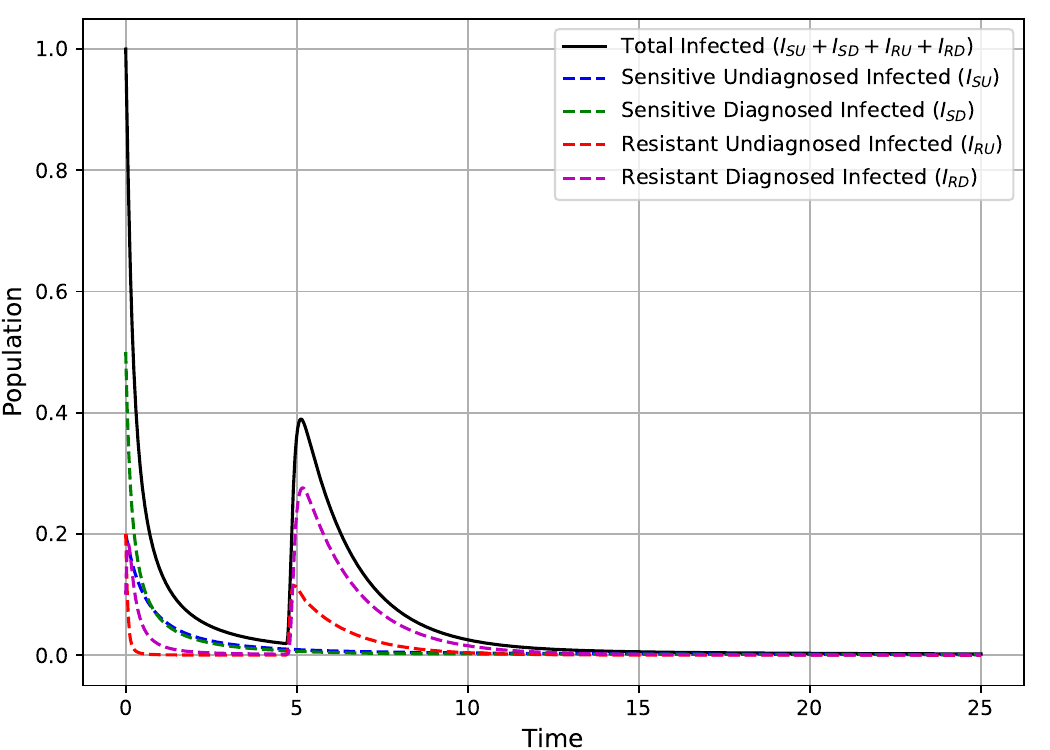}\label{03:fig:numerical_simulation_strategy_D_3}}
\quad 
\subfloat[Treatment coverage dynamics]{\includegraphics[height=6.5cm,width=9.2cm]{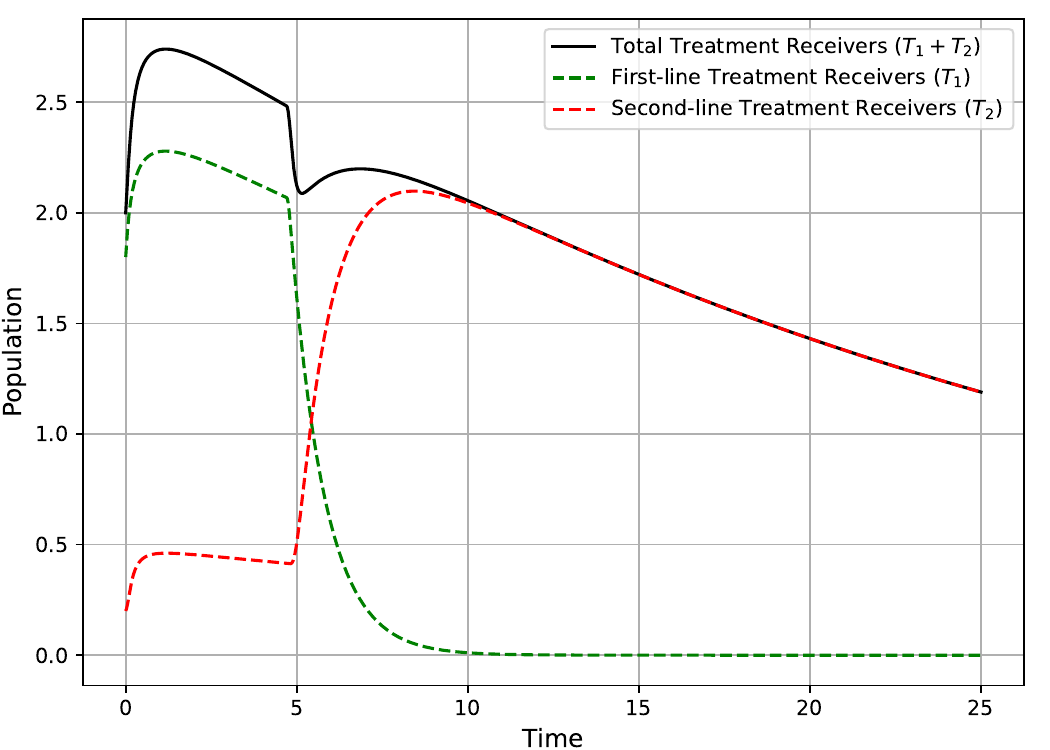}\label{03:fig:numerical_simulation_strategy_D_4}}\\
\end{center}
\caption{State and optimal control dynamics of the optimal control problem (\ref{03:eqn:control_functional})-(\ref{03:system:control}) for Strategy D.}\label{03:fig:numerical_simulation_strategy_E}
\end{figure}	

In the dynamic control optimization framework, decision points are set at one-year intervals considering the long-term dynamics of HIV infection. At each decision point, controls are updated based on feedback from the current system state. The optimal control profile in Figure \ref{03:fig:numerical_simulation_strategy_95_1} shows that the initial efforts starts from the diagnosis focused interventions to achieve the 95-95-95 targets. After a few strategy adjustments, all three targets are achieved within the first five years, followed by the implementation of a balanced strategy for the remaining simulation period. The color bar at the bottom indicates the specific strategy implemented during each corresponding time interval (see Figure \ref{03:fig:numerical_simulation_strategy_95_2}). Due to sudden changes in control inputs at each decision point, the state variables exhibit non-smooth trajectories. The results also suggest that while second-line treatment is more optimal than first-line treatment, the required coverage is lower compared to Strategies A, B, and D. (see Figure \ref{03:fig:numerical_simulation_strategy_95_2}). This dynamic approach suggests that the 95-95-95 targets can be achieved within five years through optimal resource allocation. However, although these targets are reached by the third year, sustaining them requires continuous and judicious allocation of resources across all aspects of HIV management.

\begin{figure}[t!]
\begin{center}
\subfloat[Optimal control profile]{\includegraphics[height=6.3cm,width=8.9cm]{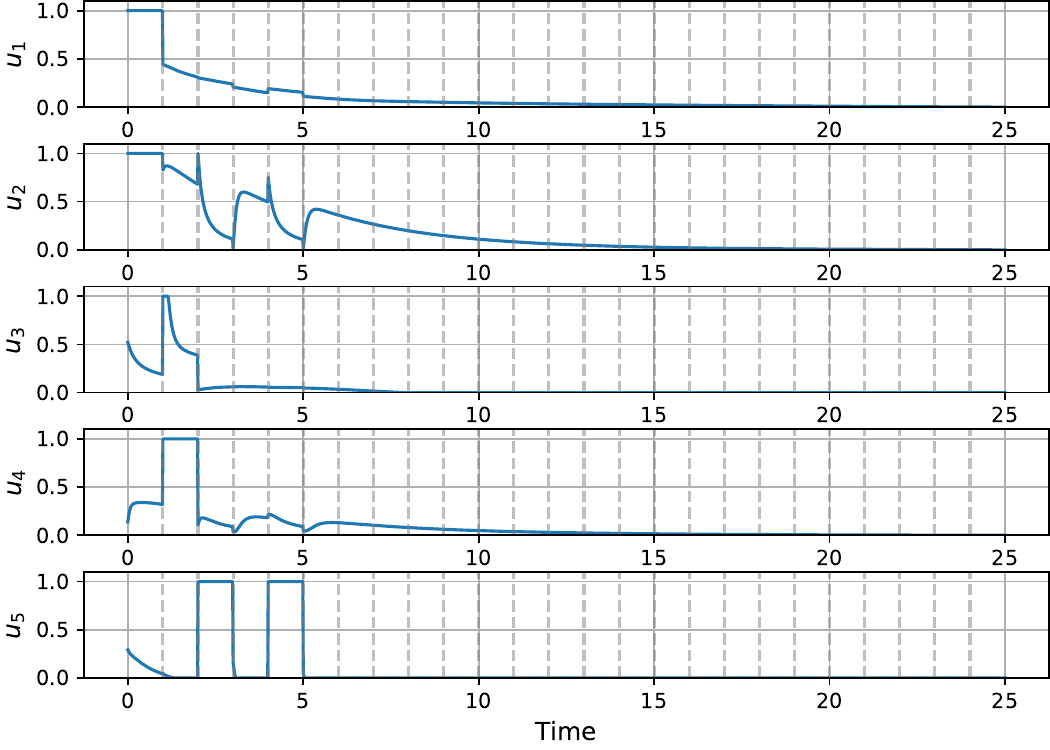}\label{03:fig:numerical_simulation_strategy_95_1}}
\quad
\subfloat[Epidemiological indicators]{\includegraphics[height=6.5cm,width=9.2cm]{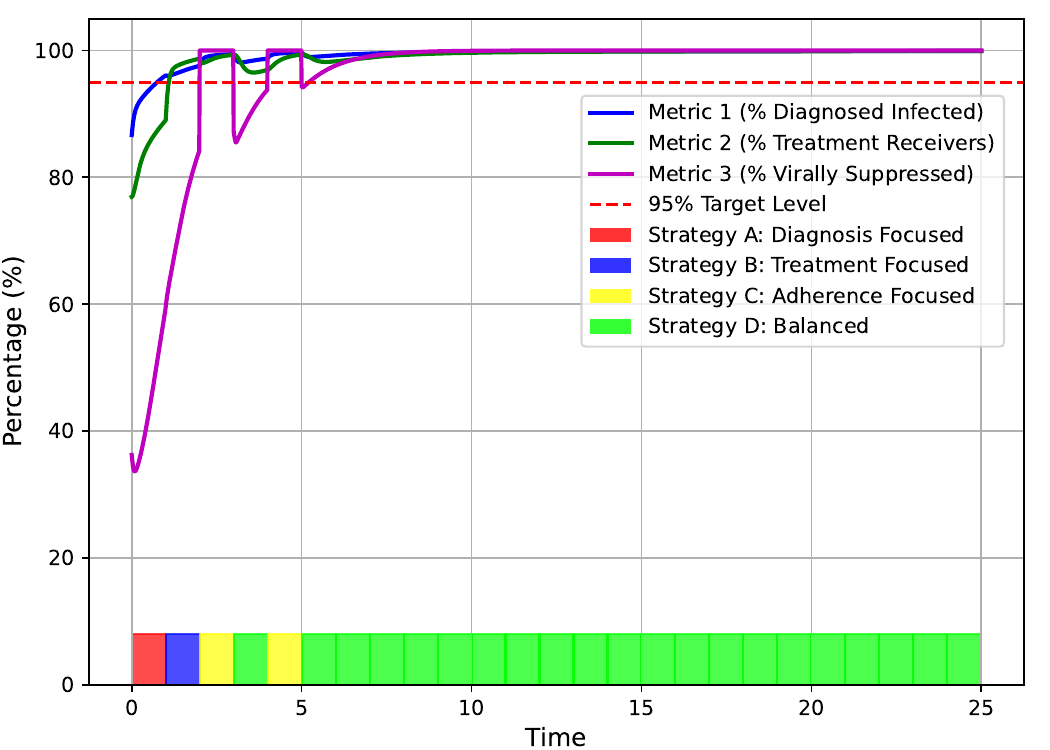}\label{03:fig:numerical_simulation_strategy_95_2}}\\
\end{center}
\begin{center}
\subfloat[Infected population dynamics]{\includegraphics[height=6.5cm,width=9.2cm]{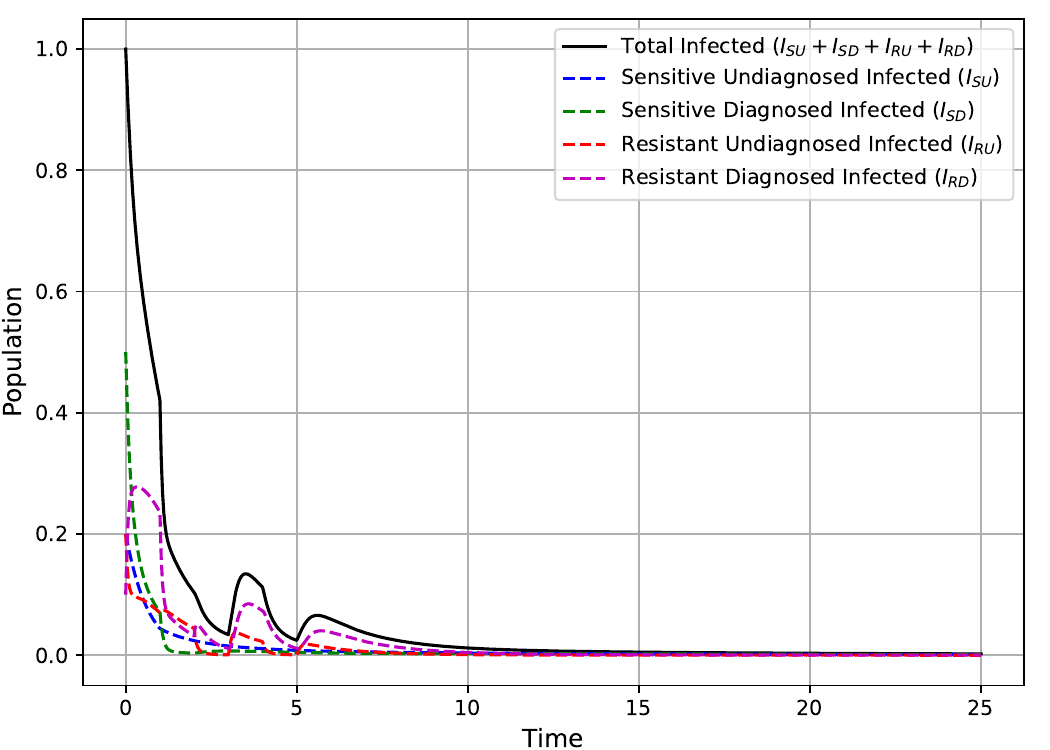}\label{03:fig:numerical_simulation_strategy_95_3}}
\quad 
\subfloat[Treatment coverage dynamics]{\includegraphics[height=6.5cm,width=9.2cm]{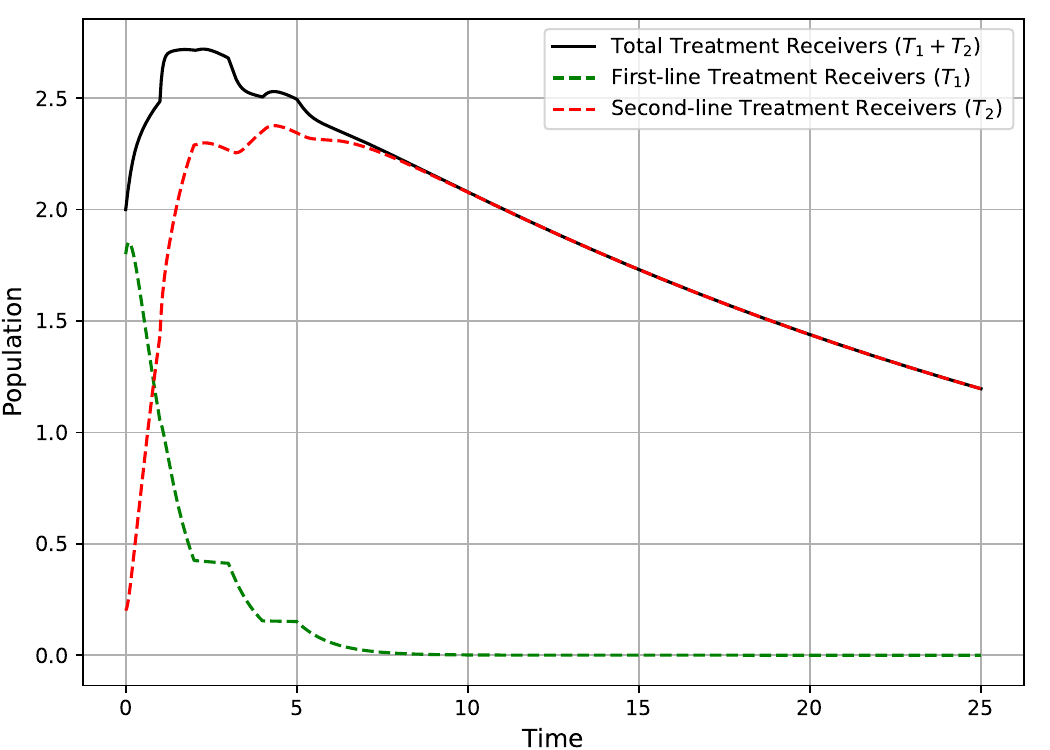}\label{03:fig:numerical_simulation_strategy_95_4}}\\
\end{center}
\caption{State and optimal control dynamics of the optimal control problem (\ref{03:eqn:control_functional})-(\ref{03:system:control}) for Strategy E.}\label{03:fig:numerical_simulation_strategy_95}
\end{figure}	

We further compared the total number of infected individuals and those receiving treatment under each control strategy with the baseline scenario of no intervention. The corresponding population dynamics for each strategy and the no-control case are illustrated in Figure \ref{03:fig:numerical_simulation_strategy_comparison}. Among all strategies, Strategy C shows highest reduction in infection burden, averting approximately 26.69 million person-years of infection, whereas Strategy D is the least effective, with 25.9 million person-years averted. In terms of treatment coverage, Strategy B leads with 44.02 million person-years of additional treatment, while Strategy D again shows the lowest impact, contributing 42.54 million person-years of treatment. A detailed comparison of infection aversion and treatment additions across strategies is provided in Table \ref{03:tab:numerical_simulation_ACER}. 

\begin{figure}[t!]
\begin{center}
\subfloat[Total infected population under different strategies and no control]{\includegraphics[height=6.3cm,width=8.9cm]{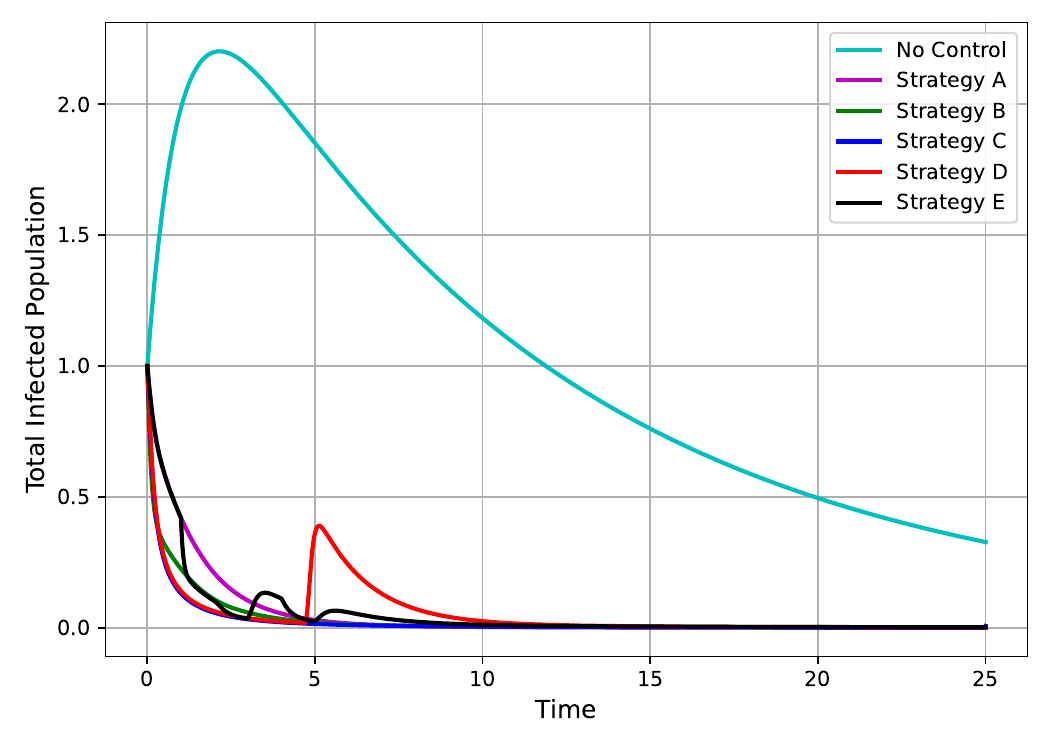}\label{03:fig:numerical_simulation_strategy_comparison_infected}}
\quad
\subfloat[Total treated population under different strategies and no control]{\includegraphics[height=6.5cm,width=9.2cm]{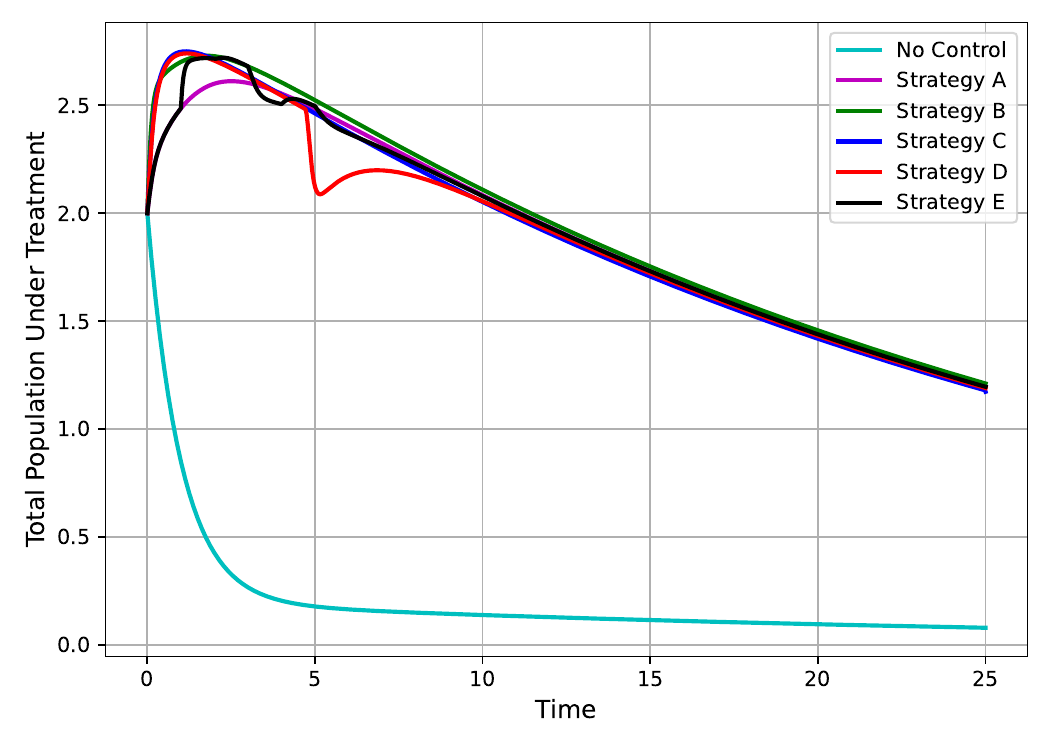}\label{03:fig:numerical_simulation_strategy_comparison_treatment}}\\
\end{center}
\caption{Comparative population dynamics of infected and treated individuals under different control strategies and the no-intervention scenario over the simulation period.}\label{03:fig:numerical_simulation_strategy_comparison}
\end{figure}	

\subsection{Cost-effectiveness analysis}{\label{03:subsec:numerical_simulation:costeffectiveness_analysis}}
Efficient allocation of resources is crucial for the effective control and elimination of a disease, particularly in resource-limited settings and low-income countries. Cost-effectiveness analysis is a valuable tool for identifying relatively low-cost interventions with substantial potential to reduce the burden of disease \cite{okos2013, agus2019}. In this section, we evaluate the cost-effectiveness of the control strategies proposed in Section \ref{03:subsec:numerical_simulation:optimal_control} and rank them according to their effectiveness in reducing infection levels with optimal resource utilization. For this purpose, we consider the average cost-effectiveness ratio (ACER) method \cite{agus2019}. The strategy with the lowest ACER value is ranked as the most cost-effective strategy.

In general, the ACER is defined as the ratio of the total cost factor associated with a control strategy to its corresponding impact on disease reduction. However, accurately estimating the exact cost and impact of implementing a strategy, particularly over a large population and long time horizon, is often impractical. Therefore, we employ proxy metrics to represent these estimates that are most suitable for our modeling framework. The impact of a control strategy can be measured by metrics such as averted total infections, averted person-years of infection, or additional person-year of treatment. We use the total averted person-year of infection and additional person-year of treatment, as these are more suitable metrics for the proposed model, which account for the duration individuals remain in a state. It is calculated as the difference in the cumulative number of infected or treated individual with and without control implementation during the simulation period. Further, the cost of a control strategy is calculated as the sum of the integrals of the squared optimal control values (control intensity), each weighted by its respective relative cost factor (RCF), over the simulation period.

The cost estimates for various control components depend on regional economic conditions and the mode of implementation. Using a relative scale for these estimates reduces this variance. The costs associated with $u_{1}$ and $u_{2}$ include expenses for testing, infrastructure development, and public awareness campaigns \cite{sund2022}. Among these, diagnosis of drug-resistant cases is significantly more expensive due to the need for specialized tests and advanced equipments. HIV drug-resistance testing at Metropolis Healthcare Limited, one of India's leading diagnostic providers, can cost 5–20 times more than standard HIV testing, depending on the method used. However, this disparity decreases when considering the cost of all components of efforts combined. Accordingly, we set the $RIF(u_{1})=1$ as the base unit and $RIF(u_{2})=5$. Treatment-related efforts, represented by $u_{3}$ and $u_{4}$, involve training healthcare professionals, expanding treatment facilities, and offering subsidized therapy \cite{muga2011}. In particular, second-line ART incurs higher costs due to more complex protocols and specialized care needs. Based on this, we assign $RIF(u_{3})=10$ and $RIF(u_{4})=15$. Further, efforts to improve drug adherence, such as individual counseling, educational campaigns, reminder systems via SMS or mobile apps, and strengthening medicine supply chains, are major contributors for the cost of $u_{5}$ \cite{mbua2015}. We assume their relative cost as $RCF(u_{5})=2$. Note that these values are not exact but serve as approximations to reflect the relative cost burden of each control measure. 

Based on the control intensities and RCF of each control variable,  and the control impact of each strategy, the ACER for a strategy is given by:
$$\text{ACER} = \frac{\text{Total Cost of Interventions}}{\text{Total Strategy Impact}} = \frac{\sum_{i=1}^{5}\text{(Control Intensity)}_i \times \text{(RCF}_i)}{\text{Total Strategy Impact}}$$

For each strategy, Table \ref{03:tab:numerical_simulation_ACER} presents the control intensity for each control, the total cost of interventions, and the respective impact on each infection and treatment classes. Based on these values, the ACER is computed in terms of $(i)$ total person-years of infection averted and $(ii)$ total person-years of treatment added. The total infection-related impact of a strategy is calculated by adding the averted person-years across all infected classes. To reflect the real-world cost implications in treatment impact, particularly in resource-limited settings, we assign a higher weight to first-line treatment in the ACER calculation for treatment coverage. This extra weight accounts for the significant cost difference between first- and second-line ART, where second-line ART can be 3 to 9 times more expensive depending on the regional economic state \cite{matt2022}. Consequently, the total impact of a strategy on added person-years of treatment is computed as the sum of 5 times the added person-years for $T_1$ (assumed for low- and lower-middle-income countries) and the added person-years for $T_2$. Based on this, the ACER for person-year of averted infection and added treatment are listed in Table \ref{03:tab:numerical_simulation_ACER}. 

The ACER values indicate that Strategy D is the most cost-efficient in reducing the infection burden, despite averting the fewest person-years of infection. The lower total implementation cost in this strategy contributes to its high cost-effectiveness. In contrast, Strategy B ranks lowest in cost-efficiency due to the high expenses associated with treatment-focused interventions. While Strategy B significantly reduces the infection burden, this high cost limits its cost-effectiveness. Also, Strategies E, A, and C rank second, third, and fourth, respectively, in terms of cost-effectiveness for infection reduction.

However, when the goal is to increase treatment coverage, Strategy C becomes the most cost-effective. This is because a large share of resources in Strategy C is allocated to adherence-related interventions, enhancing the effectiveness of first-line treatment and reducing the need for more expensive second-line treatment. Similarly, Strategy D, which also prioritizes early focus on in adherence improving interventions, proves cost-effective for improving treatment outcomes, though it is slightly costly then Strategy C. The dynamic control approach, which emphasizes first-line treatment more than Strategies A and B, also demonstrates better cost-efficiency compared to those strategies. Although Strategy B achieves the highest overall treatment coverage, the lack of adherence-enhancing interventions makes it the least cost-effective. Overall, these ACER values suggest that in settings where expanding treatment coverage is feasible, particularly in middle- and upper-middle-income countries, investing in adherence-improving interventions is a more cost-effective approach. Otherwise, a balanced allocation of resources across all intervention components results in a more sustainable and efficient approach.

\begin{table}[htbp] 
    \centering
    \footnotesize 
    \setlength{\tabcolsep}{3pt} 

    \begin{tabular}{@{}l *{5}{c} c *{6}{c} *{2}{c}@{}}
        \toprule
        \multirowthead{2}{Strategies} & 
        \multicolumn{5}{c}{\thead{Control Intensity}} &
        \multirowthead{2}{\makecell{Total \\ Cost}} &
        \multicolumn{6}{c}{\thead{Strategy Impact}} &
        \multicolumn{2}{c}{\thead{ACER}} \\ 

        \cmidrule(lr){2-6} \cmidrule(lr){8-13} \cmidrule(lr){14-15} 
        & \makecell{$u_{1}$} & \makecell{$u_{2}$} & \makecell{$u_{3}$} & \makecell{$u_{4}$} & \makecell{$u_{5}$} 
        & 
        & \makecell{$I_{SU}$} & \makecell{$I_{SD}$} & \makecell{$I_{RU}$} & \makecell{$I_{RD}$} & \makecell{$T_{1}$} & \makecell{$T_{2}$} 
        & \makecell{Infection} & \makecell{Treatment} \\ 
        \midrule

        Strategy A & 3.9842 & 7.3509 & 0.1322 & 0.3003 & 0.0247 & 46.6137 & 2.5841 & 4.0793 & 19.1606 & 0.1837 & 0.8931 & 42.0576 & 1.7923 & 1.0019\\
        Strategy B & 0.9636 & 2.3083 & 1.4431 & 3.9535 & 0.0203 & 86.2788 & 2.4519 & 4.2578 & 19.1076 & 0.6748 & 0.8828 & 43.1396 & 3.2568 & 1.8143\\
        Strategy C & 1.1014 & 0.2416 & 0.1663 & 0.0600 & 25.0000 & 54.8731 & 2.4760 & 4.1091 & 19.3202 & 0.7818 & 38.2447 & 4.7120 & 2.0562 & 0.2801\\
        Strategy D & 1.0069 & 2.1653 & 0.1434 & 0.2972 & 4.7432 & 27.2111 & 2.4556 & 4.0749 & 19.1341 & 0.2335 & 10.6799 & 31.8639 & 0.7089 & 0.3191\\
        Strategy E & 1.3043 & 2.4865 & 0.4899 & 1.2210 & 2.0251 & 41.0008 & 2.4878 & 4.1075 & 19.1107 & 0.3205 & 1.4350 & 41.6223 & 1.5754 & 0.8402\\
        \bottomrule
    \end{tabular}
\caption{Summary of control intensities, total costs, impacts on infection and treatment classes, and ACER values for all control strategies. 
ACER values are computed considering both infection reduction and treatment coverage expansion objectives.}\label{03:tab:numerical_simulation_ACER}
\end{table}

\subsection{Adjoint-based sensitivity analysis of state variables}{\label{03:subsec:numerical_simulation:adjoint_based_sensitivity_analysis}}
After identifying the most cost-effective control strategy, both in terms of averted infection burden and enhanced treatment coverage, the logical next step is to determine how to optimally allocate any additional resources to maximize public health benefits. This is relevant for particularly middle- and upper-middle income countries, where budgets may permit allocations of resources beyond the cost-effectiveness goals. To address this, we conducted an adjoint-based sensitivity analysis of each state variable with respect to each control variable. This analysis provides forward-looking insights into how small perturbations in control variables influence the cumulative population in each compartment. Controls with higher sensitivity indices, positive for the susceptible and treatment classes and negative for the infected and AIDS classes, are identified as potential priorities for additional investment. A detailed discussion of adjoint-based sensitivity analysis for differential-algebraic equations can be found in \cite{cao2003}.  

In this section, we analyze the sensitivities of state variables under Strategy C and Strategy D, as these strategies are the most cost-effective for enhancing treatment coverage and reducing infection burden, respectively. In this approach, we begin with solving the optimal control problem (\ref{03:eqn:control_functional})-(\ref{03:system:control}) for the chosen strategy to obtain the optimal control profile $u^{*}_{i}(t)$ and the corresponding state trajectories $x^{*}_{i}(t)$ over the time interval $[0,t_{f}]$. Subsequently, we formulate a new adjoint system for the state variable of interest as follows:
\begin{equation} \label{03:eqn:numerical_analysis_adjoint_sensitivity_analysis}
    \dfrac{d \sigma_{i}}{dt}=-\sum_{j=1}^{8} \sigma_{j} \dfrac{\partial f_{j}}{\partial x_i}-\dfrac{\partial \mathcal{X}}{\partial x_{i}}, \quad \sigma_{i}(t_f)=0
\end{equation}
Here, $f_{j}$ denotes the $j^{th}$ component of the right-hand side of system (\ref{03:system:control}) and $\sigma_{i}(t)$ is the adjoint variable that captures how small perturbations in a state variable $x_{i}(t)$ propagate through time to affect the state of interest $\mathcal{X}(t)$. The adjoint system is solved backward in time using the previously obtained optimal control profiles $u^{*}_{i}(t)$  and the corresponding state trajectories $x^{*}_{i}(t)$. The sensitivities of $\mathcal{X}(t)$ to a control variable $u_{k}$ is then computed using the adjoint variables $\sigma_{i}(t)$ and the system equations, as follows:
\begin{eqnarray} \label{03:eqn:numerical_analysis_sensitivities}
     \dfrac{\partial \mathcal{X}}{\partial u_{k}}&=&\int_{0}^{t_{f}} \sum_{i=1}^{8} \dfrac{\partial \mathcal{X}}{\partial x_{i}} \dfrac{\partial f_{i}}{\partial u_{k}} dt \nonumber \\
     &=& \int_{0}^{t_{f}} \sum_{i=1}^{8} \sigma_{i} \dfrac{\partial f_{i}}{\partial u_{k}} dt
\end{eqnarray}
Note that these sensitivities indicate the cumulative effect of small perturbations in the control $u_{k}(t)$ on the state of interest $\mathcal{X}(t)$.  This effect is mediated through the state variables $x_{i}(t)$, whose contributions are captured by the corresponding adjoint variables $\sigma_{i}(t)$ for $i=1,2,...,8.$ 

The adjoint-based sensitivity indices of each state variable with respect to the control variables for Strategy C and Strategy D are presented in Table \ref{03:tab:numerical_simulation_Adjoint_SA}. The relative significance of the control variables is almost similar across these two strategies. All control variables exhibit a positive impact on the susceptible population and a negative impact on the AIDS-infected population, with $u_{5}$ showing the highest influence compared to the other controls. The undiagnosed infected population with the sensitive strain is primarily negatively influenced by $u_{1}$, which governs the diagnosis rate. Further, $u_{3}$ also shows a small negative effect on $I_{SU}$, due to its strong negative impact on $I_{SD}$ population, thereby indirectly reducing the prevalence of the sensitive strain. Although $u_{1}$ tends to increase $I_{SD}$ by converting $I_{SU}$ to $I_{SD}$, this is overpowered by the strong negative influence of $u_{3}$ on $I_{SD}$. The controls $u_{2}$, $u_{4}$, and $u_{5}$ do not show any impact on both $I_{SU}$ and $I_{SD}$. Overall, after achieving the cost-effectiveness objective, control $u_{3}$ emerges as the most suitable target for further resource allocation to reduce the drug-sensitive infected population in both strategies

The first-line treatment coverage is most positively influenced by control $u_{5}$, which enhances adherence to treatment in both strategies. An Improved adherence helps infected individuals to continue with first-line treatment and reduces the emergence of drug resistance. Consequently, $u_{5}$ depicts a strong negative impact on the compartments $I_{RU}$, $I_{RD}$, and $T_{2}$. In contrast, $u_{2}$ and $u_{4}$ have negligible influence on these compartments compared to $u_{5}$. In Strategy C, controls $u_{1}$ and $u_{3}$ have no impact on the drug-resistant and second-line treatment compartments, as adherence-improving interventions (through $u_{5}$) are applied at their full capacity, and the resistant infections primarily arise through transmission. However, in Strategy D, adherence interventions are focused only in the early phase of the simulation. As a result, drug resistance can also emerge during treatment, making the compartments $I_{RU}$, $I_{RD}$, and $T_{2}$ sensitive to changes in $u_{1}$ and $u_{3}$. Nevertheless, this sensitivity remains insignificant when compared to the dominant influence of $u_{5}$. Overall, the drug-resistant infected and treatment-receiving populations are predominantly sensitive to control $u_{5}$.

\begin{table}[htbp] 
    \centering
    \footnotesize 
    \setlength{\tabcolsep}{3pt} 
    \begin{tabular}{@{}l *{5}{c} *{5}{c} *{5}{c} *{5}{c}@{}}
        \toprule
        \multirowthead{2}{State \\ Variables} & 
        \multicolumn{5}{c}{\thead{Strategy C}} &
        \multicolumn{5}{c}{\thead{Strategy D}} \\ 

        \cmidrule(lr){2-6} \cmidrule(lr){7-11} 
        & \makecell{$u_{1}$} & \makecell{$u_{2}$} & \makecell{$u_{3}$} & \makecell{$u_{4}$} & \makecell{$u_{5}$}  
        & \makecell{$u_{1}$} & \makecell{$u_{2}$} & \makecell{$u_{3}$} & \makecell{$u_{4}$} & \makecell{$u_{5}$}  \\ 
        \midrule
        $S$ & 0.5359 & 0.0179 & 0.3353 & 0.0731 & 27.1014 & 0.6158 & 0.0998 & 0.4373 & 0.2438 & 3.0096\\
        $I_{SU}$ & -1.0808 & 0 & -0.0844 & 0 & 0 & -1.3015 & 0 & -0.1706 & 0 & 0 \\
        $I_{SD}$ & 0.0825 & 0 & -2.0983 & 0 & 0 & 0.1934 & 0 & -3.6801 & 0 & 0 \\
        $T_{1}$ & 1.3701 & 0 & 3.9012 & 0 & 409.6203 & 0.4349 & 0 & 1.9751 & 0 & 34.0293 \\
        $I_{RU}$ & 0 & -0.0679 & 0 & -0.0125 & -156.1592 & 0.0073 & -0.4012 & 0.4274 & -0.0104 & -9.7815 \\
        $I_{RD}$ & 0 & -0.0143 & 0 & -1.1285 & -63.9923 & 0.0215 & -0.0223 & 0.4016 & -4.0039 & -3.6434 & \\
        $T_{2}$ & 0 & 0.1489 & 0 & 2.1617 & -104.9142 & 1.0426 & 0.7183 & 3.3622 & 7.3432 & -9.5087 \\
        $A$ & -0.1136 & -0.0076 & -0.2043 & -0.1005 & -20.6713 & -0.1283 & -0.0424 & -0.3142 & -0.3717 & -1.1643 \\
        \bottomrule
    \end{tabular}
\caption{Adjoint-based sensitivity indices of state variables with respect to each control variable under Strategy C and Strategy D. }\label{03:tab:numerical_simulation_Adjoint_SA}
\end{table}

\begin{remark}
Note that we have used $\dfrac{\partial f_{i}}{\partial u_{k}}$ as an approximation of $\dfrac{\partial x_{i}}{\partial u_{k}}$ while applying the chain rule in (\ref{03:eqn:numerical_analysis_sensitivities}). This is because the control $u_{k}$ directly modifies the rate of change of the state $x_{i}$. Thus, $\dfrac{\partial f_{i}}{\partial u_{k}}$ captures the direct and instantaneous effect of $u_{k}$ on $x_{i}$.
\end{remark}

\subsection{Control contribution analysis}{\label{03:subsec:numerical_simulation:control_contribution}}
In the previous sections, cost-effectiveness and adjoint-based sensitivity analyses helped identify the most efficient strategies and key controls under cost-constrained settings. Although these analyses offer critical insights into the relative efficiency and influence of different strategies, they do not fully capture the individual and synergistic contributions of each control variable. In this section, we quantify the contribution of each control within a given strategy to better understand its specific role, both direct and indirect, in achieving the overall public health objectives, particularly in scenarios where the cost of intervention is not a limiting factor. This analysis provides a broader and more actionable perspective for public health decision-makers, allowing informed allocation of resources.

We use the Shapley Value analysis \cite{lipo2001, lund2017, cava2021}, a tool from cooperative game theory, to achieve this objective. In game theory, the Shapley value analysis provides a systematic approach to fairly distribute the total outcome of a cooperative game among its players based on their individual marginal impacts. In the context of this study, the `players' are the five control variables, and the `outcome of the game' is the change observed in the state variables with and without the application of these controls. Shapley values quantify the average marginal contribution of each control to a given state variable, calculated across all possible combinations of the remaining controls. This ensures that each control's contribution is evaluated not only in isolation but also in the context of its interactions with all possible subsets of the remaining controls, providing a comprehensive idea about how the effectiveness of one control depends on others.

The calculation of Shapley values involves a combinatorial process that considers all possible coalitions of control variables. For a specific control $u_i$ within a set of $n$ controls $(U)$, its Shapley value $\phi(u_i)$, corresponding to the given state variable, is computed as the weighted average of its marginal contributions across all subsets $S$ of $U$ that exclude $u_i$. Mathematically, this is expressed as: 
$$\phi(u_i) = \sum_{S \subseteq U \setminus {u_i}} \frac{|S|! (n - |S| - 1)!}{n!} \left( v(S \cup {u_i}) - v(S) \right),$$  
where $v(S)$ denotes the value related to the given state variable under the coalition $S$, and the term $v(S \cup {u_i}) - v(S)$ represents the marginal contribution of control $u_i$ when added to coalition $S$. The weighting factor $\displaystyle{\frac{|S|! (n - |S| - 1)!}{n!}}$ ensures fairness by accounting for all possible positions that control $u_{i}$ could occupy in a permutation of the full set of controls. This normalizes the contributions based on coalition size and ensures that each subset is weighted appropriately.  In this study, with five control variables, the process involves evaluating all $2^5 = 32$ possible control combinations, making it computationally feasible while providing a robust and precise attribution of each control’s impact on the state variables. 

We conducted the Shapley value analysis for Strategy C and Strategy D, given their significance in terms of cost-effectiveness. The objective of this analysis is to quantify the contribution of each control variable to the observed person-year differences in each state variable, calculated with and without the implementation of the respective control strategies. As described in Section \ref{03:subsec:numerical_simulation:costeffectiveness_analysis}, these person-year differences represent the overall impact of a control strategy on disease dynamics. The person-years for each state variables under different control combinations are summarized in Table \ref{03:tab:appendix:person_years_strategy_C} for strategy C and in Table \ref{03:tab:appendix:person_years_strategy_D} for strategy D (see \ref{03:appendix:sec:control_contribution}). Tables \ref{03:tab:numerical_simulation_shapley_values_strategy_C} and \ref{03:tab:numerical_simulation_shapley_values_strategy_D} present the Shapley value-based decomposition of control contributions to person-years in each state variable under Strategy C and Strategy D, respectively. The sign of the Shapley value indicates the direction of influence, while the percentage contribution, calculated using the absolute values, reflects the relative magnitude of each control’s contribution. The detailed procedure for computing these Shapley values is outlined in Algorithm \ref{03:alg:numerical_simulation_control_contribution_shapley_values}.

\begin{table}[htbp]
\centering
\small
\begin{tabular}{@{}l@{\hspace{3pt}}c@{\hspace{3pt}}c@{\hspace{3pt}}c@{\hspace{3pt}}c@{\hspace{3pt}}c@{\hspace{3pt}}c@{\hspace{3pt}}c@{\hspace{3pt}}c@{\hspace{3pt}}c@{\hspace{3pt}}c@{}}
\toprule
\multirowthead{2}{\makecell{\textbf{State} \\ \textbf{Variables}}} & \multicolumn{2}{c}{\textbf{$u_1$}} & \multicolumn{2}{c}{\textbf{$u_2$}} & \multicolumn{2}{c}{\textbf{$u_3$}} & \multicolumn{2}{c}{\textbf{$u_4$}} & \multicolumn{2}{c}{\textbf{$u_5$}} \\
\cmidrule(lr){2-3} \cmidrule(lr){4-5} \cmidrule(lr){6-7} \cmidrule(lr){8-9} \cmidrule(lr){10-11}
 & \makecell{\textbf{Shapley} \\ \textbf{Value}} & \makecell{\textbf{Contribution} \\ \textbf{(\%)}} & \makecell{\textbf{Shapley} \\ \textbf{Value}} & \makecell{\textbf{Contribution} \\ \textbf{(\%)}} & \makecell{\textbf{Shapley} \\ \textbf{Value}}  & \makecell{\textbf{Contribution} \\ \textbf{(\%)}} & \makecell{\textbf{Shapley} \\ \textbf{Value}}  & \makecell{\textbf{Contribution} \\ \textbf{(\%)}} & \makecell{\textbf{Shapley} \\ \textbf{Value}}  & \makecell{\textbf{Contribution} \\ \textbf{(\%)}} \\
\midrule
$S$ & 0.9267 & 14.78 & 2.1131 & 33.69 & 0.5175 & 8.25 & 0.3234 & 5.16 & 2.3913 & 38.13 \\
$I_{SU}$ & -2.1323 & 86.12 & -0.0272 & 1.10 & -0.1977 & 7.99 & -0.0273 & 1.10 & -0.0915 & 3.70 \\
$I_{SD}$ & 1.1018 & 17.45 & -0.4277 & 6.78 & -3.0734 & 48.69 & -0.4280 & 6.78 & -1.2818 & 20.30 \\
$T_1$ & 0.9134 & 2.39 & 0.0490 & 0.13 & 4.9313 & 12.89 & 0.0486 & 0.13 & 32.3024 & 84.46 \\
$I_{RU}$ & 0.0004 & 0.00 & -10.5257 & 54.45 & 0.0044 & 0.02 & -0.0237 & 0.12 & -8.7754 & 45.40 \\
$I_{RD}$ & 0.0021 & 0.02 & 6.2274 & 46.97 & 0.0097 & 0.07 & -4.2249 & 31.86 & -2.7950 & 21.08 \\
$T_2$ & 0.1218 & 0.44 & 6.8865 & 24.89 & 0.7438 & 2.69 & 8.4366 & 30.49 & -11.4794 & 41.49 \\
$A$ & -0.1352 & 5.16 & -0.5265 & 20.10 & -0.3185 & 12.16 & -0.4418 & 16.87 & -1.1976 & 45.71 \\
\bottomrule
\end{tabular}
\caption{Shapley values for control contribution to person-years in each state variable for Strategy C}
\label{03:tab:numerical_simulation_shapley_values_strategy_C}
\end{table}

\begin{table}[htbp]
\centering
\small
\begin{tabular}{@{}l@{\hspace{3pt}}c@{\hspace{3pt}}c@{\hspace{3pt}}c@{\hspace{3pt}}c@{\hspace{3pt}}c@{\hspace{3pt}}c@{\hspace{3pt}}c@{\hspace{3pt}}c@{\hspace{3pt}}c@{\hspace{3pt}}c@{}}
\toprule
\multirowthead{2}{\makecell{\textbf{State} \\ \textbf{Variables}}} & \multicolumn{2}{c}{\textbf{$u_1$}} & \multicolumn{2}{c}{\textbf{$u_2$}} & \multicolumn{2}{c}{\textbf{$u_3$}} & \multicolumn{2}{c}{\textbf{$u_4$}} & \multicolumn{2}{c}{\textbf{$u_5$}} \\
\cmidrule(lr){2-3} \cmidrule(lr){4-5} \cmidrule(lr){6-7} \cmidrule(lr){8-9} \cmidrule(lr){10-11}
 & \makecell{\textbf{Shapley} \\ \textbf{Value}} & \makecell{\textbf{Contribution} \\ \textbf{(\%)}} & \makecell{\textbf{Shapley} \\ \textbf{Value}} & \makecell{\textbf{Contribution} \\ \textbf{(\%)}} & \makecell{\textbf{Shapley} \\ \textbf{Value}}  & \makecell{\textbf{Contribution} \\ \textbf{(\%)}} & \makecell{\textbf{Shapley} \\ \textbf{Value}}  & \makecell{\textbf{Contribution} \\ \textbf{(\%)}} & \makecell{\textbf{Shapley} \\ \textbf{Value}}  & \makecell{\textbf{Contribution} \\ \textbf{(\%)}} \\
\midrule
$S$ & 0.9241 & 14.97 & 2.0822 & 33.74 & 0.5127 & 8.31 & 0.2925 & 4.74 & 2.3605 & 38.25 \\
$I_{SU}$ & -2.1285 & 86.68 & -0.0231 & 0.94 & -0.1936 & 7.88 & -0.0231 & 0.94 & -0.0873 & 3.56 \\
$I_{SD}$ & 1.1048 & 17.58 & -0.4199 & 6.68 & -3.0656 & 48.78 & -0.4202 & 6.69 & -1.2740 & 20.27 \\
$T_1$ & 0.7969 & 1.69 & -9.1105 & 19.33 & 4.9630 & 10.53 & -9.1103 & 19.33 & 23.1408 & 49.11 \\
$I_{RU}$ & 0.0017 & 0.01 & -10.4666 & 54.44 & 0.0097 & 0.05 & 0.0350 & 0.18 & -8.7139 & 45.32 \\
$I_{RD}$ & 0.0055 & 0.04 & 6.4014 & 48.83 & 0.0316 & 0.24 & -4.0510 & 30.90 & -2.6210 & 19.99 \\
$T_2$ & 0.2274 & 0.62 & 15.9301 & 43.36 & 0.6621 & 1.80 & 17.4802 & 47.58 & -2.4358 & 6.63 \\
$A$ & -0.1350 & 5.12 & -0.5328 & 20.19 & -0.3184 & 12.07 & -0.4482 & 16.99 & -1.2040 & 45.64 \\
\bottomrule
\end{tabular}
\caption{Shapley values for control contribution to person-years in each state variable for Strategy D}
\label{03:tab:numerical_simulation_shapley_values_strategy_D}
\end{table}

For Strategy C, which prioritizes adherence-focused interventions, Table \ref{03:tab:numerical_simulation_shapley_values_strategy_C} highlights major contributions from control $u_{5}$ across multiple state variables, especially in the treatment compartments.  The Shapley values for $T_1$ and $T_2$ associated with $u_5$ are significantly higher (32.30 and –11.48, respectively) compared to other controls. This indicates that strong adherence support significantly influences the accumulation of person-years in treatment stages, enhancing retention in first-line therapy while reducing person-years in second-line treatment due to fewer treatment failures or transitions. In contrast, Strategy D, the balanced control approach, distributes the contribution more evenly across multiple controls. Here, although $u_5$ still plays a key role, particularly for $T_1$ (49.11\%) in negative direction, the influence of $u_2$ becomes significantly more pronounced for $T_2$ (43.36\%) and $I_{RU}$ (54.44\%), as compared to Strategy C. In addition, the contribution of $u_4$ also increases significantly for both $T_{1}$ and $T_{2}$. This redistribution indicates that under a balanced strategy, interventions targeting both adherence support to first-line treatment and resistance management collectively shape the epidemiological results. The switch from a high negative contribution of $u_5$ in $T_2$ under Strategy C (–11.48) to a comparatively small negative value (–2.44) under Strategy D further supports the idea that spreading the control effort reduces the intensity of adherence-related interventions,  consequently increasing the demand for second-line treatment. Furthermore, in both strategies, the primary control affecting $I_{SU}$ is $u_{1}$, while $I_{SD}$ is mainly influenced by $u_{3}$. Whereas, control $u_{2}$ plays a major role in shaping the resistant compartments, $I_{RU}$ and $I_{RD}$.

The synergetic effects arising from the interaction of multiple control variables across state compartments are also effectively captured by these Shapley values. For instance, controls $u_2$ and $u_4$, which target the diagnosis and treatment of drug-resistant infections, also contribute to reducing the burden of sensitive infections. This occurs because effective management of the resistant pool indirectly suppresses the overall transmission dynamics. As resistant and sensitive strains often compete for the same susceptible population, a reduction in the resistant strain can lower the overall force of infection, which benefits the control of sensitive infections. In contrast, controls $u_1$ and $u_3$ exhibit a positive contribution to the resistant infected populations, suggesting that the forward flow of individuals through the compartments, such as progression from undiagnosed to diagnosed and then to treatment failure, dominates over any competitive suppression between sensitive and resistant strains. Although these synergetic effects are relatively small, they can significantly enhance the overall effectiveness of intervention strategies. In addition, the combined implementation of $u_2$ and $u_4$ significantly reduces the drug resistance burden, compared to their individual effect. Overall, both strategies highlight the dominant influence of control $u_5$, which remains the primary driver in shaping infection dynamics.

\begin{algorithm}
\caption{Control Contribution Analysis using Shapley Values}
\label{03:alg:numerical_simulation_control_contribution_shapley_values}
\begin{algorithmic}[1]
\REQUIRE Simulation time $t_f$, model parameters $\theta$, initial conditions $y_0$, strategy weights $W$, control bounds $U_b = [0, 1]^5$, convergence tolerance $\epsilon$, maximum iterations $k_{\text{max}}$, relaxation parameter $\alpha$
\STATE Initialize: Generate all subsets $S$ of the control set $U = \{u_{1},u_{2}, u_{3},u_{4}, u_{5}\}$, i.e., the power set $\mathcal{P}(U)$.
\FOR{each subset $S \in \mathcal{P}(N)$}
    \STATE Set controls in $U \setminus S$ to 0.
    \STATE Solve the optimal control problem from 0 to $t_f$ using backward-forward sweep:
    \FOR{$k = 1$ to $k_{\text{max}}$}
        \STATE Solve state system $\dot{y} = f(y, u^{(k-1)}, \theta)$ from $0$ to $t_f$ with $y(0)=y_{0}$
        \STATE Solve adjoint system $\dot{\sigma} = -\nabla_y H(y, u^{(k-1)}, \sigma, \theta, W)$ from $t_f$ to $0$ with $\sigma(t_f)=0$
        \STATE Compute $u^{\text{new}}$ using $\frac{\partial H}{\partial u_i} = 0$ for $i = 1, \dots, 5$
        \STATE Update $u^{(k)} \gets (1 - \alpha) u^{(k-1)} + \alpha u^{\text{new}}$
        \IF{$\max_i \| u_i^{(k)} - u_i^{(k-1)} \|_\infty < \epsilon$}
            \STATE Set $u^* \gets u^{(k)}$
            \STATE \textbf{break}
        \ENDIF
    \ENDFOR
    \STATE Simulate the model over $[0, t_f]$ to obtain state trajectories $y^{(j)}_S(t)$ for $j=1, \dots ,8$.
    \STATE Compute person-years for each state variable $y^{(j)}_S(t)$:
    \[
    v_S^{(j)} = \int_{0}^{t_f} y^{(j)}_S(t) dt, \quad j = 1, \ldots, 8.
    \]
\ENDFOR
\FOR{each control $u_{i} \in U$ and state variable $y^{(j)}$}
    \STATE Calculate the Shapley value:
    \[
    \phi^{(j)}(u_{i}) = \sum_{S \subseteq U \setminus \{u_{i}\}} \frac{|S|! (n - |S| - 1)!}{n!} \left( v_{S \cup \{u_{i}\}}^{(j)} - v_S^{(j)} \right), \quad j=1, \dots ,8; ~n = 5. 
    \]
\ENDFOR
\RETURN The Shapley values $\phi^{(j)}(u_{i})$ for all controls $u_i \in U$ and state variables $y^{(j)} \in y$.
\end{algorithmic}
\end{algorithm}

\section{Conclusion}{\label{03:sec:conclusion}}

Early diagnosis of HIV followed by timely initiation of the treatment is critical for reducing infection burden and its onward transmission. Over the past few decades, ART has played a pivotal role in controlling HIV by reducing viral load and improving the quality of life of infected individuals. However, the full benefits of ART are contingent on patient's adherence level to treatment. A well-established paradox in the public health domain is: should treatment be universally provided despite concerns about sub-optimal adherence, or should resources be strategically allocated to achieve a balance between expanding treatment coverage and strengthening adherence support? Sub-optimal adherence to ART not only compromises individual health outcomes but also accelerates the emergence and transmission of drug-resistant HIV strains. In such cases, a switch to second-line treatment becomes necessary after the diagnosis of resistance development. However, this transition is often challenged by diagnostic delays, limited availability, and higher associated costs. 

In this study, we developed a compartmental model to explore the transmission dynamics of HIV, incorporating both drug-sensitive and drug-resistant strains, diagnosis status, and treatment switching from first-line to second-line therapy upon resistance detection. The proposed model admits one disease-free equilibrium and two endemic equilibria, where the drug-sensitive strain may be eliminated, while the drug-resistant strain persists at a positive level. Using the next-generation matrix method, we derived the basic reproduction numbers corresponding to both the drug-sensitive and drug-resistant strains. The existence and stability conditions of these equilibrium points highlight the critical role of these basic reproduction numbers in deciding the long-term dynamics of the system. Additionally, the relative transmission dynamics of the two strains significantly influences the outcome: higher transmission rates of the sensitive strain or relatively lower rates for the resistant strain can create the possibility for coexistence of both strains, even at a higher treatment coverage. 

A detailed sensitivity analysis has been conducted for the proposed model in order to identify the most influential parameters in the short-term and long-term dynamics. The PRCC based sensitivity analysis of basic reproduction numbers suggests that the transmission rates of both strains are main governing factors in the early stage of the infection. Time-varying sensitivity analysis was performed to identify the key parameters influencing model dynamics across different phases of disease progression, offering valuable insights for designing targeted intervention strategies. We employed the Sobol method to compute both first-order and total-order sensitivity indices. A significant difference observed between the first- and total-order indices for most parameters highlights the strong influence of parameter interactions on the total number of infections. The direct impact of individual parameters becomes minimal after a certain time, and the variance in model outcomes becomes mainly driven by complex interactions among parameters.

We extended the proposed modeling framework by incorporating optimal control theory to evaluate a set of intervention strategies aimed at improving diagnosis, treatment coverage, and adherence to ART optimally. Instead of using binary control activation, we introduced a weight-varying mechanism in the cost functional to differentiate between various intervention strategies in a more realistic manner. We proposed five different control strategies targeting key epidemiological components: (A) diagnosis focused, (B) treatment focused, (C) adherence focused, (D) a balanced approach, and (E) a dynamic optimization based approach. Among these, Strategy C, centered on improving adherence levels, produced the highest reduction in infection burden, particularly in the drug-resistant population, showing the crucial role of optimal adherence in controlling HIV transmission. Notably, a higher initiation rate of second-line treatment was supported across all strategies except Strategy C. In this case, the first-line treatment remained more effective due to the optimal allocation of resources toward adherence improvement, reducing the need to switch to second-line therapy. Despite this, Strategy C outperformed all other strategies in averting total infection cases. In contrast, Strategy D with the balanced approach proved to be the least effective in both infection reduction and increasing treatment coverage, suggesting that distributing resources evenly may reduce the potential impact of targeted interventions. Strategy B demonstrated the highest improvement in treatment coverage, although it was less effective in reducing overall infections compared to adherence-focused approach. This indicates that if expanding treatment access is the primary goal, Strategy B offers a more effective pathway, provided sufficient resources are available. Also, the dynamic control strategy shows that the 95-95-95 global targets can be achieved within five years through optimally timed and weighted resource allocation. Although these targets are met by the end of the third year, maintaining them requires sustained and judicious investment in all three aspects, which are diagnosis, treatment initiation, and adherence. This highlights the need for dynamic, adaptable control policies to ensure long-term success in HIV management.

Efficient resource allocation is essential for the effective control and elimination of infectious diseases, particularly in resource-constrained settings. To evaluate the economic viability of the proposed intervention strategies, we performed a cost-effectiveness analysis using two key metrics: total averted person-years of infection and additional person-years of treatment provided. Strategy D emerged as the most cost-efficient in reducing infection burden, despite averting the fewest person-years of infection, while Strategy B ranked lowest due to the high costs associated with treatment-focused interventions. Strategies E, A, and C rank second, third and fourth, respectively, in terms of cost-effectiveness for infection reduction. When the objective shifts to expanding treatment coverage, Strategy C emerged as the most cost-effective option. This is attributed to its strategic allocation of resources towards adherence-focused interventions, which improves the efficacy of first-line therapy and reduces the need for more expensive second-line treatment. Strategy D is cost-effective in improving treatment outcomes, although it is slightly more expensive than Strategy C. The dynamic control strategy, which favors first-line treatment more than Strategies A and B, also demonstrated superior cost-efficiency compared to those strategies. In contrast, despite achieving the highest treatment coverage, Strategy B remained the least cost-effective due to its lack of adherence-enhancing components.

Beyond cost-effectiveness goals, we conducted an adjoint-based sensitivity analysis to identify optimal directions for additional resource allocation, with a focus on middle- and upper-middle-income countries. This analysis was applied to Strategies C and D, as these were the most cost-effective strategies, and examined the sensitivity of each state variable with respect to each control variable. Results indicate that increasing first-line treatment coverage is the most impactful intervention to reduce the drug-sensitive infected population. Additionally, the dynamics of drug-resistant infections and treatment-receiving populations were found to be highly sensitive to adherence-enhancing interventions. These findings suggest that, when considering broader public health objectives beyond cost constraints, prioritizing first-line treatment and adherence support produces the highest returns for improving public health outcomes. To better understand the individual and combined contributions of each control within a given strategy, particularly in settings where intervention costs are not a constraint, we used Shapley value analysis from cooperative game theory. This approach quantified both direct and synergistic effects of each control on achieving public health goals. The analysis revealed that controls targeting the diagnosis and treatment of drug-resistant infections also contributed to reducing the burden of drug-sensitive cases by indirectly suppressing overall transmission dynamics. However, some controls exhibited unintended amplification of resistance due to forward flow through compartments. The joint implementation of targeted interventions proved more effective than isolated efforts. Among all controls, adherence-focused intervention consistently emerged as the most influential in shaping long-term epidemic outcomes.

In conclusion, our findings provide valuable insights into the intricate relationship between treatment adherence, the emergence of drug resistance, and the effectiveness of control interventions in shaping the dynamics of the HIV epidemic. While multiple targeted interventions can collectively reduce the burden of both drug-sensitive and drug-resistant HIV strains, the long-term success of epidemic control critically depends on prioritizing adherence-focused strategies alongside efforts to expand first-line treatment coverage. In settings where expanding treatment coverage is feasible, particularly in middle- and upper-middle-income countries, investing in adherence-enhancing interventions emerges as a more cost-effective strategy. Where such expansion is limited, a balanced allocation of resources across all intervention components provides a more sustainable and efficient alternative. 

\appendix
\renewcommand{\thesection}{Appendix \Alph{section}}
\renewcommand{\theequation}{\Alph{section}.\arabic{equation}}

\section{Convergence of PRCC values}
\label{03:appendix:sec:convergence_analysis}
\setcounter{equation}{0}
Figure \ref{03:fig:gsa_prcc_convergence} presents the convergence analysis of PRCC values for the global sensitivity analysis of the basic reproduction numbers $R_{0}^{(S)}$, $R_{0}^{(R)}$, and $R_{0}^{(SR)}$. For each outcome, the figure shows the three most influential parameters based on their absolute PRCC values across increasing sample sizes. Solid lines with circular markers represent the PRCC estimates, while dashed lines with square markers represent the mean Sequential Relative Change Index (RCI), which quantifies the relative change in PRCC values between consecutive sample sizes. The sample size increases in increments of 500. The RCI is computed as
\begin{equation*}
\text{RCI} = 
\displaystyle \left| \frac{\text{PRCC}_{\text{current}} - \text{PRCC}_{\text{previous}}}{\text{PRCC}_{\text{previous}}} \right| 
\end{equation*}
This metric serves as a convergence criterion, with RCI values below 0.05 typically indicating stabilization of sensitivity estimates. Shaded regions represent $95\%$ confidence intervals, obtained from bootstrap resampling $(n = 2,000)$. This analysis demonstrates that a sample size of approximately 5,000 or more is required to achieve convergence, beyond which both PRCC values and their confidence intervals stabilize. 

\begin{figure}
    \centering
    \includegraphics[width=1\linewidth]{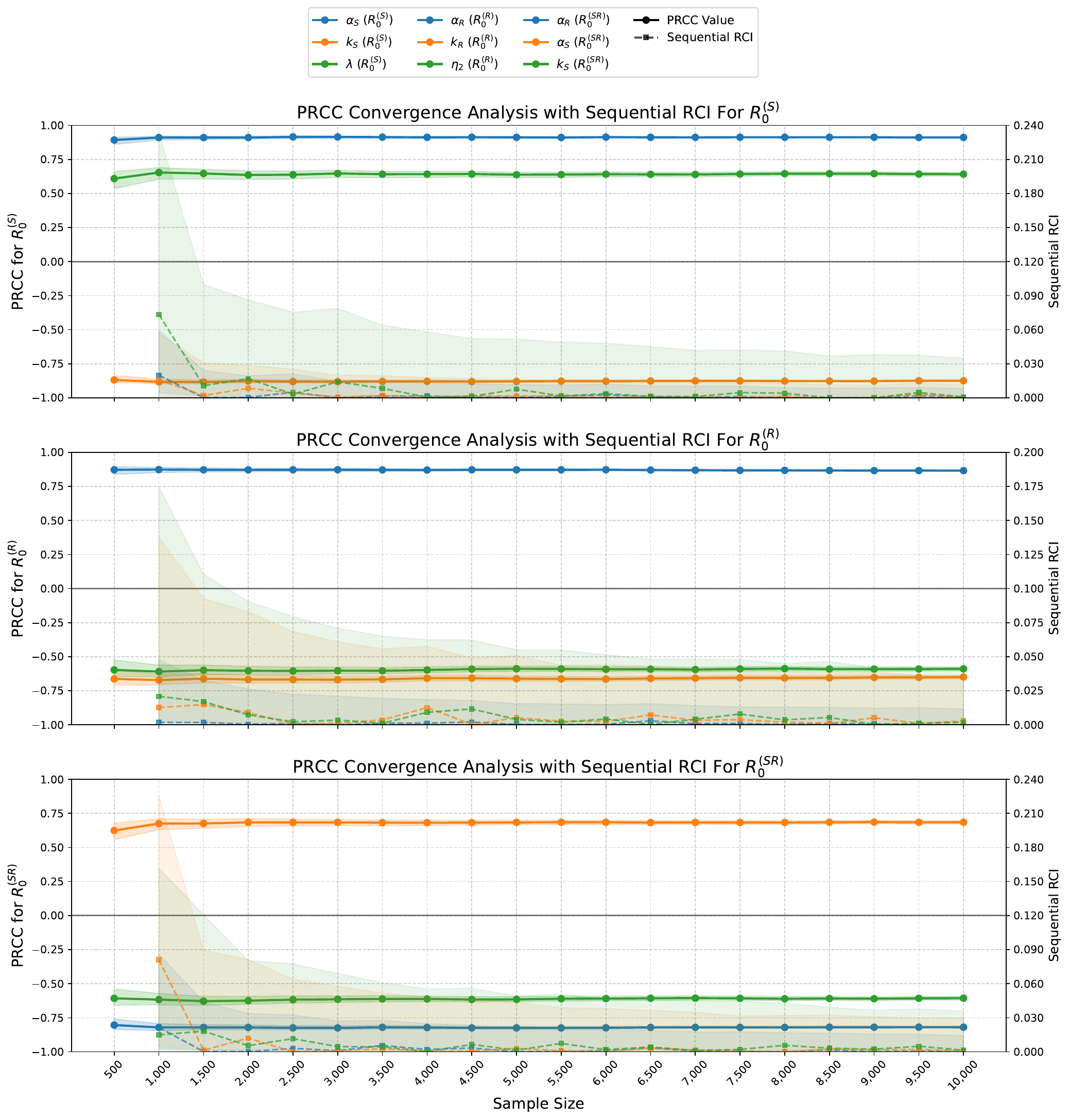}
    \caption{Convergence analysis illustrating the stability of PRCC values with increasing sample size for the three most influential parameters affecting each reproduction number. Solid lines with circle represent mean PRCC values, while dashed lines with squares show mean Sequential RCI, quantifying the relative change in PRCC estimates between consecutive sample sizes. The shaded region indicates 95\% CI derived from 2000 bootstrap resamples at each sample size. }
    \label{03:fig:gsa_prcc_convergence}
\end{figure}

\section{Optimal control analysis}{\label{03:appendix:sec:optimal_control_analysis}}
\textbf{Proof of Theorem \ref{03:thm:existence_optimal_control}:}
\begin{proof}
We rewrite the controlled state system (\ref{03:system:control}) as follows:
\begin{equation}
    G(X)=CX+F(X)
\end{equation}
where $G(X)=\begin{bmatrix}
    S^{\prime}\\
    I^{\prime}_{SU}\\
    I^{\prime}_{SD}\\
    T^{\prime}_{1}\\
    I^{\prime}_{RU}\\
    I^{\prime}_{RD}\\
    T^{\prime}_{2}\\
    A^{\prime}
\end{bmatrix}, $ 
$X=\begin{bmatrix}
    S\\
    I_{SU}\\
    I_{SD}\\
    T_{1}\\
    I_{RU}\\
    I_{RD}\\
    T_{2}\\
    A
\end{bmatrix},$
$F(X)=\begin{bmatrix}
    \lambda-\alpha_S S(I_{SU}+c I_{SD})-\alpha_R S (I_{RU}+ c I_{RD})\\
    \alpha_S S(I_{SU}+c I_{SD})\\
    0\\
    0\\
    \alpha_R S (I_{RU}+ c I_{RD})\\
    0\\
    0\\
    0
\end{bmatrix},$\\
and 
$$\resizebox{\textwidth}{!}{$C= \begin{bmatrix}
    -\mu & 0 & 0 & 0 & 0 & 0 & 0 & 0\\
     0   & -(k_{S_{max}} u_{1}+\theta_{1}+\mu) & 0 & 0 & 0 & 0 & 0 & 0\\
     0   & k_{S_{max}} u_{1} & -(\beta_{S} u_{3}+\theta_{1}(1-u_3)+\mu) & 0 & 0 & 0 & 0 & 0\\
     0   & 0 & \beta_{S} u_{3} & -(\gamma(1-u_5)+\theta_2 u_5+\mu) & 0 & 0 & 0 & 0\\
     0   & 0 & 0 & \gamma(1-u_5) & -(k_{R_{max}}u_2+\theta_1+\mu) & 0 & 0 & 0\\
     0   & 0 & 0 & 0 & k_{R_{max}} u_2 & -(\beta_{R} u_4+\theta_1(1-u_4)+\mu) & 0 & 0\\
     0   & 0 & 0 & 0 & 0 & \beta_{R} u_4 & -(\theta_2+\mu) & 0\\
     0   & \theta_1 & \theta_{1}(1-u_3) & \theta_2 u_5 & \theta_1 & \theta_1(1-u_4) & \theta_2 & -(\mu+\mu_d)
\end{bmatrix}$}.$$
Since the solutions of the system (\ref{03:system:control}) remain bounded for each bounded control variable in $\mathcal{U}$, applying the triangular inequality will give us
$$|F(X_{1})-F(X_{2})|\leq p_1 |S_1-S_2|+p_2|I_{SU_{1}}-I_{SU_{2}}| + p_3|I_{SD_{1}}-I_{SD_{2}}|
+ p_4|I_{RU_{1}}-I_{RU_{2}}| + p_5|I_{RD_{1}}-I_{RD_{2}}|,$$
where $p_{i}>0~(i=1,2,3,4,5)$ are constants. Consequently,
\begin{eqnarray}
     |G(X_{1})-G(X_{2})| & \leq & C|X_{1}-X_{2}|+|F(X_{1})-F(X_{2})| \nonumber \\
    & \leq & p|X_{1}-X_{2}|< \infty, \nonumber
\end{eqnarray}
where $p=||C||+\sum_{i=1}^{5}p_{i} < \infty.$ Therefore, the right-hand side functions of the control system (\ref{03:system:control}) are uniformly Lipschitz continuous. Consequently, by applying the Picard-Lindel\"of theorem  \cite{codd1956}, the system (\ref{03:system:control}) admits a unique solution, ensuring that condition (i) is satisfied.

Condition (ii) holds trivially, as the control set $\mathcal{U}$ is defined to be closed and convex. Also, the system (\ref{03:system:control}) is linear with respect to each control variable, with coefficients that depend only on the state variables. 

For condition (iii), the integrand L is quadratic in nature with respect to each control variable, and hence, it is convex. Further, since $T_{1}$ and $T_{2}$ are bounded above, there exists a maximum value $M$ such that $0\leq W_{5} T_{1}+ W_{6} T_{2} \leq M.$ Therefore, 
\begin{eqnarray}
    L(I_{SU}, I_{SD}, I_{RU}, I_{RD}, T_{1}, T_{2}, u_{1}, u_{2}, u_{3}, u_{4}, u_{5}) &=& W_{1} I_{SU} + W_{2} I_{SD} + W_{3} I_{RU} + W_{4} I_{RD} - W_{5} T_{1} - W_{6} T_{2} \nonumber \\&& + \frac{1}{2} \sum_{i=1}^{5}w_{i} u_{i}^{2}\nonumber \\
    &\geq & -M + \frac{1}{2} \sum_{i=1}^{5}w_{i} u_{i}^{2} \nonumber \\
    & \geq & -M + w_{m}\left(\sum_{i=1}^{5} u_{i}^{2}  \right), \nonumber
\end{eqnarray}
where $\displaystyle{w_{m}=\frac{1}{2}\min\{w_{1}, w_{2}, w_{3}, w_{4}, w_{5}\}}$. Let $\displaystyle{g(u_{1}, u_{2}, u_{3}, u_{4}, u_{5})=-M + w_{m}\left(\sum_{i=1}^{5} u_{i}^{2} \right)}$, which is continuous and satisfies $|(u_{1}, u_{2}, u_{3}, u_{4}, u_{5})|^{-1} g(u_{1}, u_{2}, u_{3}, u_{4}, u_{5}) \rightarrow \infty$ whenever $|(u_{1}, u_{2}, u_{3}, u_{4}, u_{5})| \rightarrow \infty$, since the negative term consisting $M$ vanishes while the positive term increases without bound. Therefore, the condition (iii) is also holds, ensuring the existence of an optimal control solution for the control system (\ref{03:system:control}).

\end{proof}

\section{Control contribution analysis}{\label{03:appendix:sec:control_contribution}}
Tables \ref{03:tab:appendix:person_years_strategy_C} and \ref{03:tab:appendix:person_years_strategy_D} present the cumulative person-years spent in each state variable under all control combinations for Strategy C and Strategy D, respectively. These person-year distributions serve as the basis for computing the marginal contributions of each control variable. These contributions are subsequently used to calculate the Shapley values, which provide a measure of the individual and synergistic impact of each control on the model outcomes.

\begin{table}[htbp]
\centering
\small
\begin{tabular}{@{}l@{\hspace{5pt}}c@{\hspace{9pt}}c@{\hspace{9pt}}c@{\hspace{9pt}}c@{\hspace{9pt}}c@{\hspace{9pt}}c@{\hspace{9pt}}c@{\hspace{9pt}}c@{}}
\toprule
\textbf{Control Combination} & \textbf{S} & \textbf{I$_{SU}$} & \textbf{I$_{SD}$} & \textbf{T$_1$} & \textbf{I$_{RU}$} & \textbf{I$_{RD}$} & \textbf{T$_2$} & \textbf{A} \\
\midrule
$(0,0,0,0,0)$ & 21806.0240 & 2.7202 & 4.3542 & 1.7988 & 19.3376 & 0.8708 & 3.2637 & 5.9055 \\
$({u_1,0,0,0,0})$ & 21806.8647 & 0.5731 & 6.0694 & 1.7988 & 19.3377 & 0.8708 & 3.2637 & 5.8290 \\
$({0,u_2,0,0,0})$ & 21809.3619 & 2.7204 & 4.3542 & 1.7988 & 0.2398 & 18.3332 & 3.2637 & 5.6092 \\
$({0,0,u_3,0,0})$ & 21806.0238 & 2.7199 & 4.3506 & 1.7993 & 19.3414 & 0.8708 & 3.2637 & 5.9055 \\
$({0,0,0,u_4,0})$ & 21806.0778 & 2.7202 & 4.3542 & 1.7988 & 19.3117 & 0.0442 & 4.8138 & 5.8268 \\
$({0,0,0,0,u_5})$ & 21809.9808 & 2.7204 & 4.3542 & 29.3721 & 1.9883 & 0.8708 & 3.2637 & 4.2294 \\
$({u_1, u_2, 0,0,0})$ & 21810.2029 & 0.5731 & 6.0695 & 1.7988 & 0.2398 & 18.3332 & 3.2637 & 5.5327 \\
$({u_1,0, u_3,0,0}$ & 21806.8647 & 0.5731 & 6.0691 & 1.7990 & 19.3377 & 0.8708 & 3.2637 & 5.8290 \\
$({u_1,0,0, u_4,0})$ & 21806.9186 & 0.5731 & 6.0694 & 1.7988 & 19.3118 & 0.0442 & 4.8138 & 5.7503 \\
$({u_1,0,0,0, u_5})$ & 21810.8219 & 0.5731 & 6.0695 & 29.3721 & 1.9883 & 0.8708 & 3.2637 & 4.1528 \\
$({0,u_2, u_3,0,0})$ & 21809.3619 & 2.7204 & 4.3541 & 1.7989 & 0.2398 & 18.3332 & 3.2637 & 5.6092 \\
$({0, u_2,0, u_4,0})$ & 21810.4071 & 2.7204 & 4.3542 & 1.7988 & 0.1803 & 0.5459 & 35.3498 & 3.9411 \\
$({0,u_2,0,0, u_5})$ & 21810.3301 & 2.7204 & 4.3542 & 29.3736 & 0.0306 & 2.6599 & 3.2637 & 4.1987 \\
$({0,0,u_3, u_4,0})$ & 21806.0776 & 2.7200 & 4.3506 & 1.7993 & 19.3155 & 0.0442 & 4.8138 & 5.8267 \\
$({0,0,u_3,0, u_5})$ & 21810.7239 & 2.3556 & 0.1365 & 37.3270 & 1.9877 & 0.8708 & 3.2637 & 3.7849 \\
$({0,0,0,u_4, u_5})$ & 21810.0347 & 2.7204 & 4.3542 & 29.3721 & 1.9624 & 0.0442 & 4.8138 & 4.1506 \\
$({u_1, u_2, u_3,0,0})$ & 21810.2029 & 0.5731 & 6.0692 & 1.7991 & 0.2398 & 18.3332 & 3.2637 & 5.5327 \\
$({u_1, u_2,0, u_4,0})$ & 21811.2482 & 0.5731 & 6.0695 & 1.7988 & 0.1803 & 0.5459 & 35.3498 & 3.8646 \\
$({u_1, u_2,0,0, u_5})$ & 21811.1712 & 0.5731 & 6.0695 & 29.3736 & 0.0306 & 2.6599 & 3.2637 & 4.1222 \\
$({u_1,0, u_3, u_4,0})$ & 21806.9186 & 0.5731 & 6.0691 & 1.7990 & 19.3118 & 0.0442 & 4.8138 & 5.7503 \\
$({u_1,0, u_3,0, u_5})$ & 21811.7916 & 0.2443 & 0.2451 & 40.0431 & 1.9874 & 0.8708 & 3.2637 & 3.5545 \\
$({u_1,0,0, u_4, u_5})$ & 21810.8757 & 0.5731 & 6.0695 & 29.3721 & 1.9624 & 0.0442 & 4.8138 & 4.0741 \\
$({0,u_2, u_3, u_4,0})$ & 21811.1308 & 2.3573 & 0.1562 & 2.2836 & 0.2135 & 0.6366 & 42.8143 & 3.4896 \\
$({0,u_2, u_3,0, u_5})$ & 21811.0731 & 2.3557 & 0.1365 & 37.3278 & 0.0306 & 2.6599 & 3.2637 & 3.7543 \\
$({0,u_2,0, u_4, u_5})$ & 21810.4854 & 2.7204 & 4.3542 & 29.3736 & 0.0173 & 0.0890 & 7.9757 & 3.9606 \\
$({0,0,u_3, u_4, u_5})$ & 21810.7778 & 2.3556 & 0.1365 & 37.3271 & 1.9618 & 0.0442 & 4.8138 & 3.7061 \\
$({u_1, u_2, u_3, u_4,0})$ & 21812.1756 & 0.2697 & 0.3075 & 2.4480 & 0.2278 & 0.6773 & 45.2534 & 3.2645 \\
$({u_1, u_2, u_3,0, u_5})$ & 21812.1408 & 0.2443 & 0.2451 & 40.0434 & 0.0306 & 2.6599 & 3.2637 & 3.5240 \\
$({u_1, u_2,0, u_4, u_5})$ & 21811.3264 & 0.5731 & 6.0695 & 29.3736 & 0.0176 & 0.0900 & 7.9733 & 3.8841 \\
$({u_1,0, u_3, u_4, u_5})$ & 21811.8455 & 0.2443 & 0.2451 & 40.0429 & 1.9616 & 0.0442 & 4.8138 & 3.4758 \\
$({0, u_2, u_3, u_4, u_5})$ & 21811.2283 & 2.3557 & 0.1365 & 37.3278 & 0.0176 & 0.0901 & 7.9731 & 3.5162 \\
$({u_1, u_2, u_3, u_4, u_5})$ & 21812.2960 & 0.2443 & 0.2451 & 40.0434 & 0.0176 & 0.0901 & 7.9731 & 3.2859 \\
\bottomrule
\end{tabular}
\caption{Person-years in state variables for all control combinations for Strategy C}
\label{03:tab:appendix:person_years_strategy_C}
\end{table}

\begin{table}[htbp]
\centering
\small
\begin{tabular}{@{}l@{\hspace{5pt}}c@{\hspace{9pt}}c@{\hspace{9pt}}c@{\hspace{9pt}}c@{\hspace{9pt}}c@{\hspace{9pt}}c@{\hspace{9pt}}c@{\hspace{9pt}}c@{}}
\toprule
\textbf{Control Combination} & \textbf{S} & \textbf{I$_{SU}$} & \textbf{I$_{SD}$} & \textbf{T$_1$} & \textbf{I$_{RU}$} & \textbf{I$_{RD}$} & \textbf{T$_2$} & \textbf{A} \\
\midrule
$(0,0,0,0,0)$ & 21806.0240 & 2.7202 & 4.3542 & 1.7988 & 19.3376 & 0.8708 & 3.2637 & 5.9055 \\
$({u_1,0,0,0,0})$ & 21806.8647 & 0.5731 & 6.0694 & 1.7988 & 19.3377 & 0.8708 & 3.2637 & 5.8290 \\
$({0,u_2,0,0,0})$ & 21809.3619 & 2.7204 & 4.3542 & 1.7988 & 0.2398 & 18.3332 & 3.2637 & 5.6092 \\
$({0,0,u_3,0,0})$ & 21806.0238 & 2.7199 & 4.3506 & 1.7993 & 19.3414 & 0.8708 & 3.2637 & 5.9055 \\
$({0,0,0,u_4,0})$ & 21806.0778 & 2.7202 & 4.3542 & 1.7988 & 19.3117 & 0.0442 & 4.8138 & 5.8268 \\
$({0,0,0,0,u_5})$ & 21809.9810 & 2.7204 & 4.3542 & 29.3701 & 1.9907 & 0.8708 & 3.2637 & 4.2292 \\
$({u_1, u_2, 0,0,0})$ & 21810.2029 & 0.5731 & 6.0695 & 1.7988 & 0.2398 & 18.3332 & 3.2637 & 5.5327 \\
$({u_1,0, u_3,0,0}$ & 21806.8647 & 0.5731 & 6.0691 & 1.7990 & 19.3377 & 0.8708 & 3.2637 & 5.8290 \\
$({u_1,0,0, u_4,0})$ & 21806.9186 & 0.5731 & 6.0694 & 1.7988 & 19.3118 & 0.0442 & 4.8138 & 5.7503 \\
$({u_1,0,0,0, u_5})$ & 21810.8220 & 0.5731 & 6.0695 & 29.3701 & 1.9907 & 0.8708 & 3.2637 & 4.1526 \\
$({0,u_2, u_3,0,0})$ & 21809.3619 & 2.7204 & 4.3541 & 1.7989 & 0.2398 & 18.3332 & 3.2637 & 5.6092 \\
$({0, u_2,0, u_4,0})$ & 21810.4071 & 2.7204 & 4.3542 & 1.7988 & 0.1803 & 0.5459 & 35.3498 & 3.9411 \\
$({0,u_2,0,0, u_5})$ & 21810.3301 & 2.7204 & 4.3542 & 29.3698 & 0.0342 & 2.6602 & 3.2637 & 4.1986 \\
$({0,0,u_3, u_4,0})$ & 21806.0776 & 2.7200 & 4.3506 & 1.7993 & 19.3155 & 0.0442 & 4.8138 & 5.8267 \\
$({0,0,u_3,0, u_5})$ & 21810.7240 & 2.3556 & 0.1365 & 37.3250 & 1.9900 & 0.8708 & 3.2637 & 3.7848 \\
$({0,0,0,u_4, u_5})$ & 21810.0348 & 2.7204 & 4.3542 & 29.3702 & 1.9648 & 0.0442 & 4.8138 & 4.1504 \\
$({u_1, u_2, u_3,0,0})$ & 21810.2029 & 0.5731 & 6.0692 & 1.7991 & 0.2398 & 18.3332 & 3.2637 & 5.5327 \\
$({u_1, u_2,0, u_4,0})$ & 21811.2482 & 0.5731 & 6.0695 & 1.7988 & 0.1803 & 0.5459 & 35.3498 & 3.8646 \\
$({u_1, u_2,0,0, u_5})$ & 21811.1712 & 0.5731 & 6.0695 & 29.3698 & 0.0342 & 2.6602 & 3.2637 & 4.1221 \\
$({u_1,0, u_3, u_4,0})$ & 21806.9186 & 0.5731 & 6.0691 & 1.7990 & 19.3118 & 0.0442 & 4.8138 & 5.7503 \\
$({u_1,0, u_3,0, u_5})$ & 21811.7916 & 0.2443 & 0.2452 & 40.0407 & 1.9898 & 0.8708 & 3.2637 & 3.5544 \\
$({u_1,0,0, u_4, u_5})$ & 21810.8759 & 0.5731 & 6.0695 & 29.3701 & 1.9648 & 0.0442 & 4.8138 & 4.0739 \\
$({0,u_2, u_3, u_4,0})$ & 21811.1308 & 2.3573 & 0.1562 & 2.2836 & 0.2135 & 0.6366 & 42.8143 & 3.4896 \\
$({0,u_2, u_3,0, u_5})$ & 21811.0731 & 2.3557 & 0.1365 & 37.3249 & 0.0334 & 2.6601 & 3.2637 & 3.7542 \\
$({0,u_2,0, u_4, u_5})$ & 21810.4071 & 2.7204 & 4.3542 & 1.7988 & 0.1803 & 0.5459 & 35.3497 & 3.9411 \\
$({0,0,u_3, u_4, u_5})$ & 21810.7779 & 2.3556 & 0.1365 & 37.3251 & 1.9641 & 0.0442 & 4.8138 & 3.7060 \\
$({u_1, u_2, u_3, u_4,0})$ & 21812.1756 & 0.2697 & 0.3075 & 2.4480 & 0.2278 & 0.6773 & 45.2534 & 3.2645 \\
$({u_1, u_2, u_3,0, u_5})$ & 21812.1408 & 0.2443 & 0.2452 & 40.0406 & 0.0332 & 2.6600 & 3.2637 & 3.5239 \\
$({u_1, u_2,0, u_4, u_5})$ & 21811.2482 & 0.5731 & 6.0695 & 1.7988 & 0.1803 & 0.5459 & 35.3497 & 3.8646 \\
$({u_1,0, u_3, u_4, u_5})$ & 21811.8455 & 0.2443 & 0.2452 & 40.0407 & 1.9639 & 0.0442 & 4.8138 & 3.4757 \\
$({0, u_2, u_3, u_4, u_5})$ & 21811.1414 & 2.3573 & 0.1554 & 10.3456 & 0.1972 & 0.6198 & 34.6006 & 3.4964 \\
$({u_1, u_2, u_3, u_4, u_5})$ & 21812.1960 & 0.2646 & 0.2793 & 12.4787 & 0.2035 & 0.6374 & 35.1276 & 3.2672 \\
\bottomrule
\end{tabular}
\caption{Person-years in state variables for all control combinations for Strategy D}
\label{03:tab:appendix:person_years_strategy_D}
\end{table}

\section*{Acknowledgment}
\label{2_Section_Acknowledgment}

The research work of Ashish Poonia was supported by the Council of Scientific and Industrial Research (CSIR), India (File Number 09/731(0175)/2019-EMR-1).

\section*{Statements and Declarations}
The authors have no conflicts of interest to declare that are relevant to the content of this article.

\section*{Data Availability Statement}

The data will be made available on reasonable request.

\bibliographystyle{unsrt}

\bibliography{bibfile_03}

\end{document}